\documentclass{CSML}

\usepackage{lastpage}
\newcommand{\dOi}{13(4:3)2017}
\lmcsheading{}{1--\pageref{LastPage}}{}{}%
{Jun.~20, 2007}{Oct.~26, 2017}{}

\usepackage{hyperref}
\hypersetup{hidelinks}
\usepackage{soul}
\usepackage{ae}
\usepackage{amssymb}
\usepackage{amsmath}
\usepackage{logiclmcs}
\usepackage{latexsym}
\usepackage{pstricks}
\usepackage{pst-node}
\usepackage{./gastex}

\newcommand{\pdflatex}[1]{#1}
\newcommand{\concat}{\cdot}
\newcommand{\acomp}{\succeq}

\newcommand{\bouns}{\mathbb B}
\newcommand{\ubouns}{\mathbb U}
\newcommand{\qabove}{\mathbb A}

\newcommand{\ign}[1]{}
\newcommand{\zj}{\set{0,1}}

\newcommand{\intro}[1]	{\emph{#1}}
\newcommand{\itl}	{\Join}
\newcommand{\nats}	{{\mathbb{N}}}
\newcommand{\e}		{\varepsilon}

\newcommand{\inc}{\mathrm i}
\newcommand{\reset}{\mathrm r}

\newcommand{\ralph}[1] {{[a:{#1}]}}

\newcommand{\ombs}	{\ensuremath{\omega BS}}

\newcommand{\comp}	{\preceq}
\newcommand{\notcomp}	{\succ}

\newcommand{\fepsilon} {\bar\varepsilon}

\newcommand{\val}	{\mathit{val}}
\newcommand{\lab}	{\mathcal L}

\definecolor{bleuclair}{rgb}{0.7, 0.7, 1.0}
\definecolor{jaune}{rgb}{1.0, 1.0, 0.0}
\definecolor{white}{rgb}{1.0,1.0,1.0}
\sethlcolor{red}
\newcommand{\oldhl}[1]{#1}
\newcommand{\ttc}[1]{#1}
\newcommand{\hhl}[1]{#1}
\begin{document}

\title{Boundedness in languages of infinite words}

\author{Miko{\l}aj Boja{\'n}czyk}	
\address{Warsaw University}	
\email{bojan@mimuw.edu.pl}  

\author{Thomas Colcombet}	
\address{Cnrs/Liafa/Universit\'e Paris Diderot, Paris 7}	
\email{thomas.colcombet@liafa.jussieu.fr}  

\thanks{Work supported by the EU-TNR
  network GAMES. The first author is also supported by Polish government
  grant no. N206 008 32/0810.}

\begin{abstract}
  We define a new class of languages of $\omega$-words, strictly
  extending $\omega$-regular languages.

  One way to present this new class is by a type of regular
  expressions. The new expressions are an extension of
$\omega$-regular expressions where two new variants of the Kleene star
$L^*$ are added: $L^B$ and $L^S$.  These new exponents are used to say
that parts of the input word have bounded size, and that parts of the
input can have arbitrarily large sizes, respectively.  For instance,
the expression $(a^Bb)^\omega$ represents the language of infinite
words over the letters~$a,b$ where there is a common bound on the
number of consecutive letters~$a$.  The expression~$(a^Sb)^\omega$
represents a similar language, but this time the distance between
consecutive $b$'s is required to tend toward the infinite.

We develop a theory for these languages, with a focus on decidability
and closure. We define an equivalent automaton model, extending
B\"uchi automata. The main technical result is a complementation lemma
that works for languages where only one type of exponent---either
$L^B$ or $L^S$---is used.

We use the closure and decidability results to obtain partial
decidability results for the logic MSOLB, a logic obtained by
extending monadic second-order logic with new quantifiers that
speak about the size of sets.
\end{abstract}

\maketitle

\section{Introduction}
In this paper we introduce a new class of languages of infinite words.
The new languages of this kind---called $\omega BS$-regular
languages---are defined using an extended form of $\omega$-regular
expressions.  The extended expressions can define properties such as
``words of the form $(a^*b)^\omega$ for which
  there is an upper bound on the number of consecutive letters $a$''.
Note that this bound depends upon the word, and for this reason the
language is not $\omega$-regular.  This witnesses that $\omega
BS$-regular languages are a proper extension of $\omega$-regular
languages.

The expressions for $\omega BS$-regular languages are obtained by
extending the usual $\omega$-regular expressions with two new variants
of the Kleene star $L^*$. These are called the bounded exponent $L^B$
and the strongly unbounded exponent $L^S$. The idea behind $B$ is that
the language $L$ in the expression $L^B$ must be iterated a bounded
number of times, the bound being fixed for the whole word.  For
instance, the language given in the first paragraph is described by
the expression $(a^Bb)^\omega$. The idea behind $S$ is that the number of
iterated concatenations of the language $L$ must tend toward infinity (i.e., have
no bounded subsequence).

This paper is devoted to developing a theory for the new languages.
There are different motivations for such a study. The first one is the
interest of the model itself: it extends naturally the standard
framework of $\omega$-regular languages, while retaining some closure
and decidability properties. We also show that, \oldhl{just as
  $\omega$-regular expressions define the same languages of infinite
  words as} \ttc{the ones definable in} \oldhl{monadic second-order logic, the class of~$\omega
  BS$-regular languages also admits a logical counterpart}.  The
relevance of the model is also \oldhl{quite} natural: the use of bounding
arguments in proofs is very common, and the family of~$\omega
BS$-regular languages provides a unified framework for developing such
arguments.  \oldhl{A notable example is the famous star-height
  problem, to which we will return later in the introduction.} \oldhl{Another application of our results,
which is presented in this paper,  is an algorithm deciding} if an $\omega$-automatic graph has
bounded out-degree.  We believe that more problems are related to this
notion of regularity with bounds.  In this paper, we concentrate on a
basic theoretical study of the model.

The first step in this study is \ttc{the introduction of} a
new family of automata over infinite words, \oldhl{extending B\"uchi
automata, which we call
bounding automata}.
Bounding automata have the same expressive power as 
\oldhl{$\omega BS$-regular expressions}. \oldhl{(However, the
translations between bounding automata and $\omega BS$-regular
expressions are more involved than in the case of }\ttc{$\omega$-}\oldhl{regular languages.)}
A bounding automaton is a finite
automaton equipped with a \ttc{finite} number of counters.
\hhl{During a run,} these counters can be incremented
and reset, but not read.  The counter values are used in the
acceptance condition, which depends on their asymptotic values
(whether counter values are bounded or tend toward infinity).
Thanks to \ttc{the equivalence between automata and expressions, and using} simple
automata constructions, we obtain \ttc{the} closure of $\omega BS$-regular
languages under union, intersection and projection.

Unfortunately, $\omega BS$-regular automata cannot be determinized.
Even more problematic, $\omega BS$-regular languages are not closed under
complement.
However, we are able to identify two fragments of
$\omega BS$-regular languages that complement each other. We show that
the complement of a language that only talks about bounded sequences
is a language that only talks about sequences tending toward infinity;
and vice versa.  The difficult proof of this complementation result is
the technical core of the paper.

Finally, we present a logic that \oldhl{captures exactly the  $\omega BS$-regular
languages, i.e.~is equivalent to both the $\omega BS$-regular
expressions and automata.}
As \ttc{it} is well known, languages defined by $\omega$-regular
expressions are exactly the ones definable in monadic second-order
logic. What extension of this logic corresponds to $\omega BS$-regular
expressions? Our approach is to add a new quantifier, called the
existential unbounding quantifier $\ubouns$.
A formula $\ubouns X.\phi(X)$ is true
if it is possible to find  sets satisfying $\phi(X)$ of
arbitrarily large size.
Every $\omega BS$-regular language can be defined in monadic second-order logic extended with $\ubouns$.
However, due to the failure of the closure under complementation, the
converse does not hold.
By restricting the quantification patterns,
we identify  fragments of the logic that correspond to the various types of
$\omega BS$-regular expressions introduced in this paper.

\medskip
\noindent{\bf Related work.}
This work tries to continue the long lasting tradition of
logic/automata correspondences initiated by
  B\"uchi~\cite{buchi60,buchi62} and continued by
  Rabin~\cite{rabin69}, only to mention the most famous names
  (see~\cite{thomas} for a survey).
We believe that bounding
properties extend the \ttc{standard} notion of regularity \ttc{in a meaningful way,}
and that languages defined by our extended expressions have every right to be called
regular, even though they are not captured by B\"uchi automata. For
instance, every $\omega BS$-regular language $L$ has a finite number
of quotients~$w^{-1}L$ and $Lw^{-1}$. (Moreover, the right quotients
$Lw^{-1}$ are regular languages of finite words.)  \hhl{Unfortunately, }we do not have a
full complementation result.

The quantifier $\ubouns$ in the logic that describes $\omega
BS$-regular languages was already introduced in \cite{bojanboun}.
More precisely, the quantifier studied in \cite{bojanboun} is~$\bouns$:
a formula~$\bouns X.\phi$ \ttc{expresses that}  there is  a bound on
the size of sets~$X$ satisfying property~$\phi$.
This formula is semantically equivalent to~$\neg(\ubouns X.\phi)$.
Although \cite{bojanboun} went beyond words and considered infinite
trees, the proposed satisfiability algorithm  worked for a very restricted
fragment of the logic with no (not even partial) complementation.
Furthermore, no notion of automata or regular expression was proposed.

Boundedness properties have been considered in model-checking. For
instance, \cite{boundedstack} considered systems
described by pushdown automata whose stack size is unbounded.

Our work on bounds can also be related to cardinality restrictions.
In~\cite{parikhaut}, Klaedtke and Ruess considered an extension of
monadic second-order logic, which allowed  cardinality constraints of the form
\begin{equation*}
  |X_1| + \cdots + |X_n| \le |Y_1| + \cdots + |Y_m| \ .
\end{equation*}
In general, such cardinality constraints (even $|X| = |Y|$) lead to
undecidability of the satisfiability problem.
Even though cardinality constraints can express all
$\omega BS$-regular languages, the decidable fragments considered in
\cite{parikhaut} are insufficient for our purposes: those fragments
\oldhl{capture only a small fragment of}   $\omega BS$-regular languages.

Finally, the objects we manipulate are related to the (restricted)
star-height problem. This problem  is,
given a natural number~$k$ and a regular language of finite
words~$L$,  to decide if~$L$ can be defined by a regular expression  where the
nesting depth
 of Kleene-stars is at most~$k$. This problem,  first raised by Eggan
 in 1963 \cite{eggan63}, has a central
role in the theory of regular languages. It has been first shown decidable
by Hashiguchi \cite{hashiguchi88} via a  difficult \hhl{proof. The
correctness of this proof  is still unclear.}
A new proof, due to Kirsten \cite{kirsten05},  reduces the star-height problem to the limitedness problem 
for nested distance desert automata. The latter are finite state
automata, which assign a natural number to each word.
The limitedness problem is the question whether there is a bound on
the numbers produced by a given automaton.
Nested distance desert automata---we were not aware of their existence
when developing the theory of $\omega BS$-regular languages---happen
to be syntactically equivalent to the hierarchical $B$-automata that we use
as an intermediate object in the present work. The semantics of the
two models are also tightly connected, and it is possible to derive from
the result of our paper the decidability of the limitedness problem
for nested distance desert automata (though using our more general results,
we obtain a non-elementary complexity, which is far beyond the optimal PSPACE
algorithm of Kirsten).

\medskip
\noindent{\bf Structure of the paper.} In
Section~\ref{section:regular-expressions}, we formally define the
$\omega BS$-regular expressions that are the subject of this paper. We
introduce two restricted types of expressions (where the $B$ and $S$
exponents are prohibited, respectively) and \hhl{give an overview of} the closure
properties of the respective expressions. In
Section~\ref{section:automata}, we introduce our automata models and
show that they are equivalent to the \hhl{ $\omega BS$-regular} expressions.  In
Section~\ref{section:msob}, we show how our results can be applied to
obtain a decision procedure for satisfiability in an extension of
monadic second-order logic.  The main technical result, which concerns
closure under complementation, \hhl{is given at} the end of the paper, in
Section~\ref{section:complementation}.

\medskip
We would like to thank the anonymous referees, who contributed enormously to this paper, through their many thoughtful and helpful comments.


\section{\oldhl{Regular Expressions with Bounds}}
\label{section:regular-expressions}
In this section we define $\omega BS$-regular expressions.  The
expressions are the first of three means of defining languages
with bounds. The other two---automata and logic---are presented in the
next sections.

\subsection{Definition}

In a nutshell, to the standard operations used in~$\omega$-regular
expressions, we add two variants of the Kleene star $*$: the $B$ and
$S$ exponents.
These are used to constrain the number of iterations,
or more precisely the asymptotic behavior of the number of iterations
(this makes sense since the new exponents are used in the context of
an $\omega$ exponent.
When the $B$ exponent is used, the number of iterations has to be
bounded by a bound which depends on the word.
When the $S$ exponent is used, it has to tend toward
infinity.  For instance, the expression~$(a^Bb)^\omega$ represents the
words in $(a^*b)^\omega$ where the size of sequences of
consecutive~$a$'s is bounded.
Similarly, the
expression~$(a^Sb)^\omega$ requires the size of \oldhl{(maximal)} sequences of
consecutive~$a$'s to tend toward infinity.  These new expressions are
called~$\omega BS$-regular expressions. \oldhl{A more detailed definition is
presented below.}

Just as an
$\omega$-regular expression uses regular expressions of finite words
as building blocks, for $\omega BS$-regular expressions one also
needs first to define a finitary variant, called $BS$-regular
expressions.
Subsection~\ref{subsubsection:BSexpressions} presents
this \oldhl{finitary variant}, while
Subsection~\ref{subsubsection:omegaBSexpressions} introduces $\omega
BS$-regular expressions.

\subsubsection{$BS$-regular expressions.}
\label{subsubsection:BSexpressions}
In the following we will write that an infinite sequence of natural numbers~$g\in\nats^\omega$ is
\intro{bounded} if there exists a global bound on the~$g(i)$'s, i.e.,
if~$\limsup_ig(i)<+\infty$. This behavior is denoted by the letter~$B$.
The infinite sequence is \intro{strongly unbounded} if it tends toward infinity,
i.e., $\liminf_i g(i)=+\infty$.
This behavior is denoted by the letter~$S$.
Let us remark that an infinite sequence is bounded iff it has no strongly unbounded infinite subsequence,
and that an infinite sequence is strongly unbounded iff it has no bounded infinite subsequence.

A \intro{natural number sequence}~$g$ is a finite or infinite sequence
of natural numbers that we write in a functional way:
$g(0),g(1),\dots$ Its length is~$|g|$, \oldhl{which may possibly be
  $\infty$}. We denote by~$\nats^\infty$ the set of sequences of
natural numbers, both finite and infinite.  A sequence of natural
numbers~$g$ is \intro{non-decreasing} if $g(i)\leq g(j)$ holds for
all~$i\leq j<|g|$.  \oldhl{We write~$g\leq n$ to say that~$g(i)\leq n$ holds} for
all~$i<|g|$.  For~$g$ a non-decreasing sequence of natural numbers,
\oldhl{we define~$g'$ to be $g'(i)=g(i+1)-g(i)$ for all~$i$}
\ttc{such that~$i+1<|g|$}.  The sequence $g$ is of \intro{bounded
  difference} if $g'$ is either finite or bounded; it is of
\intro{strongly unbounded difference} if the sequence $g'$ is either
finite or strongly unbounded.

A \intro{word sequence} $\vec u$ over \oldhl{an} alphabet~$\Sigma$
is a finite or infinite sequence of finite words over~$\Sigma$,
i.e., an element of~$(\Sigma^*)^\infty$.
The components of the word sequence~$\vec u$ are the finite words $u_0,u_1,\dots$;
we also write~$\vec u=\langle u_0,u_1,\dots\rangle$.
We denote by~$\varepsilon$ the finite word of length~$0$, which is
different from the word sequence of length~$0$ denoted~$\langle\rangle$.
We denote by~$|\vec u|$ the length of the word sequence~$\vec u$.

A \intro{language of word sequences} is a set of word
sequences. \oldhl{The finitary variant of $\omega BS$-regular
  expressions will describe
  languages of word sequences. Note that the finiteness concerns only
  the individual words in the sequences, while the sequences
  themselves will sometimes be infinitely long. }

\oldhl{We define the following operations, which take as parameters languages
of word sequences~$K,L$.}

\noindent
The \intro{concatenation} of~$K$ and~$L$ is defined by
  \begin{equation*}
    K \concat L = \{\langle u_0v_0,u_1v_1\dots\rangle~:~\vec u\in K,
	~\vec v\in L,~|\vec u|=|\vec v|\}\ .
  \end{equation*}

\noindent
The \intro{mix} of~$K$ and~$L$ (which is \emph{not} the union) is
defined by
  \begin{equation*}
    K+L = \set{\vec w~:~\vec u\in K,~\vec v\in L,~\forall i <
      \max(|\vec w|,|\vec u|). w_i\in\{u_i\}\cup\{v_i\}}\ ,
  \end{equation*}
  with the convention that~$\set{u_i}$ (\oldhl{respectively,}
  $\set{v_i}$) is~$\emptyset$ if~$i$\ttc{$\geq$}$|\vec u|$ (\oldhl{respectively,} if $i$\ttc{$\geq$}$|\vec v|$).

\noindent
The \intro{$*$ exponent} of~$L$ is defined by grouping words into blocks:
    \begin{multline*}
          L^* = \{\langle u_0\dots u_{g(0)-1},u_{g(0)}\dots u_{g(1)-1},\dots\rangle~:~
	\vec u\in L,~g\in\nats^\infty,~g\leq|u|,~g~\text{non-decreasing}\}\ .
    \end{multline*}

\noindent
The \intro{bounded exponent $L^B$} is defined similarly:
    \begin{multline*}
          L^B = \{\langle u_0\dots u_{g(0)-1},u_{g(0)}\dots u_{g(1)-1},\dots\rangle~: \\
	\vec u\in L,~g\in\nats^\infty,~g\leq|u|,~g~\text{non-decreasing of bounded difference}\}\ .
    \end{multline*}

\noindent
And the \intro{\hhl{strongly} unbounded exponent $L^S$} is defined by:
    \begin{multline*}
          L^S = \{\langle u_0\dots u_{g(0)-1},u_{g(0)}\dots u_{g(1)-1},\dots\rangle~: \\
	\vec u\in L,~g\in\nats^\infty,~g\leq|u|,~g~\text{non-decreasing
          of  {strongly} unbounded difference}\}\ .
    \end{multline*}

\newcommand{\sem}[1]	{[\![#1]\!]}

We now define $BS$-regular expressions using the above operators.
\begin{defi}
A \intro{$BS$-regular expression} has the following syntax ($a$ being
some letter of the given finite alphabet~$\Sigma$):\
\begin{align*}
e &=~~\emptyset~~|~~\varepsilon~~|~~\fepsilon~~|~~a~~|~~e\concat e~~|~~e+e~~|~~e^*~~|~~e^B~~|~~e^S\ .
\end{align*}
As usual, we often omit the $\concat$ symbols in expressions.

The semantic~$\sem{e}$ of a~$BS$-regular expression~$e$ is the language of word sequences defined inductively by the
following rules:
\begin{itemize}
\item $\sem{\emptyset}=\{\langle\rangle\}$.
\item $\sem{\fepsilon}=\{\langle\rangle,\langle\varepsilon\rangle,\langle\varepsilon,\varepsilon\rangle,\dots\}$,
	i.e., all finite word sequences built with~$\varepsilon$.
\item $\sem{\varepsilon}=\sem{\fepsilon}\cup\{\langle\varepsilon,\varepsilon,\dots\rangle\}$,
	i.e., all word sequences built with~$\varepsilon$.
\item $\sem a=\{\langle\rangle,\langle a\rangle,\langle a,a\rangle,\dots\}\cup\{\langle a,a,\dots\rangle\}$,
	i.e.,  all word sequences built with~$a$.
\item $\sem{e\concat f}=\sem e\concat \sem f$,
	$\sem{e+f}=\sem e+\sem f$, $\sem{e^*}={\sem e}^*$,
	$\sem{e^B}={\sem e}^B$, and $\sem{e^S}={\sem e}^S$.
\end{itemize}
A language of word sequences is called \intro{$BS$-regular}
if it is obtained by evaluating \hhl{a $BS$-regular expression.}
The \intro{$B$-regular} \oldhl{(respectively, \emph{$S$-regular})}
languages correspond to the particular case where the expression
does not use the exponent~$S$ \oldhl{(respectively, $B$)}.
\end{defi}

\begin{exa}
\hhl{The $B$-regular  expression~$a^B$
represents the finite or infinite  word sequences  that consist of  words from $a^*$ where the number of $a$'s
is bounded:}
\begin{equation*}
  \sem{a^B} = \set{ \langle a^{f(0)},a^{f(1)},\ldots\rangle~:~~f\in\nats^\infty~\text{is bounded}}\ .
\end{equation*}
The \oldhl{language of  word sequences} $\sem{a^B \concat (b \concat a^B)^S}$
consists of word sequences where the
number of consecutive $a$'s is bounded, while the number of~$b$'s in
each word of the word sequence is strongly unbounded.
\end{exa}

Except for the two extra exponents $B$ and $S$ and the constant $\fepsilon$,
these expressions coincide syntactically with the standard regular expressions. 
\oldhl{Therefore, one may ask, how do our expressions correspond to
  standard regular expressions on their common syntactic fragment, which includes
  the Kleene star $*$, concatenation $\cdot$ and union $+$? The new expressions are
  a 
conservative extension of standard regular expressions in the following sense.} If one takes a standard regular
expression~$e$ defining a language of finite words~$L$ and evaluates it as
a $BS$-regular expression, the resulting language of word sequences is
\oldhl{the set of word sequences where every component belongs to $L$, i.e.~}
$\sem e = \{\vec u~:~\forall i$\ttc{$<$}$|u|.\,u_i\in L\}$.

In the fact below we present two important
  closure properties of languages defined by $BS$-regular
  expressions.
 We will often use this  fact, sometimes without
explicitly invoking it.

\begin{fact}\label{fact:basic-BS-closure}
  For every $BS$-regular language~$L$, the following \hhl{properties} hold:
  \begin{enumerate}
  \item $L+L = L$,
  \item $L = \{\langle u_{f(0)},u_{f(1)},\dots\rangle~:~\vec u\in L,~f\in\nats^\infty,~f\leq|u|,
	\text{\ttc{$f$ strongly unbounded}}
\}$.
	\label{item:basic-subsequence}
  \end{enumerate}
\end{fact}
\begin{proof}
\oldhl{A straightforward structural induction.}
\end{proof}
\oldhl{Item}~\ref{item:basic-subsequence} implies that \oldhl{$BS$-regular}
languages are closed under taking subsequences. \oldhl{That is, every $BS$-regular language of word sequences  is closed under removing, possibly
infinitely many, component words  from the sequence.  In particular every
$BS$-regular language is 
closed under taking the prefixes of word sequences.}

\subsubsection{$\omega BS$-regular expressions.}
\label{subsubsection:omegaBSexpressions}
We are now ready to introduce the $\omega BS$-regular expressions.
These describe languages of~$\omega$-words.  From a word sequence we
can construct an $\omega$-word by concatenating all the words in the
word sequence:
\begin{align*}
\langle u_0,u_1,\dots\rangle^\omega &= u_0u_1\dots
\end{align*}
This operation---called the \intro{$\omega$-power}---is only defined if
  the word sequence has nonempty words on infinitely many coordinates.
  The $\omega$-power is naturally extended to languages of word
  sequences by taking the $\omega$-power of every word sequence in the
  language (where it is defined).
\begin{defi}
  An~\intro{$\omega BS$-regular expression} is an expression of the form:
  $$
  e=\sum_{i=1}^n e_i\concat f_i^\omega
  $$
  in which each~$e_i$ is a regular expression, and each~$f_i$
  is a~$BS$-regular expression.
  The expression is \oldhl{called}~\intro{$\omega B$-regular}
  (\oldhl{respectively,} $\omega S$-regular) if all the expressions~$f_i$
  are~$B$-regular (\oldhl{respectively,} $S$-regular).

  The semantic interpretation~$\sem e$ is the language of $\omega$-words
  $$
  \bigcup\limits_{i=1\dots n} \sem{e_i}_F\concat\sem{f_i}^\omega,
  $$
  in which~\ttc{$\sem{\cdot}_F$ denotes the standard semantic of regular expressions},
  and $\concat$ denotes the concatenation of a language
  of finite words with a language of~$\omega$-words.  \oldhl{Following a
  similar tradition for regular expressions, we will often identify an
  expression with its semantic interpretation, writing, for instance,
  $w \in (a^Bb)^\omega$ instead of $w \in
  \sem{(a^Bb)^\omega}$}.
  
\oldhl{A language of~$\omega$-words is called \emph{$\omega BS$-regular}
  (respectively, $\omega B$-regular, $\omega S$-regular) if it is
  equal to~$\sem e$ for some} \ttc{$\omega BS$-regular expression $e$
  (respectively, $\omega B$-regular expression $e$,  $\omega S$-regular expression $e$).}
\end{defi}
In \hhl{$\omega BS$-regular expressions}, the language of finite words can be $\set
\varepsilon$; to avoid clutter we omit the language of finite words
in such a situation, e.g., we write $a^\omega$ 
for~$\set{\varepsilon}\concat a^\omega$.

This definition differs from the definition of~$\omega$-regular
expressions only in that the~$\omega$ exponent is applied to $BS$-regular
languages of word sequences instead of regular word languages.  As one may
expect, the standard class of~$\omega$-regular languages corresponds
to the case of \ombs-regular languages where neither~$B$ nor~$S$ is
used (the presence of~$\fepsilon$ does not make any difference here).

\begin{exa}
The expression $(a^Bb)^\omega$ defines the
language of $\omega$-words containing an infinite number of~$b$'s
where the possible number of consecutive~$a$'s is bounded.  The language $(a^Sb)^\omega$ corresponds to the case
where the lengths of maximal consecutive sequences of $a$'s tend
toward infinity.  The language $(a+b)^*a^\omega +
((a^*b)^*a^Sb)^\omega$ is more involved. It corresponds to
the language of words where either there are finitely many $b$'s
(left argument of the sum), or
the number of consecutive $a$'s is unbounded but not necessarily
strongly unbounded (right argument of the sum).
This is the complement of the language $(a^Bb)^\omega$.
\end{exa}

\begin{fact}\label{fact:emptiness-decidable}
  Emptiness is decidable for \ombs-regular  languages.
\end{fact}
\begin{proof}
  An \ombs-regular language is nonempty if and only if one of the
  languages $M \cdot L^\omega$ in the union defining the language is such that
  the regular language $M$ is nonempty and the $BS$-regular language $L$ contains
  a word sequence with infinitely many nonempty words. Therefore  the
  problem boils down to checking if a $BS$-regular language contains a
  word sequence with infinitely many nonempty \oldhl{words}. Let $\mathcal{N}$ be the
  set of $BS$-regular languages with this property, and let $\mathcal{I}$ be the
  set of $BS$-regular languages that contain at least one infinite
  word sequence (possibly with finitely many nonempty words). The
  following rules determine which languages belong to $\mathcal I$ and
  $\mathcal N$.
  \begin{itemize}
  \item $K+L \in \mathcal I$  iff $K \in \mathcal I$ or $L\in \mathcal I$.
  \item $K\concat L \in \mathcal I$ iff $K \in \mathcal I$ and $L\in \mathcal I$.
  \item $L^*$ and $L^B$ always belong to $\mathcal I$.
  \item $L^S  \in \mathcal I$  iff $L \in \mathcal I$.

  \item $K+L  \in \mathcal N$ iff $K \in \mathcal N$ or $L\in \mathcal N$.
  \item $K\concat L \in \mathcal N$ iff either $K \in \mathcal I$ and
    $L \in \mathcal N$, or $K \in \mathcal N$ and $L \in \mathcal I$.
  \item $L^*$, $L^B$ and $L^S$ belong to $ \mathcal N$  iff $L \in \mathcal N$.
  \qedhere
  \end{itemize}
 \end{proof}


 The constant~$\fepsilon$ will turn out to be a convenient technical
 device.  In practice, $\fepsilon\concat L$ restricts a language of
 \oldhl{word} sequences~$L$ to its finite sequences. We denote by $\bar
 L$ the language \ttc{of word sequences} $\fepsilon\concat L$.  \oldhl{This} \ttc{construction} \oldhl{will be used}
   \ttc{for instance} \oldhl{when dealing with intersections:} \ttc{as an example, the equality} \oldhl{$L^B\cap L^S={\bar
     L}^*$ holds when~$L$ does not contain~$\e$.}  \hhl{It turns out
   that  $\fepsilon$ is just syntactic sugar, as stated by:}
\begin{prop}\label{proposition:fplus-exp}
  Every $\omega BS$-regular expression (\oldhl{respectively, } $\omega B$-regular and
  $\omega S$-regular ones) is equivalent to one without the~$\fepsilon$
  constant.
\end{prop}
Note that this
proposition does not mean that $\fepsilon$ can be eliminated from
$BS$-regular expressions\ttc{. It can only} be eliminated\ttc{ }%
under the scope of the $\omega$-power in $\omega BS$-regular expressions.
For instance, $\fepsilon$ is necessary to capture the  $BS$-regular language~$\bar a$.
However, \hhl{once
   under the $\omega$ exponent, $\fepsilon$ becomes redundant; for
   instance } $(\bar a)^\omega$ is equivalent to $\emptyset$, while  $(\bar a+ b)^\omega$
is equivalent to $(a+b)^*b^\omega$.

Proposition~\ref{proposition:fplus-exp} will follow immediately once
the following lemma is established:
\begin{lem}\label{lemma:finite-mix-top}
\hhl{Let $T \in \set{BS, B, S}$. 
  For every $T$-regular  expression, there is an equivalent one of the form~$\bar M+ L$
  where~$M$ is a regular expression and~$L$ is $T$-regular and does not use
 $\fepsilon$.}
\end{lem}
\begin{proof}
  By structural induction, we prove that for  every~$BS$-regular
  expression~$L$ \oldhl{can be equivalently expressed as}
  \begin{align*}
  L & = \bar M + K\ ,
  \end{align*}
\oldhl{where} $M$ is obtained from $L$ by replacing exponents $B$ and~$S$
  by Kleene stars and replacing $\fepsilon$ by $\varepsilon$,
  and $K$ is obtained from~$L$ by replacing~$\fepsilon$
  by~$\emptyset$.

  \hhl{Note that  Fact~\ref{fact:basic-BS-closure} is used  implicitly
    above.}
\end{proof}


\subsection{\oldhl{Summary: The Diamond}}
\label{subsection:diamond}
In this section we present  Figure~\ref{figure:diamond}, which
summarizes the technical contributions of this paper. We call this
figure \intro{the diamond}.  Though not all the material necessary to
understand this figure has been provided yet, we give it here as a
reference and guide to what follows.

\pdflatex
{
\begin{figure*}[ht]
\begin{center}
\begin{tabular}{ccc}
\\ \\ \\
	&\rnode{oBS}{\begin{tabular}{c}
	\ombs-regular expressions	\\
	hierarchical \ombs-automata	\\
	\ombs-automata
	\end{tabular}}					&\\
		\\ \\ \\
\rnode{oS}{\begin{tabular}{c}
$\omega S$-regular expressions		\\
hierarchical $\omega S$-automata	\\
$\omega S$-automata
\end{tabular}}\hspace{-1.2cm}
	&	&\hspace{-1.2cm}\rnode{oB}{\begin{tabular}{c}
		$\omega B$-regular expressions		\\
		hierarchical $\omega B$-automata	\\
		$\omega B$-automata
		\end{tabular}}				\\
		\\ \\ \\
	&\rnode{o}{\begin{tabular}{c}
	$\omega$-regular expressions\\
	B\"uchi automata
	\end{tabular}}					&\\
\\ \\
\end{tabular}
{\psset{linewidth=0}
\ncline{-}{o}{oB}\lput*{:U}{$\subsetneq$}
\ncline{-}{oS}{o}\lput*{:U}{$\supsetneq$}
\ncline{-}{oBS}{oB}\lput*{:U}{$\supsetneq$}
\ncline{-}{oS}{oBS}\lput*{:U}{$\subsetneq$}
}
\ncline{<->}{oS}{oB}\Aput*{$\mathcal{C}$}
\psset{linearc=.3}
\ncloop[angleA=180,angleB=0,loopsize=1,arm=.5]{->}{o}{o}
\Bput*{$\cap,\cup,\mathcal{C},\pi$}
\ncloop[angleA=270,angleB=90,loopsize=2.5,arm=.5]{->}{oS}{oS}
\aput*{0}(3.5){$\cup,\cap,\pi$}
\ncloop[angleA=90,angleB=270,loopsize=2.5,arm=.5]{->}{oB}{oB}
\aput*{0}(1.5){$\cup,\cap,\pi$}
\ncloop[angleB=180,loopsize=1,arm=.5]{->}{oBS}{oBS}
\Bput*{$\cup,\cap,\pi$}
\end{center}
\vspace{0.5cm}
\caption{The diamond}\label{figure:diamond}
\end{figure*}
}

The diamond illustrates the four variants of languages
of~$\omega$-words that we  consider: $\omega$-regular, $\omega B$-regular,
$\omega S$-regular and \ombs-regular languages.  The inclusions
between the four classes give a diamond shape.
We show in Section~\ref{subsection:limits} \oldhl{ below} that the
inclusions in the diamond are indeed strict.

To each class of languages corresponds a family of automata. The
automata come in two variants: ``normal automata'', and the
equally expressive
``hierarchical automata''. The exact definition of these automata as
well as the corresponding equivalences are the subject of \hhl{Section~\ref{section:automata}.}

All the classes are closed under union, \oldhl{since $\omega BS$-regular
  expressions have finite union built into the syntax.}  It is also easy
to show that the classes are closed under projection, i.e., images
under a letter to letter morphism (the operation denoted by~$\pi$ in
the figure), and more generally \oldhl{images under a} homomorphism.
From the automata characterizations we obtain closure under intersection for
the four classes; see Corollary~\ref{corollary:intersection}.  For 
closure under complement, things are not so easy.  Indeed\oldhl{, } in
Section~\ref{subsection:limits} we show that~$\omega BS$-regular
\oldhl{languages} are not closed under complement.  However, some
complementation results are still possible.  Namely
Theorem~\ref{theorem:complementation} establishes that complementing
an $\omega B$-regular language gives an $\omega S$-regular language, and
vice versa.  This is the most involved result of the paper, and
Section~\ref{section:complementation} is dedicated to its proof.


\subsection{Limits of the diamond}
\label{subsection:limits}

In this section we show that all the inclusions depicted in the
diamond are strict. Moreover, we show that there exists
an $\ombs$-regular language whose complement is not $\ombs$-regular.

\begin{lem}\label{lemma:B-pump}
 \hhl{ Every~$\omega B$-regular language over the alphabet~$\set{a,b}$, 
which contains a word with an infinite number of~$b$'s, also
  contains a word in~$(a^Bb)^\omega$.}
\end{lem}
\begin{proof}
Using a straightforward structural induction, one can show  that a~$B$-regular language
  of word sequences~$L$ satisfies the following properties.
\begin{itemize}
\item If~$L$ contains a sequence in~$a^*$, then it contains a sequence in~$a^B$.
\item If~$L$ contains a sequence in~$(a^*b)^+a^*$, then it contains
a sequence in~$(a^Bb)^+a^B$.
\end{itemize}
The statement of the lemma follows.
\end{proof}

\begin{cor}\label{corollary:notBnotS}
The language~$(a^Sb)^\omega$ is not~$\omega B$-regular.
The language~$(a^Bb)^\omega$ is not~$\omega S$-regular.
\end{cor}
\begin{proof}
Being~$\omega B$-regular for $(a^Sb)^\omega$ would contradict
Lemma~\ref{lemma:B-pump}, since it contains a word with an infinite number
of~$b$'s but does not intersect $(a^Bb)^\omega$.

For the second part, assume that the language~$(a^Bb)^\omega$
is~$\omega S$-regular. Consequently, so is the
language~$(a^Bb)^\omega+(a+b)^*a^\omega$.  Using
Theorem~\ref{theorem:complementation}, its complement
$((a^*b)^*a^Sb)^\omega$ would be~$\omega B$-regular.  But this is not
possible, by the same argument as above. A proof that does not use
complementation---along the same lines as in the first part---can
also be given.
\end{proof}

We now proceed to show that ~$\ombs$-regular languages are not closed
under complement. We start with a lemma similar to Lemma~\ref{lemma:B-pump}.
\begin{lem}\label{lemma:BS-pump}
  Every~$\ombs$-regular language over the alphabet~$\set{a,b}$ that
  contains a word with an infinite number of~$b$'s also contains a
  word in~$(a^Bb+a^Sb)^\omega$.
\end{lem}
\begin{proof}
  As in Lemma~\ref{lemma:B-pump}, one can  show the following properties of
  a~$BS$-regular language of word sequences~$L$.
\begin{itemize}
\item If~$L$ contains a word sequence in~$a^*$, then it contains one in~$a^B+a^S$.
\item If~$L$ contains a word sequence in~$(a^*b)^+a^*$, then it
  contains one in~$(a^Bb+a^Sb)^+(a^B+a^S)$.
\end{itemize}
The result directly follows.
\end{proof}
\begin{cor}
The complement of~$L=(a^Bb+a^Sb)^\omega$ is not~$\ombs$-regular.
\end{cor}
\begin{proof}
  The complement of~$L$ contains the word
  \begin{equation*}
    a^1~ba^1ba^2~ba^1ba^2ba^3~ba^1ba^2ba^3ba^4b~\dots\ ,
  \end{equation*}
  and consequently, assuming it is~$\omega BS$-regular, one can apply
  Lemma~\ref{lemma:BS-pump} to it.  It follows that the complement
  of~$L$ should intersect~$L$, a contradiction.
\end{proof}


\section{Automata}
\label{section:automata}
In this section we introduce a new \oldhl{automaton model for}
infinite words, which we call \ombs-automata. We show that 
these automata have the same expressive power
as $\ombs$-regular expressions.

We will actually define two \oldhl{models}: \ombs-automata, and
hierarchical \ombs-automata. We present them\oldhl, successively\oldhl, in
Sections~\ref{subsection:normal-automata}
and~\ref{subsection:hierarchical-automata}. 
  Theorem~\ref{theorem:equivalence}, which shows that both models have
  the same expressive power as $\omega BS$-regular expressions, is
  given in Section~\ref{subsection:statement-equivalences}.
  Section~\ref{subsection:statement-equivalences} also contains
  Corollary~\ref{corollary:intersection}, which \ttc{states} that $\omega
  BS$-regular languages are closed \ttc{under} intersection. 
  The proof of
 Theorem~\ref{theorem:equivalence} is presented in
Sections~\ref{section:expression-to-automata}
and~\ref{section:automata-to-expressions}\ttc{.}

\subsection{General form of \ombs-automata}
\label{subsection:normal-automata}

An \emph{\ombs-automaton} is a tuple
$(Q,\Sigma,q_I,\Gamma_B,\Gamma_S,\delta)$,
in which~$Q$ is a finite set of \intro{states}, $\Sigma$ is the
  input alphabet,
$q_I\in Q$ is the \intro{initial state}, \oldhl{and }
$\Gamma_B$ and~$\Gamma_S$ are two disjoint finite sets of
\intro{counters}. We set~$\Gamma=\Gamma_B\cup\Gamma_S$. The
mapping $\delta$ associates with each
letter~$a\in\Sigma$ its \intro{transition relation}
\begin{equation*}
  \delta_a \subseteq Q
  \times \set{\inc, \reset, \epsilon}^\Gamma \times Q\ .
\end{equation*}
 The counters in~$\Gamma_B$
are called \emph{bounding counters}, or \intro{$B$-counters} or
\intro{counters of type~$B$}. The counters in $\Gamma_S$ are called \emph{unbounding
  counters}, or \intro{$S$-counters} or \intro{counters of
  type~$S$}.
Given a counter~$\alpha$ and a
transition~$(p,v,q)$, the transition is called a \intro{reset}
of~$\alpha$ if~$v(\alpha)=\reset$; it is an \intro{increment of~$\alpha$}
if~$v(\alpha)=\inc$.

When the automaton only has
counters of type~$B$, i.e., if $\Gamma_S=\emptyset$
(\oldhl{respectively, } of type~$S$, i.e., if $\Gamma_B=\emptyset$),
then the automaton is called an \intro{$\omega B$-automaton}
(\oldhl{respectively, } an \intro{$\omega S$-automaton}).

A \emph{run} $\rho$ of an \ombs-automaton over a finite or infinite
word  $a_1
a_2 \cdots$ is a sequence of transitions $\rho=t_1 t_2 \cdots$
of same length such
that for all~$i$, $t_i$ belongs to $\delta_{a_i}$,
and the target state of the transition $t_i$ is the same as the source
state of the transition $t_{i+1}$.

Given a counter~$\alpha$, every run~$\rho$ \oldhl{can be uniquely decomposed as}
$$
{\rho=\rho_1t_{n_1}\rho_2 t_{n_2} \dots}
$$
in which, for every $i$, $t_{n_i}$ is a transition that does a reset
  on~$\alpha$, and $\rho_i$ is a subrun (sequence of transitions) that
  does no reset on~$\alpha$.  We denote by~$\alpha(\rho)$ the
 sequence of natural numbers, which on its
$i$-th position has the number of occurrences of increments
of~$\alpha$ in~$\rho_i$.  This sequence can be finite if the counter
is reset only a finite number of times\ttc{, otherwise it is} infinite.  A
\emph{run} $\rho$ over an $\omega$-word is \emph{accepting} if the
source of its first transition is the initial \ttc{state}~$q_I$, and for every
counter $\alpha$, the sequence $\alpha(\rho)$ is infinite and
furthermore, if~$\alpha$ is of type~$B$ then $\alpha(\rho)$ is bounded
and if~$\alpha$ is of type~$S$ then $\alpha(\rho)$ is strongly
unbounded.

\begin{exa}
Consider the following automaton with a single counter
(the counter action is in the parentheses):
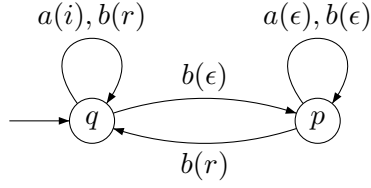
\begin{center}
    \begin{picture}(50,28)(0,0)
      \gasset{curvedepth=3,Nw=6,Nh=6}
      \node[Nframe=n](in)(-4,10){}
      \node(0)(10,10){$q$}
      \node(1)(40,10){$p$}
      \drawedge(0,1){$b(\epsilon)$}
      \drawedge(1,0){$b(r)$}

      \drawloop(0){$a( i),b( r)$}
      \drawloop(1){$a(\epsilon),b(\epsilon)$}
      \drawedge[curvedepth=0](in,0){}
    \end{picture}
\end{center}
\oldhl{Assume now that the unique counter in this automaton is of
  type~$B$. We claim that this automaton} recognizes the
language~$L=(a^Bb(a^*b)^*)^\omega$, i.e., the \oldhl{set} of
$\omega$-words of the form~$a^{n_0}ba^{n_1}b\dots$ such
that~$\liminf_i n_i<+\infty$. We only show that the automaton
  accepts all words in $L$, the converse inclusion is shown by using
  similar arguments. Let  $w=a^{n_0}ba^{n_1}b\dots$ be a word in
  $L$. There exists an infinite set~$K\subseteq\nats$ and a natural
number~$N$ such that~$n_{k}<N$ holds for all~$k\in
K$. Without loss of generality we assume that~$0\in K$ (by
replacing~$N$ by~$\max(n_0,N)$). \oldhl{We now construct an accepting run
  of the automaton on the word $w$. This run uses state $q$ in} \ttc{each}
  \oldhl{position $i \in \Nat$} \ttc{such that the distance between the
  last occurrence of $b$ before
  position $i$ and the first occurrence of $b$ after position $i$ is at most $N$.
  For other positions the state $p$ is used and the value of the counter is
  left unchanged.  
  In this way, the value
  of the  counter will never exceed $N$.
  Since furthermore $K$ is infinite, the counter will be reset infinitely many times.
  This proves that the run is accepting.}
  If the counter is of type~$S$, then the
  automaton recognizes the language~$(a^Sb(a^*b)^*)^\omega$.
\end{exa}

Although this is not directly included in the definition, an
\ombs-automaton can simulate a B\"uchi automaton, and this in two
ways: by using either unbounding, or bounding counters \oldhl{(and
  therefore B\"uchi automata are captured by both $\omega S$-automata
  and $\omega B$-automata).} Consider \hhl{then}
a B\"uchi automaton, with final states $F$. One way to simulate this
automaton is to have an \hhl{$\omega B$-automaton} with the same state space, and one
bounding counter, which is reset every time a state from $F$ is
visited, and never incremented.  In this case, the accepting condition
for $\ombs$-automata collapses to visiting $F$ infinitely often, since
a bounded value is assured by the assumption on not doing any
increments. Another way to simulate the B\"uchi automaton is to use an
unbounding counter. Whenever the simulating automaton\hhl{ visits a state in $F$}, it nondeterministically decides to either increment
the unbounding counter, or to reset it. It is easy to see that the
accepting state is seen infinitely often if and only if there is a
policy of incrementing and resetting that satisfies the accepting
condition for unbounding counters.

\medskip

\subsection{Hierarchical automata}
\label{subsection:hierarchical-automata}
\hyphenation{auto-mata}

Hierarchical \ombs-automata are a special case of {$\ombs$}\nobreakdash-\hspace{0pt}{auto}mata
  where a  \hhl{stack-like} discipline is imposed on the counter
  operations. 

An \ombs-automaton is called \emph{hierarchical} if its set of
counters is $\Gamma=\set{1,\ldots,n}$ and whenever a counter $i>1$ is
incremented or reset, the \hhl{counters~$1, \dots, i-1$} are reset.  
\oldhl{Therefore}, in a hierarchical automaton, a transition~$(q,v,r)$ can be of three
forms:
\begin{itemize}
\item $v(1)=\cdots=v(n)=\epsilon$, i.e., no counter is \oldhl{affected} by the transition.  In
      this case we write $\epsilon$ for~$v$.
\item $v(1)=\cdots=v(k)=r$ and $v(k+1)=\cdots=v(n)=\epsilon$, i.e.,
      the maximal \oldhl{affected} counter is $k$, and it is reset by the transition.
      In this case we write $R_k$ for $v$.
\item $v(1)=\cdots=v(k-1)=r$, $v(k)=i$ and $v(k+1)=\cdots=v(n)=\epsilon$., i.e.,
      the maximal \oldhl{affected} counter is $k$, and it is incremented by the transition.
      In this case we write $I_k$ for $v$.
\end{itemize}
It is
convenient to define for a hierarchical automaton its \intro{counter
type}, defined as a word in~$\set{B+S}^*$.
The length of this word is the number of counters; its $i$-th letter
is the type of counter~$i$.

\begin{exa} Consider the following hierarchical $\omega BS$-automaton:
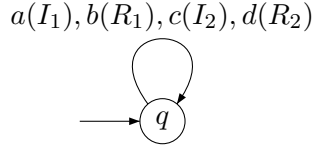
\begin{center}
    \begin{picture}(30,22)(0,6)
      \gasset{curvedepth=3,Nw=6,Nh=6}
      \node[Nframe=n](in)(-4,10){}
      \node(0)(10,10){$q$}
      \drawloop(0){$a({I_1}),b({R_1}),c({I_2}),d({R_2})$}
      \drawedge[curvedepth=0](in,0){}
    \end{picture}
\end{center}
In this picture, we use once more the convention that the resets and increments
are in parentheses. If this automaton has counter type $T_1T_2$
with~$T_1,T_2\in\{B,S\}$, then it recognizes the language
\begin{center}
\ttc{$\left(\left((a^{T_1}b)^*a^{T_1}c\right)^{T_2}(a^{T_1}b)^*a^{T_1}d\right)^\omega\ .$}
\end{center}
\end{exa}

\subsection{Equivalence.}
\label{subsection:statement-equivalences}

The key result concerning the automata is that the hierarchical
  ones are equivalent to the non-hierarchical ones, and that both are
  equivalent to the \hhl{$\omega BS$-regular expressions}. Furthermore, this equivalence also
  holds for the fragments where only $B$-counters or $S$-counters are
  allowed.
\begin{thm}\label{theorem:equivalence}
\hhl{Let $T \in \set{BS, B, S}$. 
  The following are equivalent for  a  language $L \subseteq \Sigma^\omega$.
\begin{enumerate}
\item $L$ is $\omega T$-regular.
\item $L$ is recognized by a hierarchical $\omega T$-automaton.
\item $L$ is recognized by an $\omega T$-automaton.
\end{enumerate}}
\end{thm}
Before establishing this result, we mention an important application. 
\begin{cor}\label{corollary:intersection}
  The classes of $\ombs$-regular, $\omega B$-regular and $\omega
  S$-regular languages are closed under intersection.
\end{cor}
\begin{proof}
  The corresponding automata (in their non-hierarchical form)
  are closed under intersection using the standard product construction.
\end{proof}

The implication from~(1) to~(2) is straightforward since hierarchical
\oldhl{automata are a special case of general automata.}
In the following sections, we show the two
difficult implications in Theorem~\ref{theorem:equivalence}: that
expressions are captured by hierarchical automata
(Section~\ref{section:expression-to-automata}) and that automata are
captured by expressions
(Section~\ref{section:automata-to-expressions}).
First, we introduce in Section~\ref{subsection:sequence-automata}
the notion of word sequence automata.

\subsection{Word sequence automata.}
\label{subsection:sequence-automata}

In proving the equivalence of hierarchical \ombs-automata and
\ombs-regular languages, we will use a form of automaton that runs
over word sequences.  A \emph{word sequence $BS$-automaton} $\Aa$ is
\hhl{defined as an \ombs-automaton}, except that we add a set of
accepting states, i.e., it is a tuple
$(Q,\Sigma,q_I,F,\Gamma_B,\Gamma_S,\delta)$ in
which~$Q,\Sigma,q_I,\Gamma_B,\Gamma_S,\delta$ are as for~$\omega
BS$-automata, and~$F\subseteq Q$ is \hhl{a} set of \intro{accepting states}.

A word sequence $BS$-automaton~$\Aa$ \intro{accepts}  
an infinite sequence of finite words
$\vec u$ if there is a sequence $\vec \rho$ of finite runs
of~$\Aa$ such that:
\begin{itemize}
\item \oldhl{f}or all $i \in \Nat$, the run $\rho_i$ is a run over the word
  $u_i$ that begins in the initial state and ends in an accepting
  state;
\item for every counter~$\alpha$ of type~$B$, the sequence of natural numbers
	$$\max(\alpha(\rho_0)),\max(\alpha(\rho_1)),\dots$$
	(in which~$\max$ is applied to a finite sequence of natural
	numbers with the obvious meaning) is bounded;
\item for every counter~$\alpha$ of type~$S$, the sequence of natural numbers
	$$\min(\alpha(\rho_0)),\min(\alpha(\rho_1)),\dots$$
	(in which~$\min$ is applied to a sequence of natural numbers
	with the obvious meaning) is strongly unbounded.
\end{itemize}
\hhl{The  variants
of \intro{ word sequence $B$-automata}
and \intro{ word sequence $S$-automata} are defined as expected.} \hhl{The same goes for
the hierarchical automata.}

\hyphenation{auto-maton}

An equivalent way for describing the acceptance of
a word sequence by a $BS$\nobreakdash-\hspace{0pt}automaton~$\Aa$ is \oldhl{as follows.}
Consider the $\omega BS$-automaton~$\Aa'$ obtained from~$\Aa$
by a) removing the set of accepting states~$F$, b) adding a new
symbol~$\square$ to the alphabet, and c) setting~$\delta_\square$
to contain the transitions $(q,R,q_I)$ for~$q\in F$ with
$R(\alpha)=\reset$ for every counter~$\alpha$.
Then~$\Aa$ accepts the word sequence~$\langle v_0,v_1,\dots\rangle$
iff~$\Aa'$ accepts the $\omega$-word~\oldhl{$v_0\square v_1 \square\dots$}


\newcommand{\purify}[1]	{\mathit{purify}({#1})}
\newcommand{\clean}[1]	{\mathit{clean}({#1})}

\subsection{From expressions to hierarchical automata}
\label{section:expression-to-automata}

This section is devoted to showing one of the implications in
Theorem~\ref{theorem:equivalence}:
\begin{lem}\label{lemma:expression-to-automata}
  Every \ombs-regular
  (\oldhl{respectively, } $\omega B$-regular, $\omega S$-regular)
  language can be recognized by a hierarchical \ombs-automaton
  (\oldhl{respectively, } $\omega B$-regular,  $\omega S$-regular).
\end{lem}

\hhl{ There are two main difficulties:}
\begin{itemize}
\item Our word sequence automata do not have
  $\varepsilon$-transitions, which are used in the proof for finite
  words. Instead of introducing a notion of $\varepsilon$-transition
  and establishing that such transitions can be eliminated from automata,
  we directly work on word sequence automata
  without $\varepsilon$-transitions.
\item When taking the mix or the concatenation of two languages $L,K$ defined by
  hierarchical word sequence automata, there are technical
  difficulties with combining the counter types of the automata for
  $L$ and $K$.
\end{itemize}
\hhl{We overcome  these difficulties by first rewriting an expression into
normal  form before compiling it into a hierarchical word sequence
automaton}.  The basic idea is that we move the mix $+$ to the top
level of the expression, and also remove empty words and empty
iterations.  To enforce that no empty
words occur, we use the exponents~$L^+$, $L^{S+}$
and~$L^{B+}$ which correspond to~$L\concat L^*$,
$L\concat L^S$ and~$L\concat L^{B}$, respectively.

We say that a $BS$-regular expression is \intro{pure} if it is
constructed solely with the single letter constants~$a$ and~$\bar a$,
concatenation $\concat$, and the exponents $+$, $B+$ and $S+$.
We say a $BS$-regular expression is in \emph{normal form} if
it is a mix $e_1 + \cdots + e_n$ of pure $BS$-regular expressions $e_1,\ldots,e_n$.
An $\omega BS$-regular expression is in \emph{normal form} if
all the $BS$-regular expressions in it are in normal form.

In Section~\ref{sec:conv-an-expr}, we show that every $BS$-regular
language (with no occurrence of $\varepsilon$) can be described
by a $BS$-regular expression in normal form.
Then, in Section~\ref{sec:from-normal-form}, we show that every
$\omega BS$-regular expression in normal form can be compiled into a hierarchical word sequence automaton.

\subsubsection{Converting an expression into normal form}
\label{sec:conv-an-expr}

Given a $BS$-regular language~$L$, let us define~$\clean L$ to be the
set of word sequences in $L$ that have nonempty words on all
coordinates. Remark that, thanks to Fact~\ref{fact:basic-BS-closure},
$L^\omega$ is the same as $(\clean L)^\omega$.
Therefore, we only need to work on sequence languages of the form
$\clean L$.
Using
Fact~\ref{fact:basic-BS-closure} one obtains without difficulty:
\begin{fact}\label{fact:clean}
\hhl{For every $BS$-regular language~$L$, either $L=\clean L$, 
  $L=\fepsilon+\clean L$, or $L=\varepsilon+\clean L$.}
\end{fact}
 The following lemma concludes the conversion into normal
form:
\begin{lem}\label{lemma:pure-decompositon}
  For every $BS$-regular language $L$, $\clean L$ can be rewritten as
  $L_1+\dots+L_n$, where each~$L_i$ is pure.
\end{lem}

\hhl{The proof of this lemma has a similar flavor as the proof of the
analogous result for finite words, which says that union can be
shifted to the topmost level in a regular expression.}

\begin{proof}

The proof is by induction on the structure.
\begin{itemize}
\item For $L=\emptyset,\varepsilon,\fepsilon,a$, the claim is obvious.
\item Case~$L=K \concat M$. \hhl{There} are nine subcases
	(according to Fact~\ref{fact:clean}):
	\begin{itemize}
	\item If~$K$ is $\clean K$ and~$M$ is $\clean M$ then
          \begin{equation*}
            \clean{K \cdot M}=\clean K\concat\clean M\ .
          \end{equation*}
          \oldhl{We then use the induction assumption on $\clean K$ and
          $\clean M$, and
          concatenate the two unions of pure expressions. This
          concatenation is also a union of pure expressions, since}
                \begin{equation*}
              \sum_i K_i \cdot \sum_j M_j =
                  \sum_{i,j} K_i \cdot M_j\ .
                \end{equation*}
        \item If~$K$ is $\clean K$ and~$M$ is
                $\fepsilon+\clean M$
                then
                \begin{equation*}
                  \clean{K\concat M}=\overline{\clean K}+\clean
                K\concat\clean M\ .
                \end{equation*}
                In the above, the language $\overline{\clean K}$
                consists of finite prefixes of sequences in $\clean
                K$. A pure expression for this language is obtained
                from $\clean K$ by replacing every exponent with $+$ and
                every letter $a$ with $\bar a$.
        \item If~$K$ is $\clean K$ and~$M$ is
                $\varepsilon+\clean M$ then
                \begin{equation*}
                 \clean{K\concat M}=\clean
                K+\clean K\concat\clean M\ .
                \end{equation*}
	\item The six other cases are similar.
	\end{itemize}
\item Case~$L=K+M$. We have $\clean{K+M}=\clean K+\clean M$.
\item Case~$L=K^*$. We have $\clean{K^*}=(\clean K)^+$.  By induction
      hypothesis, this becomes $(L_1 + \cdots + L_n)^+$, for pure
      expressions $L_i$. We need to show how the mix $+$ can be moved to
      the top level.  For $n=2$, we use:
	\begin{align*}
	(L_1+L_2)^+ = &L_1^+ 	+(L_2^+\concat L_1^+)^+
			+(L_2^+\concat L_1^+)^+\cdot L_2^+ + \\ &
			+(L_1^+\concat L_2^+)^+
			+(L_1^+\concat L_2^+)^+\cdot L_1^+
			+L_2^+\ .
	\end{align*}
	The general case is obtained by an inductive use of this equivalence.
      \item Case~$L=K^S$. This time, we use~$\clean{K^S}=(\clean K)^{S+}$,
	and get by induction an expression of the form $(L_1 + \cdots + L_n)^{S+}$,
	for pure expressions $L_i$. 	\oldhl{We only do the case of
          $n=2$, the general case is obtained by induction on~$n$}:
	\begin{align*}
	(L_1+L_2)^{S+}&=(L_1+L_2)^*\concat L_1^{S+}\concat(L_1+L_2)^*\\
			&+(L_1+L_2)^*\concat L_2^{S+}\concat(L_1+L_2)^*\\
			&+L_2^*\concat (L_1^+\concat L_2^+)^{S+}\concat L_1^*\ .
	\end{align*}
	The right \oldhl{side} of the equation is not yet in the correct form, i.e., it is not
	a mix of pure expressions, but it can be made so using the mix, the concatenation
	and the $*$ exponent cases described above (resulting in a mix of 102 pure expressions).

  \item Case~$L=K^B$. Same as for case~$K^*$, in which the exponent~$*$ is replaced by~$B$
	and the exponent~$+$ is replaced by~$B+$.
    \qedhere
\end{itemize}
\end{proof}

\subsubsection{From normal form to automata}
 \label{sec:from-normal-form}
 Here we show that every expression in normal form can be compiled
 into a hierarchical automaton. Furthermore, if the expression is
 $\omega B$-regular (respectively, $\omega S$-regular), then the
 automaton has the appropriate counter type.

We begin by showing how to expand counter types:
\begin{lem}\label{lemma:expand-counter-type}
  A hierarchical $BS$-automaton of counter type $tt' \in \set {B,S}^*$
  can be transformed into an equivalent one of counter type $tBt'$.
  A hierarchical $BS$-automaton of counter type $tSt' \in \set {B,S}^*$
  can be transformed into an equivalent one of counter type $tSSt'$.
\end{lem}
\begin{proof}
  Let~$\Aa$ be the automaton.
  For the first construction, we insert a new counter of type~$B$ at
  the correct position, i.e., between counter~$|t|$ \hhl{and~$|t|+1$}, and reset it as often as possible,
  i.e., we construct a new automaton~$\Aa'$ of counter type~$tBt'$
  which is similar to~$\Aa$ in all respects but
  every transition~$(p,v,q)$ of~$\Aa$ becomes a transition~$(p,v',q)$ in~$\Aa'$ with:
  \begin{equation*}
  v'=   \begin{cases}
	R_1&\text{if}~v=\epsilon~\text{and}~|t|=0\\
	\epsilon&\text{if}~v=\epsilon~\text{and}~|t|>0\\
	I_k	&\text{if}~v=I_k~\text{and}~k\leq|t|\\
	R_k	&\text{if}~v=R_k~\text{and}~k<|t|\\
	R_{|t|+1}&\text{if}~v=R_{|t|}\\
	I_{k+1}&\text{if}~v=I_k~\text{and}~k>|t|\\
	R_{k+1}&\text{if}~v=R_k~\text{and}~k>|t|
	\end{cases}
  \end{equation*}
  For every $\omega$-word~$u$, this translation \hhl{gives a natural
  bijection between
  the runs of~$\Aa$ over~$u$
 and the runs of~$\Aa'$ over~$u$}. This translation
  of runs preserves the accepting condition. Hence the language accepted by~$\Aa$
  and~$\Aa'$ are the same.

  For the second construction, we split the $S$ counter into two
  nested copies. The automaton chooses nondeterministically which
  one to increment. Formally, we transform every transition~$(p,v,q)$ of~$\Aa$
  into possibly multiple transitions in~$\Aa'$, namely the transitions $(p,v',q)$
  with~$v'\in V$ in which:
  \begin{equation*}
  V=   \begin{cases}
	\set\epsilon	&\text{if}~v=\epsilon\\
	\set{I_k}	&\text{if}~v=I_k~\text{and}~k\leq|t|\\
	\set{R_k}	&\text{if}~v=R_k~\text{and}~k\leq|t|\\
	\set{I_{|t|+1},I_{|t|+2}}&\text{if}~v=I_{|t|+1}\\
	\set{R_{|t|+2}}	&\text{if}~v=R_{|t|+1}\\
	\set{I_{k+1}}&\text{if}~v=I_k~\text{and}~k>|t|+1\\
	\set{R_{k+1}}&\text{if}~v=R_k~\text{and}~k>|t|+1
	\end{cases}
  \end{equation*}
  For every \oldhl{input} $\omega$-word~$u$, this transformation induces a natural surjective
  mapping \oldhl{from runs of~$\Aa'$ onto runs of~$\Aa$}.
  Accepting runs are \hhl{mapped} by this translation to accepting runs of~$\Aa$. Hence
  the language accepted by~$\Aa'$ is a subset of the one accepted by~$\Aa$.
  For the converse inclusion, one needs to transform an accepting run~$\rho$
  of~$\Aa$ over~$u$
  into an accepting run of~$\Aa'$ over~$u$. For this one needs to decide each time
  a transition $(p,I_{|t|+1},q)$ is used by~$\rho$, whether to use the transition
  $(p,I_{|t|+1},q)$ or the transition $(p,I_{|t|+2},q)$ of~$\Aa'$.
  For this, for all maximal subruns of~$\rho$ of the form
  $$\rho'_0(p_1,I_{|t|+1},q_1)\rho'_1\cdots(p_n,I_{|t|+1},q_n)\rho'_n$$
  in which the counter~$|t|+1$ is never reset, and such that the counter~$|t|+1$
  is never incremented in the runs~$\rho'_i$,
  one replaces it by the run of~$\Aa'$:
  $$\rho''_0(p_1,I_{x_1},q_1)\rho''_1\cdots(p_n,I_{x_n},q_n)\rho''_n$$
\hhl{where
\begin{eqnarray*}
  x_i&=  \begin{cases}
			|t|+2&\text{if~$i$~is a multiple of~$\lceil\sqrt{n}\rceil$}\\
			|t|+1&\text{otherwise,}
			\end{cases}
\end{eqnarray*}
and  $\rho''_i$ is~$\rho'_i$ in which each counter~$k>|t|+1$ is replaced by counter~$k+1$.}
  This operation transforms an accepting run of~$\Aa$ into an accepting \hhl{run of~$\Aa'$.}
  \end{proof}
We will use the following corollary of
Lemma~\ref{lemma:expand-counter-type}, which says that any number of
hierarchical $BS$-automata can be transformed into equivalent ones
that have comparable counter types.
\begin{cor}\label{cor:counter-type-prefix}
  Given hierarchical
  $BS$-automata~$\mathcal{A}_1,\ldots,\mathcal{A}_n$, there exist
  \oldhl{(respectively) equivalent}
  hierarchical $BS$-automata $\mathcal{A}'_1,\ldots,\mathcal{A}'_n$, such that
  for all~$i,j=1\dots n$
  the counter type of $\mathcal{A}'_i$ is a prefix of the counter type
  of $\mathcal{A}_j'$ or vice versa. Furthermore\oldhl{,} if
  $\mathcal{A}_1,\ldots,\mathcal{A}_n$ are B-automata (\oldhl{respectively, }
  S-automata), then so are $\mathcal{A}'_1,\ldots,\mathcal{A}'_n$.
\end{cor}

Recall that we want to compile a normal form expression
\begin{eqnarray*}
  M\concat(L_1+ \cdots + L_n)^\omega
\end{eqnarray*}
into a hierarchical $\omega BS$-automaton (or $\omega B$, or $\omega
S$ automaton, as the case may be). This is done in the next two
lemmas. First, Lemma~\ref{lemma:equiv-expression-word sequence}
translates each $L_i$ into a hierarchical word sequence automaton, and
then Lemma~\ref{lemma:squeeze-the-mix} combines these word sequence
automata into an automaton for the infinite words $(L_1 + \cdots +
L_n)^\omega$. \oldhl{Since  prefixing the language $M$ is  a trivial operation for the
automata, we thus  obtain the desired
Lemma}~\ref{lemma:expression-to-automata}\oldhl{, which says that expressions
can be compiled into hierarchical automata.}

\begin{lem}\label{lemma:equiv-expression-word sequence}
\oldhl{The language of word sequences described by a} pure $BS$-regular (\oldhl{respectively, } $B$-regular, \oldhl{respectively, } $S$-regular) expression
   \oldhl{can be recognized by} a hierarchical word sequence $BS$-automaton (\oldhl{respectively, } $B$-automaton,
  \oldhl{respectively, }  $S$-automaton).
\end{lem}
\begin{proof}
By induction \oldhl{on the operations that appear in a pure expression.}
\begin{itemize}
\item \oldhl{Languages accepted by hierarchical word sequence automata are closed
  under concatenation. }Let us compute\oldhl{ an automaton recognizing}~$L \concat L'$
  where~$L$, $L'$ are  \oldhl{languages recognized by} 
      hierarchical word sequence automata~$\Aa$, $\Aa'$ respectively.
      Using Corollary~\ref{cor:counter-type-prefix}, we can assume without
      loss of generality that the type of~$\Aa$ is a prefix of the
      type of~$\Aa'$ (\oldhl{or the other way round, which is a
        symmetric case}).
      Remark that since $L$ and $L'$ are pure,
      \oldhl{no state  in~$\Aa$ or $\Aa'$ }is both initial and final.

      We do the standard
      concatenation construction for finite automata (by passing from
      a final state of $\Aa$ to the initial state of $\Aa'$), except
      that when passing from $\Aa$ to $\Aa'$ we reset the highest
      counter available to $\Aa$.
\item Languages accepted by hierarchical word sequence automata are closed
  under the $+$ exponent.  We use the standard construction: linking
  all final states to the initial state while resetting all counters.
  In order to have non-empty words on all coordinates, the initial
  state cannot be accepting (if it is accepting, we add a new initial state).
\item Languages accepted by hierarchical word sequence automata are closed
  under the $S+$ exponent.  We add a new counter of type~$S$ of rank
  higher than all others.  We then proceed to the \hhl{construction
    for the $+$ exponent} as
  above, except that we increment the new counter whenever looping
  from a final state to the initial one.
\item For the $B+$ exponent, we proceed as above, except
  that the new counter is of type~$B$ instead of being of type~$S$.
  \qedhere
\end{itemize}
\end{proof}

The compilation of normal form expressions into automata is concluded by the following lemma:
\begin{lem}\label{lemma:squeeze-the-mix}
  Let $L_1,\ldots,L_n$ be sequence languages recognized by
  the hierarchical $BS$-automata $\Aa_1,\ldots,\Aa_n$ respectively. The language $(L_1
  + \cdots + L_n)^\omega$ is recognized by an $\omega BS$-automaton.
  Likewise for $\omega B$-automata and $\omega S$-automata.
\end{lem}
\begin{proof}
  Thanks to Corollary~\ref{cor:counter-type-prefix}, we can assume
  that the counter type of $\Aa_i$ is a prefix of the counter type of
  $\Aa_j$, for $i \le j$.  We use this prefix assumption to share the
  counters between the automata: we assume that the counters in
  $\Aa_i$ are a subset of the counters in $\Aa_j$, for $i \le j$.
  Under this assumption, the hierarchical automaton for infinite words
  can nondeterministically guess a factorization of the infinite word
  in finite words, and nondeterministically choose one of the automata
  $\Aa_1,\ldots,\Aa_n$ for each factor.
\end{proof}

\subsection{From automata to expressions}

\label{section:automata-to-expressions}
This section is devoted to showing the remaining implication in
Theorem~\ref{theorem:equivalence}:
\begin{lem}\label{lemma:automata-to-expressions}
  Every language recognized by an \ombs-automaton
  (\oldhl{respectively, } an $\omega B$-automaton,  an $\omega S$-automaton)
  can be defined using an \ombs-regular
  (\oldhl{respectively, } $\omega B$-regular,  $\omega S$-regular)
  expression.
\end{lem}

Before we continue, we would like to remark that although long, the
proof does not require any substantially original ideas. Basically, it
consists of observations of the type: ``if a word contains many $a$'s
and a $b$, then it either has many $a$'s before the $b$ or many $a$'s
after the $b$''. Using such observations and Kleene's theorem for
finite words, we obtain the desired result.

We begin by introducing a technical tool called external constraints.
These constraints are then used in
Section~\ref{section:controlling-counters} to prove
Lemma~\ref{section:automata-to-expressions}.

\subsubsection{External constraints}
\label{section:external-constraints}

External constraints provide a convenient way of constructing
\ombs-regular languages.  Let $e$ be one of $0,+,S,B$. Given a symbol
$a \not\in \Sigma$ and a word sequence language $L$ over $\Sigma\cup\set a$,
we denote by $L\ralph e$ the word sequence language
\begin{equation*}
  L\cap ((\Sigma^*\cdot a)^e\cdot\Sigma^*)
\end{equation*}
with the standard convention that~$L^0=\sem{\varepsilon}$.
This corresponds to restricting the word sequences in~$L$ to ones where the
number of occurrences of~$a$ satisfies the constraint~$e$.

\begin{figure}
\begin{align*}
\begin{array}{rl}
  \e\ralph +	& =	\emptyset			\\
  \e\ralph e	& =	\e		\qquad\text{for}~e\in\set{0,B}	\\
  \e\ralph S	& =	\fepsilon			\\
  \fepsilon\ralph e& =	\fepsilon	\qquad\text{for}~e\in\set{0,B,S}	\\
  \fepsilon\ralph +& =	\emptyset
  \end{array}&&
  \begin{array}{rl}
  b\ralph +	& =	\emptyset	\qquad\text{for}~b \neq	a	\\
  b\ralph0	&= b	\qquad\text{for}~b \neq	a			\\
  b\ralph S	& =	\bar b		\qquad\text{for}~b \neq	a	\\
  b\ralph B	& =	b			\qquad\text{for } b \neq a	\\
  \\
  a\ralph 0	& =	\emptyset							\\
  a\ralph S	& =	\bar a								\\
  a\ralph e	& =	a		\qquad\text{for}~e\in\set{B,+}	\\
  \end{array}
  \end{align*}
  \begin{align*}
  (L+L')\ralph e & = L\ralph e + L'\ralph e
  					\qquad\text{for}~e\in\set{0,+,B,S}	\\
  \\
  (L\concat L')\ralph e	& =	L\ralph e\concat L'\ralph e
  					\qquad\text{for}~e\in\set{B,0}	\\
  (L\concat L')\ralph e	& =	L\ralph e\concat L'+L\concat L'\ralph e
  					\qquad\text{for}~e\in\set{+,S}	\\
  \\
  L^*\ralph 0	& =	(L\ralph 0)^*\\
  L^*\ralph +	& =	L^*\concat L\ralph +\concat L^* \\
  L^*\ralph B	& =	(L\ralph B\concat (L\ralph 0)^*)^B \\
  L^*\ralph S	& =	L^*\concat(L\ralph +\concat L^*)^S + L^*\concat L\ralph S\concat L^*\\
 \\
  L^B\ralph 0	& =	(L\ralph0)^B			\\
  L^B\ralph +	& =	L^B\concat L\ralph +\concat L^B	\\
  L^B\ralph B	& =	(L\ralph B)^B			\\
  L^B\ralph S	& =	\fepsilon+L^B\concat L\ralph S\concat L^B	\\
  \\
  L^S\ralph 0 & = (L\ralph 0)^S								\\
  L^S\ralph + & = L^S\concat L\ralph +\concat L^* +L^*\concat L\ralph +\concat L^S	\\
  L^S\ralph B	& =	(L\ralph B\concat(L\ralph 0)^*)^B\concat
  		(L\ralph0)^S\concat (L\ralph B\concat (L\ralph 0)^*)^B	\\
  L^S\ralph S & =  L^S\concat L\ralph S\concat L^* + L^*\concat L\ralph S\concat L^S+ (L^*\concat L\ralph +)^S\concat L^*\\
\end{align*}
\caption{Elimination of external constraints.}\label{figure:elimination-ec}
\end{figure}

Here we show that external constraints can be eliminated:
\begin{lem}\label{lemma:external-constraints-closure}
 $BS$-regular languages of word sequences  are closed under external constraints.
  \ttc{In
   other words, if $L$ is a $BS$-regular language over $\Sigma$ and~$a$
   is a letter in~$\Sigma$,
   then $L \ralph e$ is a $BS$-regular language over $\Sigma$.}
 $B$-regular languages are closed under external constraints
of type $0,+,B$.
 $S$-regular languages are closed under external constraints
of type $0,+,S$.
\end{lem}
\begin{proof}
  Structural induction.
  The necessary steps are shown in Figure~\ref{figure:elimination-ec}.
  For some of the equivalences,
  we use closure properties of Fact~\ref{fact:basic-BS-closure}.
\end{proof}

\subsubsection{Controlling counters in BS-regular expressions}
\label{section:controlling-counters}

In this section we show how $BS$-regular languages can be intersected
with languages of specific forms. We use
this to write an \ombs-regular expression that describes successful runs of an
\ombs-regular automaton, thus completing the proof of
Lemma~\ref{lemma:automata-to-expressions}.

In the following lemma, the languages $L$ should be thought of as
describing runs of a $BS$-automaton.
The idea is that the language $K$ in the lemma constrains
  runs that are good from the point of view of \oldhl{one of the  counters}. The
  set of labels $I$ can be thought of \oldhl{
as representing the transitions  which increment the counter, $R$ as
representing the transitions which reset the  counter,  and $A$ as
representing the transitions which do not modify
it.}
The intersection in the lemma forces the counter operations
to be consistent with the acceptance condition.

Given a word sequence language $L$ and a subset $A$ of the alphabet, denote by $L \itl A^*$ the set of
word sequences that give a word sequence in $L$ if all letters from
$A$ are erased (i.e., a very restricted form of  \hhl{the} shuffle operator).  For
instance $B^* \itl A^*$ is the same as $(A+B)^*$.

\begin{lem}\label{lemma:interleaving}
  Let $\Sigma$ be an alphabet partitioned into sets $A,I,R$. Let
  $L$ be a $BS$-regular word sequence language over $\Sigma$.
  Then $K \cap L$ is also $BS$-regular, for  $K$ being one of:
\begin{multline*}
   (I^BR)^+I^B \itl A^*, \quad   (I^SR)^+I^* \itl A^*\ , \quad
   I^*(RI^S)^+ \itl A^*, \quad  I^S(RI^S)^+ \itl A^* \ .
\end{multline*}
  Similarly, $B$-regular languages are closed under intersection
  with the first language, and $S$-regular languages are closed
  under intersection with the three other languages.
\end{lem}

\begin{proof}
  In the proof we rewrite $K \cap L$ into an expression with external
  constraints, which will use and constrain letters from some new
  alphabet $\Sigma'$ disjoint with $\Sigma$.  In the proof, we will
  consider equality up to the removal of letters from $\Sigma'$. We
  then eliminate the external constraints using
  Lemma~\ref{lemma:external-constraints-closure}, and then erase the
  letters from $\Sigma'$ \oldhl{(erasing letters is allowed, since languages described by our
  expressions are closed under homomorphic images)}.

  We begin by showing that $(I^e \itl A^*)\cap L$ is $BS$-regular, for
  $e=0,+,B,S$. Let~$a \in \Sigma'$ be a new symbol.  The transformation is
  simple: replace everywhere in the expression the letters~$i$ in~$I$
  by~$i\concat a$, and constrain the resulting language by~$\ralph e$.

  For $K\ =(I^BR)^+I^B \itl A^*$, the construction is by induction on
  the size of the expression defining $L$. We use the following equalities
  (in which~$K'\ =I^B \itl A^*$):
\begin{align*}
  K \cap b		&= b\quad\text{if $b\in R$ and $\es$ otherwise} \\
  K \cap (L+L')	&= (K \cap L) + (K \cap L')   				\\
  K \cap (L\concat L')	&= (K \cap L)\concat ((K \cap  L')+(K'\cap L'))
  			\,+\, (K'\cap L)\concat (K \cap  L')	\\
  K \cap  L^*	&= ((K' \cap L^*)\concat(K \cap  L))^+\concat(K'\cap L^*)\ .
\end{align*}
The remaining cases, namely $K \cap L^B$ and $K \cap L^S$,
can be reduced to the $L^*$ case as follows.
First one rewrites $L^B$ (\oldhl{respectively, }~$L^S$) as~$(La)^*\ralph B$
(\oldhl{respectively, }~$(La)^*\ralph S$), where $a \in \Sigma'$ is a new letter
(recall that we consider here equality of languages up to removal of
letters from $\Sigma'$).
Second, we use the associativity of intersection, i.e.,
\begin{align*}
  K\cap((La)^* \ralph e) & = (K\cap (La)^*)\ralph e\,,\qquad{\text{for
    }}~e=B,S\,.
\end{align*}

For the case where $K$ is either $((I^SR)^+I^*) \itl A^*$ or
$(I^*(RI^S)^+)\itl A^*$, a slightly more tedious transformation is
involved.  This is also done by induction. To make the induction pass,
we generalize the result to languages $K$ of the form
\begin{equation*}
  K_{e,f}=I^eR(I^SR)^*I^f \itl A^*\,, \qquad \text{where } e,f\in\set{*,S}\ .
\end{equation*}
The transformations for $L=b$ and $L+L'$ are as follows:
\begin{align*}
K_{e,f} \cap b & = \begin{cases}
		b&\text{ if $b\in R$ and $e=f=*$,}\\
		\es&\text{ otherwise,}
		\end{cases}\\
K_{e,f} \cap (L+L')& =  (K_{e,f}\cap L)+(K_{e,f}\cap L')\ .
\end{align*}
For sequential composition $L\concat L'$, we use the convenient
operation~$\sqcup$ over the exponents
$\set{*,S}$, defined by
\begin{equation*}
  * \sqcup *=*,\qquad \text{and}~ e \sqcup f = S \text{ otherwise.}
\end{equation*}
The transformation for sequential composition $L\concat L'$ is then the
following:
\begin{align*}
K_{e,f}\cap(L\concat L')&=
	\sum\limits_{e' \sqcup f'=S} (K_{e,e'}\cap L)\concat(K_{f',f}\cap L')	\\
	&+\sum\limits_{e' \sqcup f'=e} (I^{e'}\cap L)\concat(K_{f',f}\cap L')	\\
	&+\sum\limits_{e' \sqcup f'=f} (K_{e,e'}\cap L)\concat(I^{f'}\cap L')
\end{align*}
The rule for $L^*$ is the most complex one. For conceptual simplicity
we use an infinite sum. This can be transformed into a less readable
but correct expression using standard methods for regular languages.

For $n \ge 1$, let $L_n$ be the mix of all languages of the form
\begin{align*}
  (I^{e_0}\cap L^*)\concat (K_{f_1,g_1}\cap L)\concat
  (I^{e_1}\cap L^*)\cdots
  (K_{f_n,g_n}\cap L)\concat (I^{e_n}\cap L^*)\ ,
\end{align*}
where the exponents $e_i,f_i,g_i\in\set{*,S}$ satisfy
\begin{equation*}
   g_i \sqcup e_i \sqcup f_{i+1}=S\,, \qquad e_0 \sqcup f_1=e\,,
   \qquad\text{and}~~g_n
   \sqcup e_n=f\ .
\end{equation*}
The language $L_n$ corresponds to those word sequences, where the reset is
done in $n$ separate iterations of $L$. The language  $K_{e,f}\cap
L^*$ is then  equal to the infinite mix $L_1 + L_2 + \cdots$
\end{proof}

We now use the above lemma to complete the proof of
Lemma~\ref{lemma:automata-to-expressions}. \hhl{Consider} an
\ombs-automaton $\Aa$.  We will present an expression not for the
recognized words, but the accepting runs. The result then follows by
projecting each transition onto the letter it reads. \hhl{Without loss of
generality we assume that a transition uniquely determines this
letter.}

Given a counter $\alpha$, let $I_\alpha$ represent the transitions that
increment this counter, let $R_\alpha$ represent the transitions that reset
it and let $A_\alpha$ be the remaining transitions.  Let $\Gamma_B$ be the
set of bounded counters of $\Aa$ and let $\Gamma_S$ be its unbounded
counters.  Given a state~$q$, we define~$\mathit{Pref}_q$ to be the
language of finite partial runs starting in the initial state
and ending in state~$q$, and~$\mathit{Loop}_q$ to be the language of
nonempty finite partial runs starting and ending in state~$q$.
Those languages are regular languages of finite words,
and hence can be described via regular
expressions. We use those expressions as word sequence expressions.

The following lemma concludes
the proof of Lemma~\ref{lemma:automata-to-expressions}, by showing
how the operations from Lemma~\ref{lemma:interleaving} can be used
to check that a run is accepting.

\begin{lem}\label{lemma:accepting-runs}
  A run $\rho$ visiting infinitely often a state~$q$ is accepting iff
  there is a partition of~$\Gamma_S$ into two sets $\Gamma_{S,*}$ and
  $\Gamma_{*,S}$ (either of which may be empty) such that
  \begin{equation*}
    \rho \in \mathit{Pref_q}\concat(\mathit{Loop}_q \cap L_B \cap L_{S,*} \cap L_{*,S})^\omega\ ,
  \end{equation*}
  where the languages $L_B$, $L_{S,*}$ and $L_{*,S}$ are defined as follows:
  \begin{align*}
    L_B		&=  \bigcap_{\alpha \in \Gamma_B}	((I_\alpha^BR_\alpha)^+I_\alpha^B) \itl A_\alpha^*\ ,\\
    L_{e,f}	&=  \bigcap_{\alpha \in \Gamma_{e,f}}
    	(I_\alpha^e R_\alpha(I^S_\alpha R_\alpha)^*I_\alpha^f) \itl A_\alpha^*\,,
	\qquad\text{for}~(e,f)=(*,S),(S,*)\ .
  \end{align*}
\end{lem}
\begin{proof}
\hhl{It is not difficult to show  that membership in $\mathit{Pref_q}(\mathit{Loop}_q \cap L_B \cap L_{S,*} \cap L_{*,S})^\omega$
  is sufficient for $\rho$ to be accepting.}

  For the other direction, consider an accepting run~$\rho$.  Let us
  consider an increasing infinite sequence of
  positions~$u_1,u_2,\dots$ in the run $\rho$ such that for each $n$,
  all counters are reset between~$u_n$ and~$u_{n+1}$ and the run
  assumes  state~$q$ at position~$u_n$. \oldhl{Such a sequence can
    be found since each counter is reset infinitely often.} Consider now an unbounded
  counter $\alpha$. For each~$n$ we define~$b^\alpha_n$ to be the number of
  increments of~$\alpha$ in~$\rho$ happening between~$u_n$ and the last
  reset \oldhl{before}~$u_n$, \oldhl{likewise, we define~$a^\alpha_n$ to be the} number of increments of~$\alpha$
  in~$\rho$ happening between~$u_n$ and the next reset of~$\alpha$
  after~$u_n$. By extracting a subsequence, we may assume that either
  $b^\alpha_n$ is always greater than $a^\alpha_n$, or $a^\alpha_n$ is always greater
  than $b^\alpha_n$. In the first case $b^\alpha_n$ is strongly unbounded and we
  put $\alpha$ into $\Gamma_{*,S}$; in the second \hhl{case} $a^\alpha_n$ is strongly
  unbounded and we put $\alpha$ into $\Gamma_{*,S}$. We iterate this
  process for all unbounded counters.
\end{proof}


\section{\oldhl{Monadic Second-Order Logic with Bounds}}
\label{section:msob}

In this section, we introduce the logic MSOLB.  This is a strict
extension of monadic second-order logic (MSOL), where a new
quantifier $\ubouns$ is added. This quantifier
expresses that a property is satisfied by arbitrarily large sets.
We are interested
in the \hhl{satisfiability  problem: given a formula of MSOLB, decide
  if it is satisfied in some
$\omega$-word.} We are not able to
solve this problem in its full generality. However, the diamond
properties \hhl{from the previous sections, together with the
  complementation result from Section~\ref{section:complementation}}, give an interesting partial
solution \oldhl{to the satisfiability problem}.

In Section~\ref{subsection:msob} we introduce the logic MSOLB.  In
Section~\ref{subsection:decidableMSOLB} we present some decidable
fragments of the logic MSOLB, and restate the diamond picture in this
new framework.  In Section~\ref{subsection:unbounding-quant} \ttc{we }\oldhl{show
  how the unbounding quantifier can be captured by our automaton
  model.}  In Section~\ref{subsection:automatic} we \oldhl{present an
  application of our logical results}\ttc{; namely we provide }\oldhl{an algorithm that
  decides if}\ttc{ }{an $\omega$-automatic graph has bounded degree.}

\subsection{The logic}
\label{subsection:msob}

Recall that monadic second-order logic (MSOL for short) is an
extension of first-order logic where \hhl{ quantification over sets} is allowed.
Hence a formula of this logic is made of atomic predicates, boolean
connectives ($\wedge,\vee,\neg$), first-order quantification ($\exists
x.\varphi$ and $\forall x.\varphi$) and set \ttc{quantification
(also called monadic second-order quantification)}
($\exists X.  \varphi$ and~$\forall X. \varphi$)
together with the membership predicate $x \in X$. \oldhl{A formula of
  MSOL can be evaluated in an $\omega$-word. In this case, the
  universe of the structure is the set $\Nat$ of word positions. The
  formula can also use the following  atomic
  predicates: a} binary predicate $x\leq y$ for order on
  positions, and for each letter $a$ of the alphabet, a unary
  predicate $a(x)$ that tests if a position~$x$ has the label
  $a$. \oldhl{This way, a formula that uses the above predicates defines
    a language of $\omega$-words: this is the set of those
    $\omega$-words}\ttc{ for which }\oldhl{it is satisfied.}

In the logic MSOLB we add a new quantifier\oldhl, the \intro{existential
  unbounding quantifier $\ubouns$}\oldhl, which \oldhl{can be defined as the
  following infinite }\ttc{conjunction}:
\begin{align*}
  \ubouns X. \varphi := \bigwedge_{N \in \Nat}\exists X.~(\varphi \land
  |X| \ge N) \ .
\end{align*}
The quantified variable $X$ is a set variable and~$|X|$ denotes
its cardinality.  Informally speaking, $\ubouns X. \varphi(X)$ says
that the formula $\varphi(X)$ is true for sets $X$ of arbitrarily
large cardinality. If $\varphi(X)$ is true for some infinite set $X$,
then $\ubouns X. \varphi(X)$ is immediately true. Note that $\varphi$ may contain other free variables than just $X$.

From this quantifier, we can construct other meaningful quantifiers:
\begin{itemize}
\item The \intro{universal above
    quantifier} $\qabove$ is the dual of~$\ubouns$, i.e., $\qabove X.\varphi$ is
  a shortcut for~$\neg\ubouns X.\neg\varphi$.  It is satisfied if all the
  sets~$X$ above some threshold of cardinality satisfy
  property~$\varphi$.
\item Finally, the \intro{bounding quantifier}~$\bouns$ is
  syntactically equivalent to the negation of the~$\ubouns$
  quantifier. \oldhl{Historically, this was the first quantifier to be studied,} in~\cite{bojanboun}.  It says that a formula~$\bouns X.\varphi$
  holds if there is a bound on the cardinality of sets satisfying
  property~$\varphi$.
\end{itemize}

Over finite structures, MSOLB and MSOL are equivalent: 
a subformula $\ubouns X.\varphi$ can never be satisfied \oldhl{in}
a finite structure, and consequently can
be removed from a formula.
Over infinite words, MSOLB defines strictly more
languages than MSOL. For instance the formula
\begin{gather*}
\bouns X.~[\forall x {\in} X.~a(x)] \land [\forall x \le  y \le  z.~
	(x,z\in X) \rightarrow(y \in X)]
\end{gather*}
\ttc{expresses that }\oldhl{there is a bound on
the size of contiguous segments made of~$a$'s. Over the
alphabet~$\set{a,b}$, this corresponds  to the language~$(a^Bb)^\omega$.}
As mentioned previously (recall Corollary~\ref{corollary:notBnotS}), this
language is not regular. Hence, this formula is not equivalent to any
MSOL formula.  This motivates the following decision problem:
\begin{quote}
  Is a given formula of MSOLB satisfied over some infinite word?
\end{quote}
We do not know the answer to this question in its full generality
\oldhl{(this problem may yet}\ttc{ be proved }\oldhl{to be  undecidable)}.
However, using the diamond (Figure~\ref{figure:diamond}), we can solve
this question for a certain class of formulas. This is the subject of
Sections~\ref{subsection:decidableMSOLB}
  and~\ref{subsection:unbounding-quant}. In
\oldhl {Section}~\ref{subsection:automatic}, we use the logic MSOLB to
  decide if a graph has bounded outdegree, for graphs interpreted in
  the natural numbers via monadic formulas.

\subsection{A decidable fragment of MSOLB}
\label{subsection:decidableMSOLB}

A classical approach for solving satisfiability of monadic
second-order logic is to translate formulas into automata
(this is the original approach of B\"uchi~\cite{buchi60,buchi62}
for finite and infinite words\ttc{, which has been later extended by
Rabin to infinite trees}~\cite{rabin69},
see  \cite{thomas} for a survey). To every
operation in the logic corresponds a language operation. As
\oldhl{languages recognized by} automata
are \oldhl{effectively} closed under those operations, and emptiness is decidable for
automata, the satisfaction problem is decidable for MSOL.  We use the
same approach for MSOLB. Unfortunately, our automata are not closed
under complement, hence we \oldhl{cannot use them to prove satisfiability
  for the whole logic, which is closed under complement.}

\pdflatex{
\newcommand{\temptab}[1]  {\begin{tabular}c#1\end{tabular}}
\begin{figure*}[ht]
\begin{center}
\begin{tabular}{ccc}
\\ \\ \\
	&\rnode{oBS}{\temptab{\ombs-regular languages\\$BS$-formulas}}			&\\
\\ \\ \\
\rnode{oS}{\temptab{$\omega S$-regular languages\\$S$-formulas}}\hspace{-1.2cm}
	&&\hspace{-1.2cm}\rnode{oB}{\temptab{$\omega B$-regular languages\\$B$-formulas}}\\
\\ \\ \\
	&\rnode{o}{{\temptab{$\omega$-regular languages\\MSOL-formulas}}}		&\\
\\ \\
\end{tabular}
{\psset{linewidth=0}
\ncline{-}{o}{oB}\lput*{:U}{$\subset$}
\ncline{-}{o}{oS}\lput*{:U}{$\subset$}
\ncline{-}{oB}{oBS}\lput*{:U}{$\subset$}
\ncline{-}{oS}{oBS}\lput*{:U}{$\subset$}
}
\ncline{<->}{oS}{oB}\Aput*{$\neg$}
\psset{linearc=.3}
\ncloop[angleA=180,angleB=0,loopsize=1,arm=.5]{->}{o}{o}
\Bput*{$\vee,\wedge,\neg,\exists,\forall$}
\ncloop[angleA=270,angleB=90,loopsize=2,arm=.5]{->}{oS}{oS}
\aput*{0}(3.5){$\vee,\wedge,\exists,\forall,\ubouns$}
\ncloop[angleA=90,angleB=270,loopsize=2,arm=.5]{->}{oB}{oB}
\aput*{0}(1.5){$\vee,\wedge,\exists,\forall,\qabove$}
\ncloop[angleB=180,loopsize=1,arm=.5]{->}{oBS}{oBS}
\Bput*{$\vee,\wedge,\exists,\ubouns$}
\end{center}
\vspace{2mm}
\caption{Logical view of the diamond}\label{figure:logical-diamond}
\end{figure*}
}

For this reason, we consider the following fragments
of the logic MSOLB, \oldhl{which are not, in general, closed under complementation.}
\begin{defi}\label{definition:fragments}
We distinguish the following syntactic subsets of MSOLB formulas:
  \begin{itemize}
  \item The \emph{$B$-formulas} include all of MSOL and are closed under
    $\lor, \land, \forall, \exists$ and $\qabove$. 
  \item The \emph{$S$-formulas} include all of MSOL and are closed under
    $\lor, \land, \forall, \exists$ and $\ubouns$. 
  \item The \emph{$BS$-formulas} include all $B$-formulas and
    $S$-formulas, and are closed under $\lor,\land,\exists$ and $\ubouns$.
  \end{itemize}
\end{defi}
\hhl{Note} that in this definition, $B$-formulas and~$S$-formulas are dual
in the sense that the negation of an $S$-formula  is logically equivalent
to a B-formula, and vice versa.
The above fragments are tightly connected to $\omega BS$-regular languages
according to the following fact:
\begin{fact}\label{fact:logic-is-automata}
  $BS$-formulas define exactly the \ombs-regular languages. Likewise
  for $B$-formulas, and $S$-formulas, with the corresponding languages
  being $\omega B$-regular and $\omega S$-regular.
\end{fact}
\begin{proof}(Sketch.)
\hhl{Thanks to the closure properties of $\omega BS$-regular
  languages, each $BS$-formula can be translated into an $\omega
  BS$-regular language.}
  We use here the standard coding of the valuation of free variables in the alphabet
  of the word, see, e.g. \cite{thomas}.
  Closure under $\vee$ and $\wedge$ is a direct consequence
  of closure under~$\cup$ and~$\cap$.
  Closure under~$\exists$ corresponds to closure under projection, which
  is straightforward for non-deterministic automata. Closure under
  universal quantification follows as dual of the existential
  quantifications (using the complementation result,
  Theorem~\ref{theorem:complementation}).
  Closure under~$\ubouns$ of $\omega S$-regular
  and $\omega BS$-regular
  languages is the subject of Section~\ref{subsection:unbounding-quant},
  and more precisely Proposition~\ref{proposition:U}.
  Closure under~$\qabove$ of $\omega B$-regular languages
  is obtained by duality (once more using Theorem~\ref{theorem:complementation}).

  For the converse implication, a formula of the
  logic can \oldhl{use existential }\ttc{set quantifiers in order to }guess
  a run of the automaton, and \ttc{then} check that this run is
  accepting using the new quantifiers.
\end{proof}

These fragments are summarized in the logical view of the diamond
presented in Figure~\ref{figure:logical-diamond}.
Also, from this fact we  get that $BS$-formulas are not expressively
complete for MSOLB, since MSOLB is closed under negation, while $\ombs$-regular
languages are not closed under complementation.

Since emptiness for $BS$-regular languages is decidable by
Fact~\ref{fact:emptiness-decidable}, we obtain the following decidability
result:
\begin{thm}\label{thm:dec-msob}
  The problem of satisfiability over $(\omega,<)$ is decidable for $BS$-formulas.
\end{thm}

In the \oldhl{proof of Fact}~\ref{fact:logic-is-automata}, we left out 
   the proof of closure under~$\ubouns$.
We present it in the next section.

\subsection{Closure under existential unbounding quantification $(\,\ubouns\,)$}
\label{subsection:unbounding-quant}

Here we show that the classes of~$\omega S$- and \ombs-regular
languages are closed under application of the quantifier~$\ubouns$.
This closure is settled by Proposition~\ref{proposition:U}.

\oldhl{In order to show this closure, we need to} describe the quantifier $\ubouns$ as a language
operation, in the same way existential quantification corresponds
to projection.  Let $\Sigma$ be an alphabet, and consider a language
$L \subseteq (\Sigma \times \set{0,1})^\omega$. Given a word $w \in
\Sigma^\omega$ and a set $X \subseteq \Nat$, let $w[X] \in (\Sigma
\times \set{0,1})^\omega$ be the word obtained from $w$ by setting the
second coordinate to $1$ at the positions from $X$ and to $0$ at the
other positions. We define $\ubouns(L)$ to be the set of those words
$w \in \Sigma^\omega$ such that for every $N \in \Nat$ there is a set
$X \subseteq \Nat$ of at least $N$ elements such that $w[X]$ belongs
to $L$.

Restated in terms of this operation, closure under unbounding quantification becomes:
\begin{prop}\label{proposition:U}
  Both~$\omega S$ and \ombs-regular languages are closed under $\ubouns$.
\end{prop}

We begin with a simple auxiliary result.  A \emph{partial sequence} over
an alphabet $\Sigma$ is a word in $\bot^*\Sigma^\omega$, in which~$\bot\not\in\Sigma$
is a fresh symbol.  A partial
sequence is \emph{defined} at the positions where it does not have
value~$\bot$, it is \intro{undefined} at  positions of value~$\bot$.
We say that two partial sequences \emph{meet} if there is some
position where they are both defined and have the same value.

\begin{lem}\label{lemma:infmeet}
  Let $I$ be an infinite set of partial sequences over a finite
  alphabet~$\Sigma$.  There is a partial sequence in $I$ that meets infinitely
  many partial sequences from $I$.
\end{lem}
\begin{proof}
  A \emph{constrainer} for $I$
  is an infinite word $c$ over $P(\Sigma)$ such that \hhl{for each $i
    \in \Nat$},  the $i$-th
  position of every sequence in $I$ is either undefined or belongs to
  $c_i$.  The size of a constrainer is the maximal size of a set it
  uses infinitely often.

  We prove the statement of the lemma
  by induction over the size of a constrainer for $I$.
  Since every $I$ admits a constrainer of size $|\Sigma|$,
  this concludes the proof.

  The base case is when $I$ admits a constrainer of size $1$;
  in this case, every two partial sequences in~$I$ meet,
  so any partial sequence in $I$\hhl{satisfies the statement of the lemma}.
  Consider now a set $I$ with a constrainer $c$ of
  size $n$. Take some sequence $s$ in $I$.  If $s$ meets infinitely
  many sequences from $I$, then we are done.  Otherwise let $J
  \subseteq I$ be the (infinite) set of sequences that do not meet
  $s$. \ttc{One} can verify that $d$ is a constrainer for $J$, where
  $d$ is defined by $d_i=c_i \setminus \set {s_i}$. Moreover, $d$ is
  of size $n-1$ (since $s_i = \bot$ can hold only for finitely many $i$). We then apply the induction hypothesis.
\end{proof}

Let $L$ be a language of infinite words over $\Sigma \times \zj$
recognized by an $\ombs$-automaton. We want to show that the language
$\ubouns(L)$ is also recognized by a bounding automaton.  Consider the
following language:
\begin{equation*}
  K = \set{w[X] :  \mbox{ for some $Y\supseteq X$, } w[Y] \in L}\ .
\end{equation*}
This language is downward closed in the sense that if $w[X]$ belongs
to $K$, then $w[Y]$ belongs to $K$ for every $Y \subseteq X$.
Furthermore, clearly $\ubouns(L)=\ubouns(K)$.  Moreover, if $L$ is
recognized by an $\ombs$-automaton (resp. $\omega S$-automaton), then
so is~$K$. Let \oldhl{then} $\Aa$ be an $\ombs$-automaton recognizing $K$. We will
construct an $\ombs$-automaton recognizing $\ubouns(K)$.

Given a word $w \in \Sigma^\omega$, \oldhl{we say that a sequence of sets $X_1,X_2,\ldots
\subseteq \nats$ is an \emph{an unbounding witness}} for $K$ if for
every $i$, the word $w[X_i]$ belongs to $K$ and the sizes of the \hhl{sets $X_i$}
are unbounded. An unbounding witness is \emph{sequential} if there is
a sequence of numbers $a_1<a_2<\cdots$  such that for \oldhl{each~$i$,
  all elements of the set  $X_i$ are between $a_i$ and $a_{i+1}-1$.}

The following lemma is a \oldhl{consequence of $K$ being downward closed. }
\begin{lem}\label{lemma:unboundig-sequential}
  A word \hhl{that has an unbounding witness for~$K$ also has a sequential one.}
\end{lem}

Let $X_1,X_2,\ldots$ be a sequential unbounding witness and let $a_1 <
a_2 \cdots$ be the appropriate sequence of numbers.  Let
$\rho_1,\rho_2,\ldots$ be accepting runs of the automaton $\Aa$ over the
words $w[X_1],w[X_2],\ldots$ Such runs exist by definition of
the unbounding witness. \oldhl{A  sequential witness  $X_1,X_2,\ldots$ is
  called a \emph{good
  witness}} if every two runs $\rho_i$ and $\rho_j$ agree on almost all
positions.

\begin{lem}
  A word belongs to $\ubouns(K)$ if and only if it admits a good witness.
\end{lem}
\begin{proof}
  By Lemma~\ref{lemma:unboundig-sequential}, a word belongs to
  $\ubouns(K)$, 
  if and only if it admits a sequential witness.  \oldhl{Let $X_i$, $a_i$
    and $\rho_i$ be as above.}
  For $i$, let $s_i$ be
  the partial sequence that has $\bot$ at positions before $a_{i+1}$
  and agrees with $\rho_i$ after $a_{i+1}$.  By applying
  Lemma~\ref{lemma:infmeet} to the set $\set{s_1,s_2\ldots}$, we can find
  a run $\rho_i$ and a set $J \subseteq \Nat$ such that for every $j
  \in J$, the runs $\rho_i$ and $\rho_j$ agree on some position $x_j$
  after $a_{j+1}$. For $j \in J$, let $\rho'_j$ be a run that is
  defined as $\rho_j$ at positions before $x_j$ and is defined as
  $\rho_i$ at positions after $x_j$. Since modifying the counter
  values over a finite set of positions does not violate the
  acceptance condition, the run $\rho'_j$ is also an accepting run
  over the word $w[X_j]$.  For every $j,k \in J$, the runs $\rho'_j$
  and $\rho'_k$ agree on almost all positions (i.e., positions after
  both $x_j$ and $x_k$).  Therefore the \oldhl{sequential} witness obtained by using only
  the sets $X_j$ with $j \in J$ is a good witness.
\end{proof}

\begin{lem}
  Words admitting a good witness can be recognized by a bounding
  automaton.
\end{lem}
\begin{proof}
  Given a word $w$, the automaton is going to guess a sequential
  witness
  \begin{equation*}
    a_1 < a_2 < \cdots \qquad X_1 \subseteq [a_1,a_2-1], X_2
    \subseteq [a_2,a_3-1]  \ldots
  \end{equation*}
  and a run $\rho$ of $\Aa$ over $w$ and verify the following
  properties:
  \begin{itemize}
  \item The run $\rho$ is accepting;
  \item There is no bound on the size of the $X_i$'s;
  \item For every $i$, some run over $w[X_i]$ agrees with $\rho$ on
    almost all positions.
  \end{itemize}
  The first property can be obviously verified by an 
  \ombs-automaton. For the second property, the automaton
  nondeterministically chooses a subsequence of $X_1,X_2,\ldots$ where
  the sizes are strongly unbounded. The third property is a regular
  property. The statement of the lemma then follows by closure of
  bounding automata under projection and intersection.
\end{proof}

\subsection{Bounds on the out-degree of a graph interpreted on sets}
\label{subsection:automatic}

In this section, somewhat disconnected from the rest of the paper,
we show how to use the logic MSOLB for solving a non-trivial
question concerning $\omega$-automatic structures.
An \emph{$\omega$-automatic (directed) graph (of injective
  presentation)} is a graph 
described by two formulas of MSOL, which
are interpreted in the natural numbers (hence the term
$\omega$-automatic).
The vertexes of the $\omega$-automatic graph are sets of natural numbers.
The first formula $\delta(X)$ has one free set
variable, and says which sets of natural numbers will be used as vertexes
of the graph. The second formula $\varphi(X,Y)$ has two free set
variables, and says which vertexes of the graph are connected by an edge
(if $\varphi(X,Y)$ holds, then both $\delta(X)$ and $\delta(Y)$ must
hold).  The original idea of automaticity has been proposed by Hodgson
\cite{hodgson83} via an automata theoretic presentation.
The more general approach that we use here of logically defining a structure
in the powerset of another structure is developed
in~\cite{colcombetloeding07}.
We show in this section that it is possible to decide, given the
formulas \oldhl{$\delta(X)$ and $\varphi(X,Y)$} of MSOL
defining an~$\omega$-automatic graph, whether this graph
has bounded out-degree or not.

Let $\varphi(X,Y)$ be a formula of MSOLB with two free set variables.
This formula can be seen as an edge relation on sets, i.e., it defines
a directed graph with sets as vertexes. We show here
that MSOLB can be used to say that this edge relation has unbounded
out-degree. The formula presented below will work on any
structure---not just $(\omega,<)$---but the decision
procedure will be limited to infinite words, since it requires
testing satisfiability of MSOLB formulas.

\oldhl{In the following,
we say that a set $Y$ is a \emph{successor} of a set $X$ if
$\varphi(X,Y)$ holds. } We begin by defining the notion of an $X$-witness. This is a
set witnessing that the set $X$ has many successors. (The actual successors of $X$ form a set of sets,
something MSOLB cannot talk about directly.) An \intro{$X$-witness} is a set
$Y$ such that every two elements $x,y \in Y$ can be separated by a
successor of $X$, that is:
\begin{equation}\label{eq:witness-formula}
  \forall x,y \in Y.~x\neq y \rightarrow \  \exists Z. \varphi(X,Z) \land (x \in Z \leftrightarrow y \not \in Z)\ .
\end{equation}
We claim that the graph of $\varphi$ has unbounded out-degree if and
only if there are $X$-witnesses of arbitrarily large cardinality.
This claim follows from the following fact:
\begin{fact}
  If $X$ has more than $2^n$ successors, then it has an $X$-witness of
  size at least $n$. If $X$ has~$n$ successors, then all $X$-witnesses
  have size at most $2^n$.
\end{fact}
\begin{proof}[Proof sketch]
  For the first statement, \oldhl{one first shows} that $X$ has at least $n$
  successors that are boolean independent, i.e., there exists~$X_1,\dots,X_n$
  successors of~$X$ such that for all~$i=1\dots n$, $X_i$ is not a boolean
  combination of~$X_1,\dots,X_{i-1},X_{i+1},\dots,X_n$.
  From $n$ boolean independent successors
  one can then construct by induction an $X$-witness of size $n$.

  For the second statement, consider $X$ with~$n$ successors
  as well as an $X$-witness.
  To each element~$w$ of the~$X$-witness, associate the
  characteristic function of~`$w\in Y$' for~$Y$ ranging over the
  successors of~$X$. If the $X$-witness had more than~$2^n$ elements,
  then at least two would give the same characteristic function,
  contradicting the definition of an~$X$-witness.
\end{proof}

\oldhl{As witnessed by}~\eqref{eq:witness-formula}\oldhl{,  being an $X$-witness can be defined by an }\ttc{MSOL}\oldhl{ formula.
Therefore, the} existence of arbitrarily large $X$-witnesses is expressible by an
MSOLB formula with a single $\ubouns$ operator at an outermost
position.  This formula belongs to one of the classes with decidable
satisfiability in Theorem~\ref{thm:dec-msob} below. This shows:
\begin{prop}
  It is decidable if an $\omega$-automatic graph has unbounded out-degree.
\end{prop}


\section{Complementation}
\label{section:complementation}

\subsection{The complementation result}

The main technical result of this paper is the following
 theorem:
\begin{thm}\label{theorem:complementation}
  The complement of an $\omega S$-regular language is~$\omega
  B$-regular, while the complement of an $\omega B$-regular language
  is~$\omega S$-regular.
\end{thm}

The proof of this result is long, and takes up the rest of this paper.
We begin by describing some proof ideas.  For the sake of this
introduction, we only consider the case of recognizing the complement
of an $\omega S$-regular language by an $\omega B$-regular automaton.

Consider first the simple case of a language described by an~$\omega
S$-automaton~$\Aa$ which has a single counter. Furthermore, assume
that in every run, between any two resets, the increments form a
single connected segment.  In other words, between two resets of the
counter, the counter is first left \oldhl{unaffected for some time}, then
during the $n$ following transitions the counter is always
incremented, then it is not incremented any more before reaching the
second reset. Below, we use the name \intro{increment interval} to
describe an interval of word positions (a set of consecutive word
positions) corresponding to a maximal sequence of increments. We do
not, however, assume that the automaton is deterministic \hhl{(the
whole difficulty of the complementation result comes from the fact that we are dealing
with nondeterministic automata). 
This means that the automaton resulting from the complementation
construction must check  that all
possible runs of the complemented automaton are rejecting.}

The complement $\omega B$-automaton~$\Bb$ uses a single $B$-counter,
which \hhl{ticks} as a clock along the $\omega$-word (independently of any
run of~$\Aa$).  A \hhl{tick} of the clock is a reset of the
  counter. Between every two \hhl{ticks}, the counter is constantly
  incremented. Since the counter is a $B$-counter, the \hhl{ticks} have to
  be at bounded distance from each other. We say that an interval of word positions is
\intro{short} (with respect to this clock) if it contains at most one
\hhl{tick} of the clock.  If the clock \hhl{ticks} \oldhl{at most} every \hhl{$N$} steps, then short
intervals have length at most $2N-1$.  Reciprocally, if an interval
has length at most~\oldhl{$N$}, then it is short with respect to every clock
\hhl{ticking} with a tempo greater than \oldhl{$N$}. According to these remarks,
being short is a fair approximation of the length of an interval.

The complement automaton~$\Bb$ works by guessing the \hhl{ticks} of a clock
using non\nobreakdash-\hspace{0pt}determinism together with a $B$-counter, and then checks the
following property: every run of~$\Aa$ that contains an infinite number of
resets, also contains an infinite number of short increment intervals.
Once the clock is fixed, checking this is definable in monadic
second-order logic, i.e., can be checked by a finite state automaton
without bounding conditions. Using this remark it is simple to
construct~$\Bb$.

It is easy to see that if~$\Bb$ accepts an $\omega$-word, then this
word is not accepted by~$\Aa$. The converse implication is a
consequence of the following compactness property: if no run of~$\Aa$
is accepting, then there is a threshold~$N \in \Nat$ such that 
every run of~$\Aa$ either has less than $N$ increments between two resets
infinitely often, or does finitely many resets.  Such a property can
be established using  Ramsey-like arguments.

Consider now a single counter $\omega S$-automaton, but without the
constraint of increments being performed during intervals. In our
construction, we use an algebraic decomposition result, Simon's
factorization forest theorem~\cite{simon}.  Using \oldhl{Simon's} theorem,
we reduce the \hhl{complementation problem} to a bounded number of instances of the above
construction.  In this case, the complement $\omega B$-automaton uses
one counter for each level of the factorization, the result being a
structure of nested clocks. As above, once the \hhl{ticks} of the clocks are
fixed, checking if a run makes few increments can be done by a finite
state automaton without any bounding conditions.

Finally, for treating the general case of $\omega S$-automata with
more than one counter, we use automata in their hierarchical form and
do an induction on the number of counters.

\subsection{Overview the complementation proof}
\label{sec:overview-compl}
Theorem~\ref{theorem:complementation} talks about complementing two
classes, and two proofs are necessary. The two proofs share a lot of
similarities, and we will try to emphasize common points.

From now on, a \intro{hierarchical automaton}~$\Aa$ \hhl{with
  states~$Q$, input
 alphabet $\Sigma$ and  counters~$\Gamma$ is fixed}.  Either all
the counters are of type~$S$ or all the counters are of type~$B$.  We
will denote this type by~$T \in \set{B,S}$, and by~$\bar T$ we will
denote the other type.  We say $f \in  \Nat^\Nat$ is a
\emph{$B$-function} if it is bounded, and an \emph{$S$-function} if it
is strongly unbounded.  We set~$\comp$ to be~$\leq$ if~$T=S$
and~$\geq$ if~$T=B$.  The idea is that greater means better. For
instance, if $n \comp m$, then replacing $n$ increments by $m$
increments leads to a run that is more likely to be accepting.

In this
terminology, a run of an $\omega T$-automaton is accepting
if its counter values \hhl{(at the moment before they are reset)} are
greater than some $T$-function for the order $\comp$. \hhl{When}
complementing an~$\omega T$-automaton our goal is to show that all
runs are rejecting, i.e., for every run,
either a counter is reset a finite number of times,
or infinitely often the number of
increments between two resets is smaller than a $\bar T$-function
\hhl{with respect to} the order~$\comp$.

Our complementation proof follows the scheme introduced by B\"uchi in
his seminal paper \cite{buchi60} (we also refer the reader to
\cite{revisited}). B\"uchi \ttc{establishes that} languages
recognized by nondeterministic B\"uchi automata are closed under
complement.  \hhl{In this proof B\"uchi did not determinize
the automata as usual in the finite case (this is impossible
for B\"uchi automata). Instead, he used a novel technique which allows
to directly construct a nondeterministic automaton for the complement of
a recognizable language. The key idea is that---thanks to
Ramsey's theorem---}each
$\omega$-word can be cut into a prefix followed by an infinite
sequence of finite words, \ttc{which are indistinguishable in any context by the automaton
(we also say that those words have same type)}.  Whether
or not the $\omega$-word is \oldhl{accepted} by that automaton depends
only on the type of the prefix and the type appearing in the infinite
sequence.  The automaton accepting the complement guesses this cut and
checks that the two types correspond to a rejected word.  For this
reason the proof can be roughly\ttc{ }decomposed into two parts. The
first part shows that each word can be cut in the way specified.  The
second part shows that a cut can be guessed and verified by a B\"uchi
automaton.

Our proof strategy is similar. That is why, in order to help the
reader gather some intuition, we summarize below B\"uchi's
complementation proof, in terms similar to our own proof for the
bounding automata. Then, in Section~\ref{sec:compl-prelim}, we outline
our own proof.

\subsubsection{The B\"uchi proof.}
In this section, we present a high-level overview of B\"uchi's
complementation proof.

A \emph{B\"uchi specification} \hhl{describes}  properties of a finite
  run of a given B\"uchi automaton. It is a positive boolean
combination of atomic \hhl{specifications}, which have three possible forms:
\begin{enumerate}
 \item The run begins with state $p$;
\item The run ends with state $q$;
\item The run contains/does not contain an accepting state.
\end{enumerate}
We will use the letter $\tau$ to refer to specifications, be they
B\"uchi, or the more general form introduced later on for bounding
automata.

The following statement shows the B\"uchi strategy for
complementation. The statement could be simplified for B\"uchi
automata, but we choose the presentation below to \ttc{stress} the
similarities with our own strategy, as stated in
Proposition~\ref{prop:our-main-compl}.
\begin{prop}\label{prop:main-compl}
  Let $\Bb$ be a B\"uchi automaton with input alphabet $\Sigma$. One
  can compute B\"uchi specifications $\tau_1,\ldots,\tau_n$ and
  regular languages $L_1,\ldots,L_n \subseteq \Sigma^*$ such that the
  following \hhl{statements} are equivalent for every infinite word $u \in
  \Sigma^\omega$:
  \begin{itemize}[leftmargin=7mm]
\item[(A)] The automaton $\Bb$ rejects the word $u$;
  \item[(B)] There is some $i=1,\ldots,n$ such that $u$ admits a
    factorization $u=w v_1 v_2 \cdots$ where:
    \begin{itemize}
    \item[(i)] The prefix $w$ belongs to $L_i$;
    \item[(ii)] For every $j$, every run $\rho$ \hhl{over} the finite word
          $v_j$ satisfies $\tau_i$.
    \end{itemize}
  \end{itemize}
  Moreover, for each $i=1,\ldots,n$, the language $K_i \subseteq
  \Sigma^*$ of words \oldhl{where} every run satisfies $\tau_i$ is regular.
\end{prop}

\oldhl{In the above statement, the prefix $w$ could also be described in
  terms of specifications, but we stay with just  the regular
  languages $L_i$ to
  accent the similarities with Proposition}~\ref{prop:our-main-compl}.

Proposition~\ref{prop:main-compl} immediately implies closure under complementation of
B\"uchi automata, since the complement of the language recognized by
$\Bb$ is the union
\begin{equation*}
  L_1 K_1^\omega + \cdots + L_n K_n^\omega\ ,
\end{equation*}
which is clearly \hhl{recognizable} by a B\"uchi automaton.

The ``Moreover..'' part of the proposition is \hhl{straightforward}, while the
equivalence of conditions (A) and (B) requires an application of
Ramsey's Theorem.

Our proof follows similar lines, although all parts require additional
work. In particular, our specifications will be more complex, and will
require more than just a regular language to be captured. The
appropriate definitions are presented in the following section.

\subsubsection{Preliminaries}
\label{sec:compl-prelim}

\newcommand{\cs}		{{\leq}}
\newcommand{\leqM}		{{\leq\!M}}
\newcommand{\lc}		{{\text-\hspace{-0.5mm}\circ}}
\newcommand{\rc}		{{\circ\hspace{-0.4mm}{\text-}}}
\newcommand{\mc}		{{\circ\hspace{-0.4mm}{\text-}\!\circ}}
\newcommand{\cnr}		{{\text-}}
\newcommand{\actions}		{\mathrm{Act}}
\newcommand{\otheractions}[1]	{\mathrm{Act}_{\neq{#1}}}
\newcommand{\constraints}	{\mathit{events}}
\newcommand{\sM}{{\comp\!M}}

\ign{
Given a finite run~$\rho$ (i.e., a sequence of transitions) and a
counter~$\alpha$, we denote by~$\rho|_\alpha$ the word in~$\{I,R\}^*$
that contains the actions of $\rho$ on the counter $\alpha$. This word
is obtained by mapping $I_\alpha$ to $I$, $R_\alpha$ to $R$ and
erasing the other positions.}

Below we define the \ttc{type} of a \hhl{finite} run. Basically, the type
corresponds to the information stored in a B\"uchi specification,
along with some information on the counter operations. The type
contains: 1) the source and target state of the run, \oldhl{like in a
  B\"uchi specification}; and 2) some
information on the counter operations that happen in the run.
Since we want to
have a finite number of types, we can only keep limited information on
the counter operations \ttc{(in particular}\oldhl{, we cannot keep track of the
  actual number of increments)}. Formally, a \ttc{\emph{type}}~$t$ is an element of:
\begin{eqnarray*}
  Q \times \set{\emptyset,
    \set{\cnr},\{\lc,\rc\},\{\lc,\mc,\rc\}} ^\Gamma
    \times Q\ .
\end{eqnarray*}
To a given run, we associate its type in the following way.  The two
states contain respectively the source (i.e.~first) and target (i.e.~last) state
of the partial run. \hhl{(A partial run is a finite run  that does not
  necessarily begin at the beginning of the word, and does not
  necessarily end at the end of the word.)}   The middle component associates to each counter
$\alpha \in \Gamma$ a \emph{counter profile} denoted---by slight
abuse---$t(\alpha)$.  The counter profile \oldhl{is $t(\alpha)=\emptyset$} when
there is no increment nor reset \oldhl{on counter $\alpha$}.  The counter profile
\oldhl{is $t(\alpha)=\set{\cnr}$}  when  the counter  is incremented but not reset.
The counter profile  \oldhl{is $t(\alpha)=\set{\lc,\rc}$ when} the counter is reset just
once, while \oldhl{$t(\alpha)=\set{\lc,\mc,\rc}$} is used for the other cases, when the
counter is reset at least twice.

Note that the counter profiles are themselves sets, and elements of
these sets have a meaningful interpretation.  Graphically\ttc{, each symbol among}
$\cnr,\lc,\rc,\mc$ represents a possible kind of sequence of
increments of a given counter.  The circle~$\circ$ \ttc{symbolizes a
reset starting or ending the sequence, while the dash~- represents the sequence of increments
itself.}
\ttc{For instance, $\lc$ identifies the segment that starts at the beginning of the run
and end at the first occurrence of a reset of the counter.} Given a type $t$, we use the
name \emph{$t$-event} for any pair
\begin{equation*}
  (\alpha,c) \in \Gamma \times \set{\cnr,\lc,\mc,\rc}\,,
  \qquad \text{ with }c \in t(\alpha)\ .
\end{equation*}
The set of $t$-events is denoted by $\constraints(t)$.  Given a
$t$-event $(\alpha,c)$ and a finite run $\rho$ of type $t$, the value
$\val(\rho,\alpha,c)$ is the natural number defined below:

\medskip\noindent
\begin{tabular}{ll}
  $\val(\rho,\alpha,\cnr)$ &  the number of increments on counter $\alpha$ in the run.\\
  $\val(\rho,\alpha,\lc)$  & the number \oldhl{of} increments on counter $\alpha$
  before the first reset of counter $\alpha$.\\
  $\val(\rho,\alpha,\rc)$ & the number of  increments \oldhl{on} counter
  $\alpha$  after the last reset of counter $\alpha$.\\
  $\val(\rho,\alpha,\mc)$ &  the minimal (with respect to $\comp$)
  number of increments on counter $\alpha$  \\ &  between two consecutive
  resets on counter $\alpha$.
\end{tabular}
\medskip

We comment on the last value. When $T=S$, $\val(\rho,\alpha,\mc)$ is
the smallest number of increments on counter $\alpha$ that is done
between two successive resets of $\alpha$.
When $T=B$, this is the
largest number of increments. At any rate, this is the worst number
of increments, as far as the acceptance condition is concerned.

A \emph{ specification} is a property of finite runs. It is a positive
boolean combination of the following two kinds of \emph{atomic
  specifications}:
\begin{enumerate}
\item The run has type $t$. 
\item\label{item:compare} The run satisfies $\val(\rho,\alpha,c) \comp
  K$. 
\end{enumerate}
\hhl{The first atomic specification is defined by giving $t$. The
  second atomic specification is defined by giving $\alpha$ and $c$,
  but not $K$, which remains undefined. 
The  number $K$  is treated as a special parameter}, or free variable, of
the specification. This parameter is shared by all atomic specifications.  Given a value of $K \in \Nat$, we define the
notion that a run $\rho$ \intro{satisfies a specification $\tau$ under~$K$} in
the natural manner.

Unlike the B\"uchi proof, it is important that the boolean combination
in the specification is positive.  The reason is that we will use
$\bar T$-automata to complement $T$-automata, and therefore we only
talk about one type of behavior (bounded, or strongly unbounded).

\subsubsection{The decomposition result}
In this section we present the main decomposition result, which yields
Theorem~\ref{theorem:complementation}.

\begin{prop}\label{prop:our-main-compl}
  For every $\omega T$-automaton $\Aa$ \ttc{one can effectively obtain regular languages $L_1,\ldots,L_n$ 
  and specifications $\tau_1,\ldots,\tau_n$ such that}
  the following \hhl{statements} are equivalent for every $\omega$-word $u$\ttc{:}
  \begin{itemize}[leftmargin=7mm]
  \item[(A)] The automaton $\Aa$ rejects $u$;
  \item[(B)] There is some $i=1,\ldots,n$ such that $u$ admits a
    factorization $u=w v_1 v_2 \cdots$ where:
    \begin{itemize}
    \item[(i)] The prefix $w$ belongs to $L_i$, and;
    \item[(ii)] There is a $\bar T$-function $f$ such that for every
          $j$, every run $\rho$ \hhl{over} $v_j$ satisfies $\tau_i$ under
          $f(j)$.
    \end{itemize}
  \end{itemize}
  Moreover, for each $i=1,\ldots,n$, one can verify with a \hhl{word}
  sequence $\bar T$-automaton $\Bb_i$ if a sequence of
  words $v_1,v_2,\ldots$ satisfies condition (ii).
  (Equivalently, the set of \oldhl{word sequences} satisfying (ii) is $\bar T$-regular.)
\end{prop}

First we note that the above proposition implies
Theorem~\ref{theorem:complementation}, since property (B) can be
recognized by an $\omega \bar T$-automaton.  This is thanks to the
``Moreover...'' and the closure of
$\omega \bar T$-automata under finite
union and the prefixing of a regular language.

The rest of this paper is devoted to showing the proposition.  In
Sections~\ref{subsection:ramsey} and~\ref{subsection:description} we
show the first part, i.e., the equivalence of (A) and (B).
Section~\ref{subsection:ramsey} develops extensions of Ramsey's
theorem.  Section~\ref{subsection:description} uses these results to
show the equivalence.  In
Sections~\ref{subsection:verifying_single_constraints}
and~\ref{subsection:verifying-weak-description} we prove the
``Moreover...''  part.
Section~\ref{subsection:verifying_single_constraints} contains the
construction for a single counter, while
Section~\ref{subsection:verifying-weak-description} \hhl{extends} this
construction to multiple counters. The difficulty in the ``Moreover...''
part is that (ii) talks about ``every run $\rho$'', and therefore the
construction has to keep track of many simultaneous runs.


\subsection{Ramsey's theorem and extensions}
\label{subsection:ramsey}

Ramsey-like statements (as we consider them in our context) are
statements of the form \emph{``there is an infinite
  set~$D\subseteq\Nat$ and some index~$i \in I$ such that the property
  $P_i(x,y)$ holds for any~$x<y$ in~$D$''}. This statement is relative
to a family of properties $\set{P_i}_{i \in I}$. In
general, the family~$\set{P_i}_{i \in I}$ may be infinite.
The classical theorem of \oldhl{Ramsey, as
stated below, follows this scheme, but for  a family of two properties: $P_1=R$ and} \ttc{$P_2= \Nat^2\setminus R$.}
\begin{thm}[Ramsey]
\oldhl{  Given~$R\subseteq \Nat^2$ and an infinite set $E \subseteq \Nat$,  there is an
  infinite set~$D\subseteq E$
  such that;}
\begin{itemize}
\item for all~$x<y$ in~$D$, $(x,y)\in R$, or;
\item for all~$x<y$ in~$D$, $(x,y)\not\in R$.
\end{itemize}
\end{thm}
\oldhl{In the original statement of Ramsey's
theorem, the set $E$ is not used (i.e.~$E=\Nat$).  We use the more
general, but obviously equivalent, formulation to \emph{compose} Ramsey-like
statements as follows.}  Assume that
there  are two Ramsey-like statements, one using
properties~$\set{P_i}_{i \in I}$, and the other using $\set{Q_j}_{j
  \in J}$. We can apply the two in cascade and obtain a new statement
of the form ``there is an infinite set~$D\subseteq\Nat$ and
indexes~$i\in I,j \in J$ such that both $P_i(x,y)$ and~$Q_j(x,y)$ hold
for any~$x<y$ in~$D$''.  This is again a Ramsey-like statement.  This
composition technique is heavily used below.  We will simply refer to
it as the \emph{compositionality of Ramsey-like statements} and
shortcut the corresponding part of the proofs.

The following lemma is our first Ramsey-like statement which uses an
infinite (even uncountable) number of properties.
\begin{lem}\label{lemma:ramsey-bounds}
  For any $h: \Nat^2 \to \Nat$ there is an infinite
  set~$D\subseteq\Nat$ such that either
\begin{itemize}
\item \oldhl{there} is a natural number~$M$ such that $h(x,y)\leq M$ holds for
  all~$x<y\in D$, or;
\item \oldhl{there} is an S-function~$g$ such that $h(x,y)> g(x)$ holds for
  all~$x<y\in D$.
\end{itemize}
\end{lem}

Note that in the above statement, only values $h(x,y)$ for $x < y$ are
relevant. This will be the case in the other Ramsey-like statements below.

\begin{proof}
  By induction we construct a sequence of sets of natural numbers $
  D_0 \supseteq D_1 \supseteq \cdots$.  The set~$D_0$ is \oldhl{defined
    to be $\Nat$. For $n>0$, the set 
  $D_{n}$ is defined to be  the infinite set $D$ obtained by applying Ramsey's
  theorem to~$E=D_{n-1}\setminus \set{\min D_{n-1}}$ with the binary property
  $R=h(x,y)\leq n$.}

  Two cases may happen. 
  Either for some~$n$, \oldhl{the value of  $h(x,y)$ is at most $n$ for
    all~$x<y$ taken from $D_n$}. In this case the first disjunct in the conclusion
  of the lemma holds (with~$D=D_n$ and $M=n$).
  Otherwise, \oldhl{for all $n$, the value  $h(x,y)$ is greater than $n$
    for all~$x<y$ taken from  $D_n$}.  In this
  case, we set~$D$ to be $\{\min D_i\,:\,i\in\Nat\}$
  and $g$ to satisfy \oldhl{$g(\min D_i)=i$.}
  The second conclusion of the lemma
  holds.
\end{proof}

\begin{defi}
  A \intro{separator} is a pair $(f,g)$ where~$f$ is a~$\bar
  T$-function and~$g$ is a~$T$-function.
\end{defi}

The following lemma restates the previous one in terms of separators.

\begin{lem}\label{lemma:ramsey-B}
  For any $h: \Nat^2 \to \Nat$ there is an infinite
  set~$D\subseteq\Nat$ and a separator~$(f,g)$ such that either;
\begin{itemize}
\item for all~$x<y\in D$, $h(x,y)\comp f(x)$, or;
\item for all~$x<y\in D$, $h(x,y)\notcomp g(x)$.
\end{itemize}
\end{lem}

\hhl{Lemma~\ref{lemma:ramsey-C} below generalizes
Lemma~\ref{lemma:ramsey-B}: instead of having a single
element $h(x,y)$ for each $x < y$, we have a set $E_{x,y}$ of vectors. The conclusion of
the lemma describes which components of the input vectors from
$E_{x,y}$ are bounded or unbounded simultaneously. }

When applied to complementing automata, \ttc{each set $h(x,y)$ will
gather information relative to the possible runs of the}
automaton over the part of a word that
begins in position $x$ and ends in position $y$. Since the automaton
is nondeterministic, $h(x,y)$ contains not a single element, but a set
of elements, one for each \ttc{possible} run. Since the automaton has many counters,
and a counter may come with several events, elements of $h(x,y)$ are
vectors, with each coordinate corresponding to a single event.

Before stating the lemma, we introduce some notation.  Let~$C$ be a
finite set, which will be used for coordinates in vectors.  Given a
vector~$v\in\Nat^C$, a set of coordinates~$\sigma\subseteq C$ and a
natural number~$M$, the expression $v\comp_\sigma M$ means that
$v(\alpha)\comp M$ holds for all coordinates~$\alpha\in\sigma$. We use
a similar notation for~$\notcomp$, i.e.~$v \notcomp_\sigma M$ means
that $v(\alpha) \notcomp M$ holds for all coordinates $\alpha \in
\sigma$\ .  Note that~$\not\comp_\sigma$ is different
from~$\notcomp_\sigma$. The first says that $\notcomp_{\{\alpha\}}$
holds for some coordinate~$\alpha\in\sigma$, while the second says
that $\notcomp_{\{\alpha\}}$ has to hold for all
coordinates~$\alpha\in\sigma$.

\begin{lem}\label{lemma:ramsey-C}
  Let $C$ be a finite set, and for every natural numbers~$x<y\in\Nat$,
  let $E_{x,y} \subseteq \Nat^C$ be a finite nonempty set of vectors.
  There is a family of coordinate sets
  $\Theta\subseteq\mathcal{P}(C)$, an infinite set~$D\subseteq\Nat$
  and a separator $(f,g)$ such that for all~$x<y$ in~$D$,
\begin{itemize}
\item[(1)] for all $v\in E_{x,y}$, there is a coordinate
  set~$\sigma\in\Theta$ such that $v\comp_\sigma f(x)$, and;
\item[(2)] for all coordinate sets $\sigma\in\Theta$, there is $v\in
  E_{x,y}$ such that $v\comp_\sigma f(x)$ and
  $v\notcomp_{C\setminus \sigma}g(x)$.
\end{itemize}
\end{lem}
\begin{proof}
  If we only take~(1) into account, we can see~$\Theta$ as a
  disjunction of conjunctions of boundedness constraints, i.e., a DNF
  formula. The property (1) says that each vector satisfies one of the
  disjuncts.  Keeping this intuition in mind, for two coordinate sets
  $\sigma,\sigma' \subseteq C$ we write~$\sigma\Rightarrow \sigma'$
  if~$\sigma\supseteq\sigma'$.  Given two families of coordinate sets
  $\Theta,\Theta' \subseteq \mathcal P (\mathcal C)$, we
  write~$\Theta\Rightarrow\Theta'$ if for every disjunct~$\sigma
  \in\Theta$ there is a disjunct~$\sigma'\in\Theta'$ such
  that~$\sigma\Rightarrow\sigma'$.  This notation corresponds to
  the following property: if (1) holds for~$\Theta$
  and~$\Theta\Rightarrow\Theta'$, then~(1) also holds
  for~$\Theta'$.  There is a minimum element for this preorder which is $\emptyset$
  (the empty disjunction, equivalent to false),
  \hhl{and} a maximum element~$\{\emptyset\}$
  (a single empty conjunction, equivalent to true).
  The preorder~$\Rightarrow$ induces an equivalence
  relation~$\Leftrightarrow$, which corresponds to logical equivalence
  of DNF formulas.

  Let $\Theta \subseteq \mathcal{P}(C)$ be a nonempty family of
  coordinate sets. For two natural numbers $x<y$, we define
  $h_\Theta(x,y) \in \Nat$ to be
\begin{align*}
h_\Theta(x,y)&={\min}_\comp\left\{M~:~\forall~v\in E_{x,y}.~
	\exists~\sigma\in\Theta.
	~v\comp_\sigma M\right\}\ .
\end{align*}
By applying Lemma~\ref{lemma:ramsey-B} to the
property~$h_\Theta(x,y)$, we obtain an infinite set~$D \subseteq \Nat$
and a separator $(f,g)$ such that either;
\begin{itemize}
\item[(a)] for all~$x<y\in D$, $h_\Theta(x,y)\comp f(x)$,\\
	i.e., for all~$v\in E_{x,y}$ there exists~$\sigma\in\Theta$ such that~$v\comp_\sigma f(x)$,
	or;
\item[(b)] for all~$x<y\in D$, $h_\Theta(x,y)\notcomp g(x)$,\\
	i.e., there exists~$v\in E_{x,y}$ such that for all~$\sigma\in\Theta$, $v\not\comp_\sigma g(x)$.
\end{itemize}
Using compositionality of Ramsey-like statements, we can assume that
the separator $(f,g)$ works for all possible families $\Theta$
simultaneously. Note however, that the choice of item (a) or (b) may
depend on the particular family $\Theta$.
Furthermore remark that if to~$\Theta$ corresponds property~(a),
and $\Theta\Rightarrow\Theta'$ \hhl{holds,
then property~(a) also corresponds to~$\Theta'$.}
By removing a finite
number of elements of~$D$, we can further assume that for any~$x$
in~$D$ we have $f(x)\comp g(x)$.

Let~$\Theta \subseteq \mathcal{P}(C)$ be a family of coordinate sets
that satisfies property (a), but is minimal in the sense that for
every family $\Theta'$ satisfying (a), the implication $\Theta
\Rightarrow \Theta'$ holds.  The family $\Theta$ exists since
$\{\emptyset\}$ satisfies (a) for any separator~$(f,g)$ (it is even
unique up to~$\Leftrightarrow$).  Without loss of generality, we
assume that $\Theta$ does not contain two coordinate sets $\sigma'
\subseteq \sigma$, since otherwise we can remove the larger coordinate
set $\sigma$ and still get an equivalent family.

We will show that the properties~(1) and~(2) of the lemma hold for
this family $\Theta$. Property~(1) directly comes from~(a). Let us prove (2).
Fix a coordinate set~$\sigma$ in $\Theta$,
as well as $x<y$ in $D$.  We need to show that for some~$v\in
E_{x,y}$, both $v\comp_\sigma f(x)$ and $v\notcomp_{C\setminus\sigma}g(x)$.
Let $\Theta'$ be the family of coordinate sets obtained from
$\Theta$ by removing $\sigma$ and adding all coordinate sets of the
form $\sigma \cup \set \beta$, for $\beta$ ranging over $C\setminus\sigma$.
It is easy to see that $\Theta' \Rightarrow \Theta$,
while $\Theta \iff \Theta'$ does not hold. In particular, by minimality
of $\Theta$, the family $\Theta'$ cannot satisfy property (a), and
therefore it must satisfy property (b).  Let~$v$ be the vector
of~$E_{x,y}$ existentially introduced by property~(b) applied
to~$\Theta'$.  We will show that this vector satisfies both
$v\comp_\sigma f(x)$ and $v\notcomp_{C\setminus\sigma}g(x)$.

First, we show $v \comp_\sigma f(x)$. Since the
family~$\Theta$ satisfies property~(a), there must be a coordinate
set~$\sigma'\in \Theta$ such that~$v\comp_{\sigma'}f(x)$.  We claim
that $\sigma'=\sigma$.  Indeed, otherwise $\sigma'$ would belong to
$\Theta'$, and by~(b) we would have $v\not\comp_{\sigma'} g(x)$.
This is in contradiction with~$v\comp_{\sigma'}f(x)$
since $f(x)\comp g(x)$.

Second, we show $v\notcomp_{C\setminus\sigma}g(x)$. Let
then~$\beta$ be a coordinate in $C\setminus\sigma$.  By definition
of~$\Theta'$, the coordinate set $\sigma\cup\{\beta\}$ belongs to
$\Theta'$.  By~(b) for~$\Theta'$, we get $v(\alpha)\notcomp g(x)$ for
some~$\alpha\in \sigma\cup\{\beta\}$.  Since above we have
shown~$v\comp_\sigma f(x)$ and~$f(x)\comp g(x)$, it follows
that~$\alpha=\beta$ and~$v(\beta)\notcomp g(x)$.  Since
this is true for all~$\beta\in C\setminus\sigma$,
we have~$v\notcomp_{C\setminus\sigma}g(x)$.
\end{proof}

\subsection{Descriptions}
\label{subsection:description}

Recall the mapping $val(\rho,\alpha,c)$, which \hhl{describes} the number of
increments that run $\rho$ does on counter $c$ in the event
$\alpha$. By fixing a run $\rho$ of type $t$, we can
define~$val(\rho)$ as a mapping from the set of
$t$-events~$\constraints(t)$ to~$\Nat$, i.e., a vector of natural
numbers.  This vector measures the number of increments in the
run~$\rho$ for each event and each counter, keeping track of the
numbers which are the \emph{worst} for the acceptance condition.  Note
that the coordinates of $val(\rho)$ depend on the type $t$ of
$\rho$---more precisely on $\constraints(t)$---and therefore the
vectors $val(\rho)$ and $val(\rho')$ may not be directly comparable
for runs $\rho$, $\rho'$ of different types.  The key property of
$val$ is that given a sequence of finite runs~$\rho_1,\rho_2, \ldots$,
it is sufficient to know their respective \hhl{types~$t_1,t_2,\dots$ and} the
vectors $val(\rho_1),val(\rho_2),\ldots$ in order to decide whether or
not the infinite run $\rho_1 \rho_2 \ldots$ satisfies the acceptance
condition.  The descriptions defined below gather this information in
a finite object.

\begin{defi}[description]
  A \intro{description} is a set of pairs~$(t,\gamma)$ where~$t$ is a
  type and~$\gamma$ a set of $t$-events.
\end{defi}
\oldhl{The intuition is that $\gamma$ is the set of events where the
  counter values are bad for the automaton's acceptance condition,
  i.e..~small in the case when $T=S$ and large in the case when
  $T=B$.}

A \emph{cut} is any infinite set of natural numbers~$D$, which are
meant to be word positions.  We also view a cut as a sequence of
natural numbers, by ordering the numbers from $D$ in increasing order.
Given a cut $D=\set{d_1 < d_2 < \cdots}$ and an $\omega$-word $w \in
\Sigma^\omega$, we define $w|_D$ to be the infinite sequence of finite
words obtained by cutting the word at all positions in~$D$:
\begin{equation*}
w|_D =   w[d_1,\ldots,d_2-1],~w[d_2,\ldots,d_3-1],~\ldots
\end{equation*}
\hhl{Note that the prefix $w[0,\ldots,d_1-1]$  of $w$ up to position
  $d_1-1$  is not used here; } as in the
proof of B\"uchi, it is treated separately.

\begin{defi}[strong description]
  Let $w \in \Sigma^\omega$ be an  input  $\omega$-word, $\tau$ a
  description and~$D \subseteq \Nat$ a cut.  We say that $\tau$
  \intro{strongly describes}~$w|_D$ if there is a separator~$(f,g)$
  such that for every~$x<y\in D$ \hhl{the following conditions hold.}
\begin{itemize}
\item for every partial run~$\rho$ of type~$t$ \hhl{over} $w$ from position~$x$
  to position~$y$, there is a pair $(t,\gamma)\in\tau$ such
  that~$val(\rho)\comp_{\gamma}f(x)$, and;
\item for every pair~$(t,\gamma)\in\tau$, there is a run~$\rho$ \hhl{over} $w$
  of type $t$ from position~$x$ to position~$y$ such \hhl{that~$val(\rho)\comp_{\gamma}f(x)$
  and~$val(\rho)\notcomp_{\constraints(t)\setminus\gamma} g(x)$.}
\end{itemize}
\end{defi}

\begin{lem}\label{lemma:strong-description}
  For every $\omega$-word~$w \in \Sigma^\omega$, there is a
  description $\tau$ and a cut $D$ such that $\tau$ strongly
  describes~$w|_D$.
\end{lem}
\begin{proof}
  Using Ramsey's theorem and its compositionality, we find a cut $D_0$
  and a set of types $A$ such that for all $x<y$ in $D_0$ the
  following conditions are equivalent;
\begin{itemize}
\item there is a partial run of type $t$ between positions $x$ and $y$, and;
\item the type $t$ belongs to $A$.
\end{itemize}
The rest follows by applying Lemma~\ref{lemma:ramsey-C} for all types~$t$
in~$A$, and using  compositionality of Ramsey-like statements.
\ign{
Fix a type~$t$ in~$A$.  For each $x<y \in D_0$ we set~$E^t_{x,y}$ to
be the set of vectors~$val(\rho)$, with~$\rho$ ranging over finite
run of type~$t$ that start in position~$x$ and end in position~$q$.
Using Lemma~\ref{lemma:ramsey-C} and compositionality of Ramsey-like
statements, we obtain an infinite cut~$D\subseteq D_0$, a
separator~$(f,g)$ and a family $\set{\Theta_t}_{t \in A}$ such that
for any type $t$ in~$T$, the conclusions of Lemma~\ref{lemma:ramsey-C}
hold with~$E_{x,y}$ replaced by~$E^t_{x,y}$ and~$\Theta$ replaced
by~$\Theta_t$.  It is not difficult to see that the description
$\{(t,\gamma)~:~t\in A,~\gamma\in\Theta_t\}$ strongly describes
$w|_D$.}
\end{proof}

The problem with strong descriptions is that an $\omega\bar
T$-automaton cannot directly check if a description strongly describes
some cut word~$w|_D$. There are two reasons for this. First, we need
to check the description for each~$x<y$ in~$D$, and therefore deal
with the overlap between the finite words $w[x,\ldots,y]$. Second, the
second condition of strong description involves guessing the
$T$-function~$g$, and this cannot be done using a~$\omega\bar
T$-automaton.

Hence, we reduce the property to checking weak descriptions
(see Definition~\ref{df:weak-description}) which are
stable (see Definition~\ref{df:stability}). Lemma~\ref{lemma:stability}
shows that this approach makes sense.

\begin{defi}[weak description]\label{df:weak-description}
  Given a type $t$, a set of events $\gamma$ and a natural number~$N$,
  a finite run~$\rho$, is called \intro{consistent with~$(t,\gamma)$
    under~$N$} if~$\rho$ has type~$t$ and $\val(\rho)\comp_\gamma N$.
  Given a description~$\tau$, a run $\rho$ is called \intro{consistent
    with~$\tau$ under~$N$} if it is consistent with
  some~$(t,\gamma)\in\tau$ under~$N$.

  Let $w \in \Sigma^\omega$ be an  input $\omega$-word, $\tau$ a
  description and~$D$ a cut.  We say that
  $\tau$ \intro{weakly describes}~$w|_D$ if there is
  a~$\bar T$-function~$f$
  such that for all~$i\in\Nat$, every run~$\rho$ \oldhl{over the $i$-th
    word in  $w|_D$}
  is consistent  with~$\tau$ under~$f(i)$.
\end{defi}
\hhl{A weak} description is a weakening of strong descriptions for two
reasons: only runs between consecutive elements of the cut are
considered, and only the first constraint of the strong description is kept.

The following lemma shows that weak descriptions can be \hhl{expressed} by
specifications.
\begin{lem}\label{lemma:descriptions-simulated-by-specifications}
  For every weak description there is an equivalent specification. In
  other words, for every description $\tau$ there is a specification
  $\tau'$ such that the following are equivalent for an
  $\omega$-word~$w$ and a cut $D$\,;
  \begin{itemize}
  \item the description $\tau$ weakly describes $w|_D$, and;
  \item there is a $\bar T$-function $f$ such that for all  $i \in \Nat$, every
    run $\rho$ over \oldhl{the $i$-th word in $w|_D$ satisfies $\tau'$ under $f(i)$}.
  \end{itemize}
\end{lem}
\begin{proof}
  All the conditions in the definition of weak descriptions can be
  expressed by\hhl{ a specification.}
\end{proof}

In Definition~\ref{df:stability}, we present the notion of a stable
description. The basic idea is to mimic the notion of idempotency
used in the case of B\"uchi.

\begin{defi}[stability]\label{df:stability}
  A description~$\tau$ is \intro{stable} if there is a \hhl{$T$-function}
  $h$ such that for all natural numbers $N$ and all finite
  runs~$\rho=\rho_1\dots\rho_k$, if each $\rho_i$ is consistent
  with~$\tau$ under~$N$ for all~$i=1\dots k$, then $\rho$ is
  consistent with~$\tau$ under~$h(N)$.
\end{defi}

To illustrate the above definition, we present an example. In this
particular example, the description will be stable for $T=B$, but it
will not be stable for $T=S$.  \oldhl{Let then $T=B$ and consider} an automaton with one counter
$\alpha=1$, and one state $q$. We will show that for
$t=(q,\set{\cnr},q)$ \hhl{(we write $\{\cnr\}$ instead of  the mapping which to counter~$1$ associates~$\{\cnr\}$)}
and $\gamma=\set{(1,\cnr)}$, the description
$\tau=\set{(t,\gamma)}$ is stable.  In the case of $\tau$, the
function $h$ from the definition of stability will be the identity
function $h(N)=N$. To show stability of $\tau$, consider any finite
run $\rho$ decomposed as $\rho_1 \cdots \rho_k$, with each $\rho_i$
consistent with $\tau$ under $N$.  To show stability, we need to show
that
\begin{eqnarray*}
  \val(\rho) \comp_\gamma h(N)=N \ .
\end{eqnarray*}
Since $T=B$, the relation $\comp$ is $\ge$, so we need to show that
$\val(\rho)$ is at least $N$ on all events in $\gamma$. Since $\gamma$
has only one event $(1,\cnr)$, this \hhl{boils} down to proving $\val(\rho
,1,\cnr) \ge N$. But this is simple, because
\begin{eqnarray*}
  \val(\rho
,1,\cnr)  = \val(\rho_1,1,\cnr) + \cdots + \val(\rho_k,1,\cnr)  \ ,
\end{eqnarray*}
and each $\rho_i$ satisfies $\val(\rho_i,1,\cnr) \ge N$ by
assumption. Note that this reasoning would not go through with $T=S$,
which is the reason why the above description is stable only in the
case $T=B$.

\medskip We now show that strong descriptions are necessarily stable.
\begin{lem}\label{lemma:stability}
  If~$\tau$ is a strong description of some~$w|_D$ then $\tau$ is
  stable.
\end{lem}
\begin{proof}
  We \oldhl{prove the statement for~$T=S$ first, and then for~$T=B$.}

  {\bf Case~$T=S$.} In this case~$\comp$ is~$\leq$.
\oldhl{Let~$\rho_1,\dots,\rho_k$ be  runs such that~$\rho_i$ is
  consistent with~$(t_i,\gamma_i)\in\tau$ under~$N$ for
  all~$i=1,\dots,k$. We need to show that the composition
  $\rho=\rho_1\cdots \rho_k$ of these runs is consistent with some
  $(t,\gamma) \in \tau$ under $h(N)$, for some $S$-function $h$ independent of
  $\rho_1,\ldots,\rho_k$. The function $h$ will be a  linear function,
  with the linear constant taken from the assumption that $\tau$ strongly
  describes $w|_D$.  Let then~$(f_0,g_0)$ be}\ttc{ the }\oldhl{separator
  obtained by unraveling the definition of  $w|_D$ being strongly
  described by~$\tau$.} We can
  assume without loss of generality that~$f_0$ is constant, equal
  to~$M$. Since~$g_0$ tends toward infinity, we can chose a natural
  number $n$ such that $g_0(i)\geq M$ holds for all~$i \ge n$.

  We now mimic each run~$\rho_i$ by a similar run~$\pi_i$
  over \oldhl{the}~$(i+n)$-th word in the sequence $w|_D$.  By definition of a
  strong description, one can find for~$i=1,\dots,k$ a run~$\pi_i$
  over the $(i+n)$-th word in $w|_D$ such that $\val(\pi_i)\leq_{\gamma_j}M$
  and $val(\pi_i)>_{\constraints(t_i)\setminus\gamma_i}g_0(i+n)$.
  From the last inequality together with~$g_0(i+n)\geq M$, we obtain:
  \begin{align*}
  \text{for all}~(\alpha,c)\in\constraints(t_i),\qquad\val(\pi_i,\alpha,c)\leq M\quad\text{implies}\quad(\alpha,c)\in\gamma_i\,.
  \tag{$\sharp$}
  \end{align*}
  Let $\pi$ be~$\pi_1\dots\pi_k$.
  Since~$\tau$ strongly describes~$w|_D$,
  the run $\pi$ is consistent with~$(t,\gamma)$ under~$M$
  for some~$(t,\gamma)$ in~$\tau$. Formally,
  \begin{align*}
  \val(\pi)\leq_\gamma M\qquad(\sharp2) \ .
  \end{align*}

  We will show that $\rho$ is consistent with~$(t,\gamma)$
  under~$N(M+2)$, which establishes the stability of~$\tau$, under the
  linear function $h(N)=N(M+2)$.  Let then~$(\alpha,c)$ be an event
  in~$\gamma$.  We have to prove
  \begin{align*}
  \val(\rho,\alpha,c)\leq N(M+2).
  \end{align*}
  This is done by a case distinction depending on~$c$.
\begin{enumerate}[leftmargin=13mm]
\item[$c=\cnr$] This means that $\pi$
      does not contain any reset of~$\alpha$.
      Let~$I \subseteq \set{1,\ldots,k}$ be the set
      of \oldhl{indexes}~$j$ for which $\pi_j$
      contains an increment of counter $\alpha$.
      By~($\sharp2$), and since~$(\alpha,\cnr)$ belongs to $\gamma$,
      the number of increments of counter $\alpha$ in $\pi$
      is at most $M$.
      In particular\oldhl{, the set $I$ contains at most $M$ indexes.}

      Let now~$i\in I$. Still by~($\sharp2$), the run~$\pi$, \oldhl{and
        hence also $\pi_i$,}
      contains at most~$M$ increments of~$\alpha$.
      By ($\sharp$), this means that~$(\alpha,\cnr)$ belongs to~$\gamma_i$.
      Hence, since~$\rho_i$ is consistent with~$(t_i,\gamma_i)$
      under~$N$, $\rho_i$ contains at most~$N$ increments of~$\alpha$.
      Furthermore, for~$i\not\in I$, $\pi_i$ does not increment~$\alpha$,
      and~$\pi_i$ has the same type as~$\rho_i$. Hence~$\rho_i$ does not
      increment the counter~$\alpha$ either. Summing up: at most $M$
      runs among $\rho_1,\ldots,\rho_k$ increment counter $\alpha$,
      and those that do, do so at most $N$ times.

      Overall there are at most~$MN$ increments of~$\alpha$ in~$\rho$,
      i.e., $\val(\rho,\alpha,c)\leq MN$.
    \smallskip
    \item[$c=\lc$] Let~$j$ be the first index such that~$\rho_j$
      contains a reset of~$\alpha$.
      Using the previous case, but for $k=j-1$, we infer that the
      prefix~$\rho_1\dots \rho_{j-1}$ contains at most~$MN$
      increments of~$\alpha$.

      By ($\sharp2$), there are at most~$M$ increments of~$\alpha$
      before the first reset in~$\pi$. Since this
      first reset occurs in~$\pi_j$, the same holds for~$\pi_j$.
      By~($\sharp$) we obtain that~$(\alpha,\lc)$ belongs to~$\gamma_j$.
      Finally, since~$\rho_j$ is consistent with~$(t_j,\gamma_j)$ under $N$,
      there are at most~$N$ increments of~$\alpha$ before the first reset
      in~$\rho_j$.

      Overall there are at most~$N(M+1)$ increments of~$\alpha$
      before the first reset in~$\rho$.
      \smallskip
\item[$c=\rc$] As in the previous case.
      \smallskip
\item[$c=\mc$] As previously, but we need the bound $N(M+2)$ since
  both \hhl{ends} of the interval have to be considered.
\end{enumerate}

{\bf Case~$T=B$.}  In this case~$\comp$ is~$\geq$.
Let~$\rho=\rho_1\dots\rho_k$ be a run such that~$\rho_i$ is consistent
with~$(t_i,\gamma_i)\in\tau$ under~$N$ for all~$i=1,\dots,k$. We will
show that $\rho$ is consistent with $\tau$ under $N$, i.e.~$h$ is the
identity function.  \oldhl{Let then~$(f_0,g_0)$ be}\ttc{ the }\oldhl{separator
obtained by unraveling the definition of  $w|_D$ being strongly described
by~$\tau$.} We can
assume without loss of generality that~$g_0$ is constant, equal
to~$M$. As previously, since~$f_0$ tends toward infinity, we can chose
a natural number $n$ such that $f_0(n+1)\geq M(N+2)$.

As in the case of $T=S$, we mimic each run~$\rho_i$ by a similar
run~$\pi_i$ over the \oldhl{$(i+n)$-th word in~$w|_D$}.  By definition of a strong
description, one can find for~$i=1,\dots,k$ a run~$\pi_i$ over the
\oldhl{$(i+n)$-th word in~$w|_D$} such that 
\begin{eqnarray*}
  \val(\pi_i)\geq_{\gamma_j}f_0(i+n) \quad \mbox{and} \quad 
val(\pi_i)<_{\constraints(t_i)\setminus\gamma_i} M\ .
\end{eqnarray*}
  From the last
inequality we obtain:
  \begin{align*}
  \text{for all}~(\alpha,c)\in\constraints(t_i),\qquad\val(\pi_i,\alpha,c)\geq M\quad\text{implies}\quad(\alpha,c)\in\gamma_i\,.\tag{$\sharp$}
  \end{align*}

  Let $\pi$ be~$\pi_1\dots\pi_k$.
  Since~$\tau$ strongly describes~$w|_D$,
  the run $\pi$ is consistent with~$(t,\gamma)$ under~$f_0(n+1)$
  for some~$(t,\gamma)$ in~$\tau$. In combination with~$f_0(n+1)\geq M(N+2)$, we have:
  \begin{align*}
  \val(\pi)\geq_\gamma M(N+2)\,.\tag{$\sharp2$}
  \end{align*}

  We will show that $\rho$ is consistent
  with~$(t,\gamma)$ under~$N$,
  which establishes the stability of~$\tau$.
  For this, let~$(\alpha,c)$ be an event in~$\gamma$, we have to prove
  \begin{align*}
  \val(\rho,\alpha,c)\geq N.
  \end{align*}
  This is done by a case distinction depending on~$c$.
\begin{enumerate}[leftmargin=13mm]
\item[$c=\cnr$] In this case the run $\pi$
      does not contain any reset of~$\alpha$.
      Let~$I \subseteq \set{1,\ldots,k}$ be the set
      of \oldhl{indexes}~$j$ for which $\pi_j$
      does an increment of counter $\alpha$.
      By~($\sharp2$) applied to~$(\alpha,\cnr)$,
      there are at least~$MN$ increments (actually, at least $M(N+2)$ increments,
      but we only need $MN$ here)  of $\alpha$ in $\pi$.
      Two cases can happen; either~\oldhl{$I$ contains at least $N$ indexes},
      or there is \oldhl{some~$j \in I$} such that $\rho_j$
      contains at least~$M$ increments of~$\alpha$.

      Consider first the case when~\oldhl{$I$ has at least $N$ indexes}. Since there is at least one increment of~$\alpha$
      in every~$\rho_i$ for~$i\in I$,
      there are at least~$N$ increments of~$\alpha$ in~$\rho$.

      Otherwise there is some~$j \in I$ such that~$\pi_j$ contains at least~$M$
      increments of~$\alpha$.
      By~($\sharp$), this means that~$(\alpha,\cnr)$
      belongs to~$\gamma_j$. Finally, since~$\rho_j$
      is consistent with~$\gamma_j$
      under~$N$, we obtain that $\rho_j$, and by consequence also $\rho$,
      contains at least~$N$ increments of~$\alpha$.

      Overall there are at least~$N$ increments of~$\alpha$ in~$\rho$,
      i.e., $\val(\rho,\alpha,c)\geq N$.
    \smallskip
    \item[$c=\lc$] Let~$j$ be the first index for which~$\pi_j$
      contains a reset of~$\alpha$.
      By~($\sharp2$), there are at least~$M(N+1)$  increments (again,
      we do not need to use $M(N+2)$ here) of~$\alpha$
      before the first reset in~$\pi$.
      Two cases  \oldhl{are possible}: either~$MN$ increments of~$\alpha$
      happen in~$\pi_1\dots\pi_{j-1}$, or $\pi_j$ contains $M$
      increments of~$\alpha$ before the first reset of~$\alpha$.

      In the first case we use the same argument as for~$c=\cnr$
      (note that we were just using a bound of~$MN$ in this case)
      over the run~$\pi_1\dots\pi_{j-1}$.
      We obtain that there are at least~$N$ increments of~$\alpha$
      in~$\rho_1\dots\rho_{j-1}$.

      Otherwise, there are $M$
      increments of~$\alpha$ in~$\pi_j$ before the first reset.
      Using~($\sharp$) we deduce that $(\alpha,\lc)$
      belongs to~$\gamma_j$.
      Since~$\rho_j$ is consistent with~$(t_j,\gamma_j)$ under $N$,
      we deduce that there are at least~$N$ increments of~$\alpha$ in~$\rho_j$
      before the first reset.

      Overall there are at least~$N$ increments of~$\alpha$ in~$\rho$
      before the first reset,
      i.e., $\val(\rho,\alpha,c)\geq N$.
      \smallskip
\item[$c=\rc$] As in the previous case.
      \smallskip
\item[$c=\mc$] As previously (this time using the bound~$M(N+2)$).
\qedhere
\end{enumerate}
\end{proof}

We will now show how to tell if a word is rejected by inspecting one
of its descriptions.  This notion of rejection will be parametrized by
the set~$S$ of states in which the cut may be reached.
Note that rejecting loops  only make sense for stable descriptions.

\begin{defi}[rejecting loop]\label{def:rejecting-loop} We say a state $q$ is a \emph{rejecting
    loop} in a description $\tau$, if for all $(t,\gamma) \in \tau$
  where the source and target state of the type $t$ is $q$ we have:

  {\bf Case $T=S$.} There exists a counter~$\alpha$ such that either:
\begin{itemize}
\item $t(\alpha)=\emptyset$ or~$t(\alpha)=\set\cnr$, or;
\item $(\alpha,\mc)\in\gamma$, or;
\item Both $(\alpha,\lc)$ and $(\alpha,\rc)$ belong to $\gamma$.
\end{itemize}

{\bf Case $T=B$.} There exists a counter~$\alpha$ such that either:
\begin{itemize}
\item $t(\alpha)=\emptyset$ or~$t(\alpha)=\set\cnr$, or;
\item $(\alpha,c)\in\gamma$ for some $c\in\{\lc,\mc,\rc\}$.
\end{itemize}
\end{defi}

The idea behind the above definition is as follows: if the description
$\tau$ weakly describes a cut word $w|_D$, $\rho$ is a run that
assumes state $q$ in every position from $D$, then \oldhl{the fact that
  $q$ is a rejecting loop implies that $\rho$ not
  accepting}.  Since every infinite run can be decomposed into loops,
this is the key information when looking for a witness of rejection.

Given a description $\tau$ and a state $p$, we write $p\tau$ to denote
the set of states $q$ such that for some $(t,\gamma) \in \tau$, the
source state of $t$ is $p$ and the target state of $t$ is $q$. This
notation is extended to a set of states $P\tau$ in the natural
manner.

If $D$ is a cut and $w$ is an $\omega$-word, then the
\emph{$D$-prefix} of $w$ is defined to be the prefix of $w$ that leads
to the first position in $D$. Every $\omega$-word $w$ is decomposed
into its $D$-prefix, and then the concatenation of words from $w|_D$.

\begin{lem}\label{lemma:strong-rejection}
  Let $w|_D$ be a cut $\omega$-word strongly described by $\tau$. Let
  $P$ be the states reachable after reading the $D$-prefix of $w$.  If
  $w$ is rejected, then every state in $P\tau$ is a rejecting loop.
\end{lem}
\begin{proof}
  We only do the proof for~$T=S$, the case of $T=B$ being similar.
\hhl{Let $D=\set{d_1,d_2,\ldots}$.}

  \hhl{To obtain a contradiction}, suppose that $q \in P\tau$ is not a
  rejecting loop.  By definition, there must be a pair $(t,\gamma)$ in
  $\tau$---with $q$ the source and target of $t$---such that for every
  counter~$\alpha$, none of the conditions from
  Definition~\ref{def:rejecting-loop} hold. That is, the value
  $t(\alpha)$ contains $\lc$ and $\rc$, the event $(\alpha,\mc)$ is
  outside $\gamma$, and one of the events~$(\alpha,\lc),(\alpha,\rc)$
  is outside $\gamma$.  Without loss of generality, let us assume
  $(\alpha,\lc)$ is outside~$\gamma$.

  Let $(f,g)$ be the separator appropriate to $w|_D$ obtained from the
  definition of strong descriptions.  For each natural number~$i$ there
  is a run~$\pi_i$ of type $t$ between positions $d_i$ and $d_{i+1}$
  such that~$val(\pi_i,\alpha,c)>{g(d_i)}$ holds for all events
  $(\alpha,c)$ not in $\gamma$.

  Since~$t(\alpha)$ contains $\lc$ and $\rc$,
  the counter~$\alpha$ is reset at least once in~$\pi_i$.
  Furthermore, since $(\alpha,\mc)$ is outside $\gamma$,
  every two consecutive resets of~$\alpha$ in~$\pi_i$
  are separated by at least~$g(d_i)$ increments.
  Finally, since
  $(\alpha,\lc)$ is outside $\gamma$,
  there are at least $g(d_i)$ increments of~$\alpha$ before the first reset in~$\pi_i$. Since this holds for every counter,
  we obtain that the run~$\pi_1\pi_2\dots$
  satisfies the accepting condition.

  By assumption on $q \in P\tau$, there is some state $p \in P$ and a
  type in $\tau$ that has source state~$p$ and target state~$q$. In
  particular, the state $q$ can be reached in the second position of
  the cut $D$: by first reaching $p$ after the $D$-prefix, and then
  going from $p$ to $q$.  From $q$ in the second position of $D$, we
  can use the run $\pi_2\pi_3\ldots$ to get an accepting run over the
  word~$w$. This contradicts our assumption that $w$ was rejected by
  the automaton.
\end{proof}

The following lemma gives the converse of
Lemma~\ref{lemma:strong-rejection}. The result is actually stronger
than just the converse,
since we use weaker assumptions (the description need only be weak and
stable, which is true for every strong description, thanks to
Lemma~\ref{lemma:stability}).
\begin{lem}\label{lemma:rejection}
  Let $w$ be an $\omega$-word, $D \subseteq \Nat$ a cut, and assume
  that the word sequence $w|_D$ is weakly described by a stable
  description $\tau$. Let $P$ be the states reachable after reading
  the $D$-prefix of $w$.  If every state in $P\tau$ is a rejecting
  loop, then $w$ is rejected.
\end{lem}
\begin{proof}
  Let~$\rho$ be a run of the automaton over~$w$. We will show that
  this run is not accepting.

  Let~$q_i$ be the state used by~$\rho$ at position~$d_i$.  For~$i<j$,
  we denote by~$\rho_{i,j}$ the subrun of~$\rho$ that starts in
  position $d_i$ and ends in position $d_j$. Let~$t_{i,j}$ be the type
  of this run.

  Consider now a run $\rho_{i,j}$. This run can be decomposed as 
  \begin{eqnarray*}
    \rho_{i,j}=\rho_{i,i+1} \rho_{i+1,i+2} \cdots \rho_{j-1,j}\ .    
  \end{eqnarray*}
  Let~$f$ be the $\bar T$-function from the assumption that $\tau$
  weakly describes $w|_D$. By recalling the definition of weak
  descriptions, there must be sets of events
  $\gamma_i,\ldots,\gamma_{j-1}$ such that
\begin{eqnarray*}
    (t_{i,i+1},\gamma_i) \in \tau &\cdots &(t_{j-1,j},\gamma_{j-1}) \in \tau  \\
 val(\rho_{i,i+1})\comp_{\gamma_i} f(d_i)
    &\quad \cdots \quad&
    val(\rho_{j-1,j})\comp_{\gamma_{j-1}} f(d_{j-1})\ .
\end{eqnarray*}
Let $g$ be a function, which to every $k \in \Nat$ assigns the
maximal, with respect to $\comp$, value among $f(k),f(k+1),\ldots$ We
claim that not only $g$ is well defined, but it is also a $\bar T$
function. Indeed, when $\bar T$ is $B$ then there are finitely many
values of $f$, so a maximal one exists, and $g$ has also finitely many
values. If, on the other hand, $\bar T$ is $S$, then $\comp$ is $\ge$.
In this case, $g(k)$ is the least---with respect to the standard
ordering $\ge$ on natural numbers---number among $f(k),f(k+1),\ldots$.
This number is well defined, furthermore, $g$ is an $S$-function since
$f$ is an $S$-function.

Since $ f(d_k) \comp g(d_i)$ holds for any $k \ge i$, we also have 
\begin{eqnarray*}
 val(\rho_{i,i+1})\comp_{\gamma_i} g(d_i)
    &\quad \cdots \quad&
    val(\rho_{j-1,j})\comp_{\gamma_j} g(d_i)\ .
\end{eqnarray*}
  By assumption on stability of~$\tau$, there is an $S$-function $h$,
  such that for all~$i<j$, there is
  \begin{align*}
     (t_{i,j}, \gamma_{i,j}) \in \tau \qquad \mbox{such that} \qquad
     val(\rho_{i,j})\comp_{\gamma_{i,j}} h(g((d_i)) \ .
  \end{align*}

  Using the Ramsey theorem in a standard way, we can assume without
  loss of generality that all the $t_{i,j}$ are equal to the same
  type~$t$, and all $\gamma_{i,j}$ are equal to the same $\gamma$. If
  $t(\alpha)=\{\cnr\}$ holds for some counter $\alpha$, then this
  counter is reset only finitely often, so the run $\rho$ is rejecting
  and we are done. Otherwise, $t(\alpha)=\{\lc,\mc,\rc\}$ holds for
  all counters $\alpha$.
  
  For the rest of the proof, we only consider the case~$T=S$, with $B$
  being treated in a similar way.  The function $g$, by its
  definition, assumes some value $M$ for all but finitely many
  arguments.  Since~$\tau$ is rejecting there exists a
  counter~$\alpha\in\Gamma$ such that either~$(\alpha,\mc)\in\gamma$
  or both $(\alpha,\lc)$ and $(\alpha,\rc)$ belong to $\gamma$.  In
  the first case, an infinite number of times there are at most~$h(M)$
  increments of~$\alpha$ between two consecutive resets of~$\alpha$.
  In the second case, the same happens, but this time with at
  most~$2h(M)$ increments. In both cases the run is not accepting.
\end{proof}

We are now ready to establish the main lemma of this section.
\begin{lem}\label{lemma:ramsey-description}
  Let~$\Aa$ be an $\omega T$-automaton. There exist
  \ttc{regular languages
  $L_1,\ldots,L_n$ and stable
  descriptions $\tau_1,\ldots,\tau_n$}
  such that for every $\omega$-word $w$ the following \ttc{items} are equivalent:
\begin{itemize}
\item $\Aa$ rejects~$w$, and;
\item There is some $i=1,\ldots,n$ and a cut $D$ such that the
      $D$-prefix of $w$ belongs to $L_i$, and $\tau_i$ weakly
      describes~$w|_D$.
\end{itemize}
\end{lem}

\begin{proof}
  We need to construct a finite set of pairs~$(L_i,\tau_i)$.  Each
  such pair~$(L_P,\tau)$ comes from a set of states~$P$ and a
  description~$\tau$ that is stable and rejects all loops in~$P\tau$.
  The language~$L_P$ is the set of finite words~$v$ that give exactly
  states $P$ \hhl{(as far as reachability from the initial state is
    concerned)} after being read by the automaton.

  The bottom-up implication is a direct application of
  Lemma~\ref{lemma:rejection}, and therefore only the top down
  implication remains.  Let~$w$ be an $\omega$-word rejected by $\Aa$. By
  Lemma~\ref{lemma:strong-description}, there exists a cut~$D$ and a
  strong---and therefore also weak---description~$\tau$ of~$w|_D$.
  Let $P$ be the set of states reached\ttc{ }after reading the $D$-prefix
  of $w$. Clearly the $D$-prefix of $w$ belongs to $L_P$. By
  Lemma~\ref{lemma:stability}, the description~$\tau$ is stable, and
  hence Lemma~\ref{lemma:strong-rejection} can be applied to show that
  all loops in $P\tau$ are rejecting.
\end{proof}

The first part of Proposition~\ref{prop:main-compl} follows, since
weak descriptions are captured by specifications thanks to
Lemma~\ref{lemma:descriptions-simulated-by-specifications}. Finally,
we need to show that our construction is effective:
\begin{lem}
It is decidable if a description is stable.
\end{lem}
\begin{proof}
  Stability can be verified by a formula of monadic second-order logic
  over $(\Nat, \le)$.
\end{proof}


\subsection{Verifying single events}
\label{subsection:verifying_single_constraints}

We now begin the part of the complementation proof  where we show
that specifications can be recognized by automata.  Our goal is as
follows: given a specification~$\tau$, we \ttc{want to construct} a hierarchical
sequence~$\bar T$-automaton that accepts the word sequences that are
consistent with~$\tau$. Recall that a specification is a positive
\hhl{boolean combination} of two types of atomic conditions.  In this
section we concentrate solely on atomic specifications \oldhl{where the
  boolean combination  consists of only one atomic condition} of the form:
\begin{quote}
  (2) The run satisfies~$\val(\rho,\alpha,c) \comp K$.
\end{quote}

\newcommand{\states}{\ensuremath{M}}

\oldhl{In particular,} only one event $(\alpha,c)$
is involved in the specification. Furthermore, we also assume that the
counter~$\alpha$ is the lowest-ranking counter~$1$.  However, the
result is \hhl{ stated so that it can then be generalized
to any specification.}  

\medskip

\subsubsection{Preliminaries and definitions}

In this section, we define transition graphs---which are used to
represent possible runs of an automaton---and then we present a
decomposition result for transition graphs, which follows from a
result of Simon on factorization forests~\cite{simon}.

\subsubsection*{Transition graph}

For the rest of  Section~\ref{subsection:verifying_single_constraints}, we fix a finite set of states
$\states$. This set~$\states$ is possibly different from the set of
states of the automaton we are complementing. The reason is that in
Section~\ref{subsection:verifying-weak-description}, we will increase
the state space of the complemented automaton inside \hhl{an  induction.}

The first concept we need is an explicit representation of the
configuration graph of an automaton reading a word, called here an
$\states$-transition graph.
Fix a finite set~$\lab$ of \intro{transition labels}, and a finite set
of \emph{states} $\states$.
An
\emph{\states-transition graph~$G$ (labeled by~$\lab$) of length $k
  \in \Nat$} is a directed edge
\hhl{labeled} graph, where the nodes---called \emph{configurations} of the
graph---are pairs $M \times  \set{0,\ldots,k}$
and the edge labels\ttc{ are of the form~$((q,i),l,(r,i+1))$
for~$q,r$ in~$M$, $0\leq i<k$ and~$l\in\lab$.
The vertexes of the graph are called \emph{configurations},
their first component is called the \emph{state}, while
the second is called the \emph{position}.}
The \oldhl{edges of the graph} are called the
\intro{transitions}.  We define the \intro{concatenation} of
$\states$-transition graphs in a natural way.
\oldhl{A \emph{partial run} in a transition graph is just a path in the
  graph.} 
A \intro{run} in a transition graph is   \oldhl{a path in the graph that
  begins in a configuration at the first position $0$ and ends in a
  configuration at the last position.}  The \intro{label} of a run
is the  sequence \oldhl{of labels on edges of the  path.}

A transition graph of length~$k$
can also be seen as a word of length~$k$ over the alphabet
\begin{equation*}
  \mathcal{P}(\states\times\lab\times\states)\ .
\end{equation*}
In this case, the concatenation of transition graphs coincides with
the standard concatenation of words.
When speaking of regular sets of transition graphs, we refer to this
representation.

Given a \hhl{$\omega T$-automaton~$\Aa$} over the alphabet~$\lab$ of states~$Q$, the
\intro{product of an $M$-transition graph~$G$ with~$\Aa$}---noted
$G\times \Aa$---is the $(M\times Q)$-transition graph\oldhl{ which has an
  $l$-labeled edge from $((p,q),i)$ to $((p',q'),i+1)$ whenever $G$
  has an $l$-labeled edge from $(p,i)$ to $(p',i+1)$ and~$(q,l,q')$} is a
transition of the automaton~$\Aa$.  Furthermore, in the product graph
only those configurations are kept that can be reached via a run that
begins in a  \oldhl{ configuration where the second component of the
  state is the} initial state of
$\Aa$. 

\subsubsection*{Factorization forest theorem of Simon, and decomposed transition graphs}

From now, we will only consider transition graphs where the first
component of the labeling \ttc{ranges over the}
actions of a hierarchical automaton over counters, i.e.,
the label alphabet is of the form $\actions \times \lab$, with
\begin{align*}
	\actions\ &=\set{\e,I_1,R_1,I_2,\dots,R_n} \ .
\end{align*}
The type of a run is defined as in
the previous section, i.e., a type gives the source and target states,
as well as a mapping from the set of counters
to~$\set{\emptyset,\set{\cnr},\set{\lc,\rc},\set{\lc,\mc,\rc}}$.
Given an $\states$-transition graph, its \intro{type} is the element \hhl{of}
\begin{align*}
S\ &=\mathcal{P}(M\times\set{\emptyset,\set{\cnr},\set{\lc,\rc},\set{\lc,\mc,\rc}}^\Gamma\times M)\,,
\end{align*}
which contains \hhl{those types~$t$ such that there is a run over~$G$ of type~$t$}.  This set
$S$ can be seen as a semigroup, \ttc{when equipped with the product defined by}:
\begin{eqnarray*}
  s_1 \cdot s_2 = \set{ t_1 \cdot t_2~:~t_1 \in s_1,
      t_2 \in s_2}\ .
\end{eqnarray*}
In the above, the concatenation of two types $t_1 \cdot t_2$ is
defined in the natural way: the source state of $t_2$ must agree with
the target state of $t_1$, and the counter operations are
concatenated. \hhl{(An example of how counter operations are
  concatenated is: 
  \begin{eqnarray*}
    \set{\lc , \rc}\cdot \set{\lc, \rc} =\set{\lc,\mc,\rc}\ .
  \end{eqnarray*}}
In particular, the mapping \oldhl{that} assigns the type  to a counter
transition graph is a semigroup morphism, with \hhl{transition}
graphs interpreted as words.
Two~$\states$-transition graphs~$G,H$ are \oldhl{called} \intro{equivalent} if
they have \hhl{the} same type.
This equivalence is clearly a congruence of finite index with respect
to concatenation.
A transition graph~$G$ is \intro{idempotent} if the concatenation~$G G$
is equivalent to~$G$.

We now define a complexity measure on graphs, which we call their
\emph{Simon} level.  A transition graph~$G$ has {Simon level~$0$} if
it is of length~$1$.  A transition graph~$G$ has {Simon  level \hhl{at
    most}~$k+1$}
if it can be decomposed as a concatenation
\begin{equation*}
  G=H G_1\cdots G_n H'\ ,
\end{equation*}
where all the transition graphs~$H,G_1,\dots,G_n,H'$ have Simon level
at most~$k$ and~$G_1,\dots,G_n$ are equivalent and idempotent.
The following theorem, which in its \hhl{original statement} concerns
semigroups, is presented here in a form adapted to our
context.
\begin{thm}[Simon~\cite{simon}, \ttc{and}~\cite{colcombetFCT07}
  \ttc{for the bound $|S|$}]\label{theorem:simon}
  Given a finite set of states~$\states$, the Simon level of
  $\states$-transition graphs is bounded by $|S|$.
\end{thm}

The Simon level is defined in terms of a nested decomposition of the
graph into factors.  We will sometimes need to refer explicitly to
such decompositions; for this we will use symbols~$(,),|$ and write
$G$ as~$(H|G_1|\dots|G_n|H')$, and so on recursively for the graphs
$H,G_1,\ldots,G_n,H'$.  Theorem~\ref{theorem:simon} shows that each
graph admits a decomposition where the nesting of parentheses in this
notation is bounded by $|S|$.
We refer to the transition graphs written in this format as
\intro{decomposed transition graphs}.  It is not difficult to see that
the set of decomposed~$\states$-transition graphs is a regular
language: a finite automaton can
check that the symbols~$(,),|$ indeed describe a Simon decomposition
\oldhl{(thanks to Simon's theorem, the automaton does not need to count
  too many nested parentheses).}


A \intro{hint} over a decomposed graph~$G$ is a subset of the
positions labeled~$|$ in the decomposition of~$G$.  We call those
positions \intro{hinted positions}. (Note that a hint is relative not
just to a transition graph $G$, but also to some decomposition of this
transition graph.) For~$K \in \Nat$, a hint is said
to be~$\geq K$ if at the same nesting level of the decomposition,
every two distinct hinted positions are separated by at least~$K$
non-hinted symbols~$|$.  Similarly, a hint is~$\leq K$ if sequences
of consecutive non-hinted positions at a given Simon level
have length at most~$K-1$.

Let $G$ be a transition graph \hhl{of length $k$ (in the word representation)} (with labels $\actions \times \lab$),
along with a decomposition and a hint $h$. Let
\begin{eqnarray*}
  G_1 \cdots G_k \in   (\mathcal{P}(\states\times \actions \times \lab \times\states))^*
\end{eqnarray*}
be the interpretation of this graph as word.  In the \hhl{proofs below,} it will be
convenient to decorate the graph---by expanding the transition
labels---so that the label of each run contains information about the
decomposition and the hint. The decorated graph is denoted by $(G,h)$
(we do not include a name for the decomposition in this notation,
since we assume that the \hhl{hint $h$  also contains the information
  on the Simon 
decomposition}).  The graph $(G,h)$ has the same configurations, states
and transitions as $G$; only the labeling of the transitions changes.
Instead of having a label in $\actions \times \lab$, as in the graph
$G$, a transition in the graph $(G,h)$ has a label in
\begin{align*}
\lab_S\ &=\actions\times \lab \times \{\bot,1,\dots,|S|\}\times \{0,1\}\,.
\end{align*}
The first two coordinates are inherited from the graph $G$. The other
coordinates are explained below. \hhl{To understand the encoding, we
  need two properties. First, in the decomposition, there {is exactly}
  one symbol~$|$ between any two successive letters $G_i,G_{i+1}$ of
  the $\states$-transition graph seen as a word. Second, the
  decomposition is entirely described by the nesting depths of those
  symbols \hhl{with respect to the }parentheses. Recall also that
  these nesting depths are bounded by $|S|$.}  For \oldhl{a} transition
in the graph $G_i$, the coordinate $\set{\bot, 1,\ldots,|S|}$ stores
the nesting depth \hhl{(with respect to the parentheses)} of the
symbol $|$ preceding the transition graph $G_i$; \hhl{by the first
  property mentioned above,} the undefined value $\bot$ is used only
for the first transition in transition graph. Finally, the coordinate
$\set{0,1}$ says if the symbol $|$ is included in the hint.

\subsubsection*{Runs over decomposed graphs}
As remarked above, the transition graph $(G,h)$ only changes the
labels of transitions in the graph $G$. Therefore with each run $\rho$
in $G$ we can associate the unique corresponding run \hhl{over} $(G,h)$,
which we denote by $(\rho,h)$. By abuse of notation, we will sometimes
write $(\rho,h) \in L$, where $L \subseteq (\lab_S)^*$. The intended meaning is that the labeling of the run
$(\rho,h)$ belongs to the language $L$.

\newcommand{\eqr}[1]	{\equiv_{#1}}

Recall that we are only going to be verifying properties for the
counter $\alpha=1$ in this section.
We consider two runs over the same transition graph to be~$\eqr
\rc$-equivalent, \hhl{if they  agree before
   the last~$1$-reset (i.e.~use the same transitions on all positions
   up to and including the last $1$-reset).  In the notation, the index shows where the runs can be different. Similarly, two runs are
$\eqr \lc$-equivalent if they agree    after
  the first $1$-reset.}
Two runs are $\eqr\mc$-equivalent if they agree before the first
$1$-reset, after the last $1$-reset, and over all $1$-resets (but do
not necessarily agree between two successive $1$-resets).
Finally, two runs are~$\eqr \cnr$-equivalent if both increment
counter~$1$ but do not reset it.
The fundamental property \oldhl{of $\eqr c$-equivalence, for $c \in
\set{\rc, \lc, \mc, \cnr}$},
is that two
$\eqr c$-equivalent runs are indistinguishable
in terms of resets and increments
of counters greater or equal to~$2$,
or in terms of events~$(1,c')$
with~$c'\neq c$. This means that as long
we are working inside a $\eqr c$-equivalence
class, the values $\val(\rho,\alpha,c')$
for $(\alpha,c')\neq(1,c)$ are \oldhl{constant}.
(This \oldhl{also} is the reason why we  only work with
counter~$1$ in a hierarchical automaton.)

The key lemmas in this section are Lemmas~\ref{lemma:hinting-S}
and~\ref{lemma:hinting-B}.  These talk about complementing $S$-automata
and $B$-automata respectively. They both follow the same structure,
which can be uniformly expressed in the following lemma, using the
$\comp$ order (which is~$\le$ for complementing $S$-automata and~$\ge$
for complementing $B$-automata):

\newcommand{\commonlemma}{ Let~$c$ be one of~$\lc,\mc,\rc$ or~$\cnr$.
  There are two strongly unbounded functions~$f$ and~$g$, and a
  regular language~$L_c \subseteq (\lab_S)^*$ such that for every
  natural number~$K$,
  every run~$\rho$ over every~$\states$-transition graph labeled by~$\lab$
  whose type contains the event~$(1,c)$ satisfies:}

\begin{lem}\label{lemma:hinting-generic}
  \commonlemma
\begin{itemize}
\item (correctness) if a hint~$h$ in the graph is~$\comp K$ and every
      run~$\pi \eqr c \rho$ satisfies~$(\pi,h)\in L_c$,
      then~$\val(\rho,1,c)\comp f(K)$, and;
\item (completeness) if a hint~$h$ in the graph is~$\acomp g(K)$ and
      $\val(\rho,1,c)\comp K$ then~$(\rho,h)$ belongs to~$L_c$.
\end{itemize}
\end{lem}
We would like to underline here that the regular language is a
regular language of finite words \ttc{in the usual sense}, i.e., no
counters are involved.  

The above lemma says that the language~$L_c$ ``approximates'' the runs~$\rho$
satisfying $\val(\rho,1,c)\comp K$. This ``approximation'' is
\hhl{given} by the two functions~$f$ and~$g$.  There is however a
dissymmetry in this \oldhl{statement.  The} completeness clause states that if a
run \oldhl{has a bad value on event $(1,c)$, i.e.~it satisfies~$\val(\rho,1,c)\comp K$,} then it is detected by~$L_c$ for
all sufficiently good hints. \oldhl{On the other hand, in the} correctness clause, all
$\eqr c$-equivalent runs must be detected by~$L_c$ in order to \oldhl{say
  that the value of $\rho$ is bad on event $(1,c)$}.  The reason for
\oldhl{the weaker correctness clause is that we will not be able to check
  the value on event $(1,c)$ for every run; we will  only do it} for runs of a special simplified form. And the
simplification process happens to transform each run in a~$\eqr
c$-equivalent one.

The proof of this lemma differs significantly depending on~$T=S$
or~$T=B$.
\oldhl{The} two cases correspond to
Lemmas~\ref{lemma:hinting-S} and~\ref{lemma:hinting-B} which are
instantiations of Lemma~\ref{lemma:hinting-generic}.

\medskip

\subsubsection{Case of complementing an~$\omega S$-automaton}
We consider first the case of complementing an~$S$-automaton.
Therefore, the order~$\comp$ is~$\leq$. To  \hhl{aid reading, below we
  restate  Lemma~\ref{lemma:hinting-generic} for the case when $T=S$.}

\begin{lem}\label{lemma:hinting-S}
\commonlemma
\begin{itemize}
\item (correctness) if a hint~$h$ in the graph  is~$\leq K$ and every run~$\pi \eqr c
  \rho$ satisfies~$(\pi,h)\in L_c$, then~$\val(\rho,1,c)\leq
  f(K)$, and;
\item (completeness) if a hint~$h$ in the graph is~$\geq g(K)$ and
  $\val(\rho,1,c)\leq K$ then~$(\rho,h)$ belongs to~$L_c$.
\end{itemize}
\end{lem}
Slightly ahead of time, we remark that~$g$ will be the identity
function, while~$f$ will be a polynomial, whose degree is the maximal
Simon level of~$\states$-transition graphs, taken from
Theorem~\ref{theorem:simon}.

Before proving the lemma, we would like to give some intuition
  about the language~$L_c$. The idea is that we want to capture the
  runs which do few increments on counter~$1$.  However, this ``few''
  cannot be encoded in the state space of the automaton, since it can
  be arbitrarily large. That is why we use the hint. One can think of
  the hint as a clock: if the hint is~$\le K$, then the clock \hhl{ticks}
  quickly, and if the hint is~$\ge K$, then the clock \hhl{ticks}
  slowly. The language $L_c$ looks at a run and compares it to the
  clock.  For the sake of the explanation, let us consider a piece of
  a run without resets of counter~$1$: we want to estimate if a lot of
  increments of counter~$1$ are done (at least $K$) or not (at most
  $f(K)$) by comparing it to a suitable clock (in the first case a
  slow clock, in the second case a quick one).  The first case is when
  the counter is incremented in every position between two \hhl{ticks} of
  the clock; then the value of the counter concerned is considered
  `big', since at least $K$ increments are done if the \hhl{ticks}
  are~$\ge K$.  \hhl{For the second case,  when there are few
    increments,  we have a  more  involved
  argument that  uses the idempotents from the Simon decomposition. }

\begin{proof}
  The language~$L_c$ is defined by induction on the Simon level of
  the transition graph. We construct a language~$L^k_c$ which has
  the stated property for all~$\states$-transition graphs of Simon
  level at most~$k$. Since there is a bound on the Simon level,
  the result follows.

  For~$k=0$, the construction is straightforward, since the transition
  graphs are of length~$1$ and there is a finite number of possible
  runs to be considered.

  We now show how to define the language \oldhl{$L_c^{k+1}$} for runs in transition
  graphs of Simon level~$k+1$, based on the languages for runs in
  transition graphs of Simon level up to~$k$.

  Consider a decomposed~$\states$-transition graph
  \begin{equation*}
    G=(H|G_1|\dots|G_n|H')\ .
  \end{equation*}
  of Simon level~$k+1$. Let~$h$ be a hint for this decomposition,
  which we decompose into sub-hints~$(h_0| \cdots |h_{n+1})$ for the
  graphs~$H,G_1,\dots,G_n,H'$. In the proof, instead of writing
  $(\rho,h) \in L^{k+1}_c$, we write that~$\rho$ is
  \emph{$c$-captured}.  Let~$\rho$ be a run over~$G$.  The run~$\rho$
  can also be decomposed into
  subruns~$(\rho_0|\rho_1|\dots|\rho_{n+1})$.  Each of these runs
  $\rho_0,\ldots,\rho_{n+1}$ is \hhl{over} a transition graph whose Simon
  level is at most~$k$. By abuse of notation, we will talk about a
  subrun~$\rho_i$ being~$c$-captured, the intended meaning being that
  $(\rho_i,h_i)$ belongs to the appropriate language~$L^{k'}_c$,
  with~$k'\leq k$ being the Simon level of~$G_i$ (or $H$ if~$i=0$,
  or $H'$ if~$i=n+1$).

\oldhl{We now proceed to define the language $L_c^{k+1}$. In other words, we
 need to say when $\rho$ is $c$-captured.}
  We only do the case of~$c=\mc$, which is the most complex situation.
  The idea is that a run is~$\mc$-captured if there are two
  consecutive resets between which the run does few increments.  We
  define~$\rho$ to be $\mc$-captured if \oldhl{either:}
\begin{enumerate}
\item some~$\rho_i$ is~$\mc$-captured, or;
\item for some~$i < j$, the subrun~$\rho_i$ is~$\rc$-captured, the run
  $\rho_{j}$ is~$\lc$-captured, each
  of~$\rho_{i+1}, \dots,\rho_{j-1}$
  is~$\cnr$-captured, and one of the following holds:
	\begin{enumerate}
  \item there are~$m \le m'$ in~$\set{i,\ldots,j}$
          such that: (i) one of~$\rho_m,\rho_{m+1},\dots,\rho_{m'}$
          does not increment counter~$1$; and (ii) the source state
          of~$\rho_m$ and the target state of $\rho_{m'}$ are the same
          state $q$, and~$G_1$ admits a run from $q$ to $q$ that
          increments counter~$1$, or;
    \item between any two hinted (by hints on level~$k+1$) positions
          in~$\set{i,\ldots,j}$, at least one of the runs~$\rho_m$
          does not increment counter~$1$.
	\end{enumerate}
\end{enumerate}
It is clear that this definition corresponds to a regular property
\oldhl{$L_{\mc}^{k+1}$ of the sequence of labels in a hinted
runs (once the hints are provided, the statement above is first-order
definable).}

We try to give an intuitive description of the above conditions.  The
general idea is that a run gets~$\mc$-captured if it does few
increments, at least relatively to the size of the hint.  The first
reason why a run may do few increments between some two resets, is
that it does it inside one of the component transition graphs of
smaller Simon level. This is captured by condition 1.  The second
condition is more complicated. The idea behind 2(a) is that the run
is---in a certain sense---suboptimal and can be converted into a~$\eqr
\mc$-equivalent one that does ``more'' increments. Then the more optimal
run can be shown---using conditions 1 and 2(b)---to do few increments,
which implies \ttc{that} the original suboptimal run also did few increments.

We now proceed to show that the properties defined above satisfy the
completeness and correctness conditions in the statement of the lemma.

\medskip
\noindent{\it Completeness. }
For this, we set~$g(K)=K$.  Let~$\rho$ be a run such
that
\begin{equation*}
  \val(\rho,1,\mc)\leq K
\end{equation*}
and let~$h$ be hint over~$G$ that is~$ \geq K$. We have to show that
$\rho$ is~$\mc$-captured.

Consider a minimal subrun~$\rho_i\dots\rho_j$ that resets
counter~$1$ twice and satisfies
\begin{equation*}
  \val(\rho_i\dots\rho_j,1,\mc)\leq K \ .
\end{equation*}
In particular, the counter~$1$ is reset in the subruns~$\rho_i$ and
$\rho_{j}$.  If~$i=j$, this means that $\val(\rho_i,1,\mc)\leq K$, and hence
$\rho_i$ is~$\mc$-captured on a level below~$k+1$ by induction
hypothesis. We conclude with item~$1$.

Otherwise, we have
\begin{align*}
&\val(\rho_i,1,\rc)\leq K\,, \\
\val(\rho_{i+1},1,\cnr)&\leq K \quad \cdots \quad \val(\rho_{j-1},1,\cnr)\leq K\,,\\
\text{and}\quad
&\val(\rho_{j},1,\lc)\leq K \ .
\end{align*}

By induction hypothesis, we obtain the header part of item~$2$.  Since
the run~$\rho_i\dots\rho_j$ contains less than~$K$ increments
of~$1$, no more than~$K$ runs among~$\rho_i,\dots,\rho_j$ can
increment counter~$1$.  Since the hint~$h$ is $ \geq K$, we get
item (b).

\medskip
\noindent{\it Correctness. }
Let~$f'$ be the strongly unbounded function obtained by the induction
hypothesis for Simon level~$k$.  We set
\begin{equation*}
  f(K)=(2K|\states|+2)f'(K)
\end{equation*}
Assume now that the hint~$h$ is~$\le K$ and take a run~$\rho$ such
that every run~$\pi\eqr \mc \rho$ is~$\mc$-captured.  We need to show
that
\begin{equation}\label{eq:bounded-inequality}
  \val(\rho,1,\mc)\leq f(K) \ .
\end{equation}
As before, we decompose the run~$\rho$ into
$(\rho_0|\dots|\rho_{n+1})$.

We are going to first transform~$\rho$ into a new~$\eqr \mc$-equivalent
run~$\pi$ which, intuitively, is more likely to have many increments
on counter~$1$.  We will then show that the new run~$\pi$ satisfies
inequality (\ref{eq:bounded-inequality}); moreover we will show that
this inequality can then be transferred back to~$\rho$.

We begin by describing the transformation of~$\rho$ into~$\pi$. This
transformation\ttc{ }is decomposed into two stages.

In the first stage, which is called the \emph{local transformation},
we replace the subruns $\rho_0,\ldots,\rho_{n+1}$ with new equivalent ones of
the same type. Each such replacement step works as follows.  We take
some~$c=\lc,\mc,\rc,\cnr$ and~$i=0,\ldots,n+1$. If there is a subrun
$\pi_i \eqr c \rho_i$ that is not~$c$-captured, then we replace
$\rho_i$ with~$\pi_i$. (We want the local transformation to keep the
$\eqr \mc$-equivalence class, so we do not modify subruns before the
first or after the last reset of~$1$ in~$\rho$.) The local
transformation consists of applying the replacement steps as long as
possible. This process terminates, since the replacement steps for
different~$c$'s work on different parts of the subrun and each step
decreases the number of captured subruns.

In the second stage, which is called the \emph{global transformation},
we add increments on counter~$1$ to some subruns.  The idea is that at
the end of the global transformation, a run does not satisfy condition
2(a). Just as the local transformation, it consists of applying a
replacement step as long as possible. The replacement step works as
follows. We try to find a subrun~$\rho_m\dots\rho_{m'}$ as in
item~2(a).  By assumption 2(a) and since all the graphs~$G_i$'s are
equivalent to~$G_1$, we can find new subruns~$\pi_m,\dots,\pi_{m'}$
in~$G_m,\dots,G_{m'}$ respectively, which increment counter~$1$
without resetting it (that is, of type~$\cnr$) and go from~$q$ to~$q$.
We use these runs instead of $\rho_m \dots \rho_{m'}$.  The iteration
of this replacement step terminates, since each time we add new
subruns with increments.

Neither the local nor the global transformation change the~$\eqr
\mc$-equivalence class of the run.

Let~$\pi$ be a run obtained from~$\rho$ by applying first the local
and then the global transformation. (Since the global transformation
is nondeterministic, there may be more than one such run.) This run
cannot satisfy 2(a), since the global transformation could still be
applied.

Since~$\pi$ is~$\eqr \mc$-equivalent to~$\rho$, it must be
$\mc$-captured by assumption on~$\rho$.  There are two possible
reasons: either because of 1, or because of 2(b). We will now do a
case analysis on the reasons why this happens. In each case we will
conclude that the original run~$\rho$
satisfies~(\ref{eq:bounded-inequality}).  As before, we decompose
$\pi$ into~$(\pi_0| \cdots | \pi_{n+1})$.

Assume now item~$1$ holds for~$\pi$, i.e., some subrun~$\pi_i$ is
$\mc$-captured for some~$i$. We also know that any run~$\pi'_i \eqr
\mc \pi_i$ would also be~$\mc$-captured, since otherwise~$\pi_i$ would
be replaced in the local transformation process (the global
transformation does not touch subruns with resets on counter~$1$). In
particular, the induction hypothesis gives
\begin{equation*}
 \val(\pi_i,1,\mc) \leq  K \ .
\end{equation*}
Moreover, the runs~$\rho_i$ and~$\pi_i$ agree on the part between the
first and last reset of counter~$1$. This is because neither of the
transformation processes touched this part. Indeed, the first local
transformation process never modified~$\rho_i$ for~$c=\mc$ (since
otherwise it would cease being captured), while the global process
only modifies subruns without resets of counter~$1$.  This gives the
desired~(\ref{eq:bounded-inequality}), since
\begin{eqnarray*}
  \val(\rho,1,\mc)\leq \val(\rho_i,1,\mc) =   \val(\pi_i,1,\mc) \leq
  f'(K) \le f(K) \ .
\end{eqnarray*}

Otherwise,~$\pi$ satisfies item 2(b). Let us fix~$i,j$ as in 2(b).
Let~$I \subseteq \set{i+1,\ldots,j-1}$ be the set of those indexes~$l$
where~$\pi_l$ does at least one increment on counter~$1$. A maximal
contiguous \oldhl{(i.e.~containing consecutive numbers)} subset of~$I$ is called an \intro{inc-segment}.  According
to case 2(b), an inc-segment cannot contain two distinct hinted
positions, its size is therefore at most twice the maximal width~$K$
of~$h$.  Assume now that there are more than~$|\states|$ inc-segments.
Then two inc-segments contain runs with the same source state, say
state~$q$, at respective positions~$l$ and~$l'$ with~$l<l'$.  This
means that there is a run that goes from~$q$ to~$q$ in some graph
\ttc{$G_{l+1}\cdots G_{l'}$} and does at least one increment but no resets on
counter~$1$.  Using idempotency and the equivalence of all the graphs
$G_1,\ldots,G_n$, this implies that the graph~$G_{l'}$ admits such a
run.  Since~$\rho_{l'}$ contains no increments of~1 by definition of
an inc-segment, this means that 2(a) holds; a contradiction.
Consequently there are at most~$|\states|$ inc-segments,  each of size
at most $2K$.  It follows
that at most~$2K|\states|$ subruns among~$\pi_i,\dots,\pi_j$
\oldhl{ contain} an increment of the counter~$1$.

Let us come back to the original run~$\rho$.
As in the case of item 1, we use the induction hypothesis to show that
\begin{align*}
  \val(\rho_i,1,\rc) &= \val(\pi_i,1,\rc)\leq  f'(K)\,, \\
  \text{and}
\quad\val(\rho_j,1,\lc) &= \val(\pi_j,1,\lc)\leq  f'(K)\,.
\end{align*}
Since the global transformation only adds increments on counter~$1$,
we use the remarks from the previous paragraph to conclude that there
are at most~$2K|\states|$ subruns among $\rho_{i+1},\dots,\rho_{j-1}$
that increment counter~$1$.  Furthermore, all runs~$\eqr
\cnr$-equivalent to one of the $\rho_{i+1},\cdots,\rho_{j-1}$ are
$\cnr$-captured, since otherwise the local transformation would be
applied, and then one of the runs~$\pi_{i+1},\cdots,\pi_{j-1}$ would
not be captured.  Therefore, we can use the induction hypothesis to
show that
\begin{equation*}
  \val(\rho_{i+1},1,\cnr),\cdots,  \val(\rho_{j-1},1,\cnr) \leq f'(K)\ .
\end{equation*}
All this together witnesses the expected
\begin{equation*}
  \val(\rho,1,\mc)\leq  (2K|\states|+2)f'(K)\ .
  \tag*{\qEd}
\end{equation*}
\def\popQED{}
\end{proof}

\medskip

\subsubsection{Case of complementing an~$\omega B$-automaton}
As when complementing an $\omega S$-automaton, we restate the lemma
with $\comp$ expanded to its definition, which is $\geq$ in this case.
\begin{lem}\label{lemma:hinting-B}
\commonlemma
\begin{itemize}
\item (correctness) if a hint~$h$ is~$\geq K$ in the graph and every
      run~$\pi \eqr c \rho$ satisfies~$(\pi,h)\in L_c$,
      then~$\val(\rho,1,c)\geq f(K)$, and;
\item (completeness) if a hint~$h$ in the graph  is~$\leq g(K)$ and
 $\val(\rho,1,c)\geq K$ then~$(\rho,h)$ belongs to~$L_c$.
\end{itemize}
\end{lem}

We remark here that $f$ will be the identity function, while $g$ will
be more or less a $k$-fold iteration of the square root, with $k$ the
maximal Simon level of~$\states$-transition graphs, taken from
Theorem~\ref{theorem:simon}.

\begin{proof}
  The structure of the proof follows the one in
  Lemma~\ref{lemma:hinting-S}.  As in that lemma, the language~$L_c$
  is defined via an induction on the Simon level of the transition
  graph.  We also only treat the case of~$c=\mc$. The other cases can
  be treated with the same technique.

  Consider a decomposed transition graph
  \begin{equation*}
    G=(H|G_1|G_2|\dots|G_n|H')
  \end{equation*}
  of Simon level~$k+1$, along with a corresponding hint~$h$,
  decomposed as~$(h_0| \dots |h_{n+1})$. Let~$\rho$ be a run over~$G$, decomposed
  as~$(\rho_0|\dots|\rho_{n+1})$. As in
  Lemma~\ref{lemma:hinting-S}, instead of saying that~$(\rho,h)$
  belongs to the language~$L^{k+1}_c$, we say that~$\rho$ is
  $c$-captured; likewise for the runs~$\rho_0,\ldots,\rho_{n+1}$.

We define~$\rho$ to be $\mc$-captured  if either:
\begin{enumerate}
\item \oldhl{some}~$\rho_i$ is~$\mc$-captured or there is a subrun of
  the form~$\rho_i\dots\rho_j$ such that counter~$1$ is reset in
  both~$\rho_i$ and~$\rho_{j}$ but not
  in~$\rho_{i+1},\dots,\rho_{j-1}$, and either:~$\rho_i$ is
 $\rc$-captured, one of~$\rho_{i+1}, \dots, \rho_{j-1}$ is
 $\cnr$-captured, or~$\rho_j$ is~$\lc$-captured;
\item there is a subrun of the form~$\rho_i\dots\rho_j$ such that
  counter~$1$ is reset in both~$\rho_i$ and~$\rho_{j}$ but not
  in~$\rho_{i+1},\dots,\rho_{j-1}$, and either
	\begin{enumerate}
	\item there are~$m \le m'$ in~$\set{i\ldots j-1}$ such that
          there is an increment of~$1$ in one
          of~$\rho_m,\dots,\rho_{m'}$, but~$G_1$ contains a path from
          the source state of $\rho_m$ to the target state
          of~$\rho_{m'}$ without any increment nor reset, or;
      \item there are two hinted positions~$m<m'$ in~$\set{i+1 \dots
          j-1}$ such that all the subruns~$\rho_{m},\dots,\rho_{m'}$
        increment counter~$1$.
	\end{enumerate}
\end{enumerate}
As in the case of $T=S$, this definition corresponds to a regular
property \oldhl{$L_{\mc}^{k+1}$ of hinted}  runs.

The intuition is that a \oldhl{$\mc$-captured} run does a lot of increments on counter~$1$,
at least relative to the size of the hint. The first reason for doing
a lot of increments\oldhl{ is} that there are a lot of increments in a subrun
inside one of the component transition graphs, which corresponds to
item 1.  The clause 2(b) is also self-explanatory: the number of
increments is at least as big as the size of the hint.  The first
clause 2(a) is more involved. It says that the run is suboptimal in a
sense, i.e., it can be converted into one that does fewer increments.
We will then show---using 1 and 2(b)---that the run with fewer
increments also does a lot of increments.

We now proceed to show that the above definition satisfies the
completeness and correctness conditions from the statement of the
lemma.

\medskip
\noindent{\it Completeness. }
Let~$g'$ be the strictly increasing function obtained by induction
hypothesis for Simon level~$k$.  We set~$g(K)$ to be
\begin{equation*}
  \min\left(g'(\sqrt
K),  \frac{\sqrt{K}-2}{|\states|+1}\right)\ .
\end{equation*}
Assume now that the hint~$h$ is~$\leq g(K)$ and that the run~$\rho$
satisfies~$\val(\rho,1,\mc)\geq K$.  We want to show that~$\rho$ is
$\mc$-captured.

Let~$\rho_i\dots\rho_j$ be a minimal part of~$\rho$ for which
\begin{equation*}
  \val(\rho_i\dots\rho_j,1,\mc)\geq K \ .
\end{equation*}
The general idea is very simple. There are two possible cases: either
one of the subruns~$\rho_i,\ldots,\rho_j$ does more than~$\sqrt K$
increments on counter~$1$, or there are at least~$\sqrt K$ subruns
among~$\rho_i,\dots,\rho_j$ that increment counter~$1$. In either case
the run~$\rho$ will be~$\mc$-captured.

In the first case, we use the induction hypothesis and
\begin{equation*}
  g'(\sqrt K)\geq g(K)
\end{equation*}
to obtain item 1 in the definition of being~$\mc$-captured.

The more difficult case is the second one.  We will study the subruns
$\rho_{i+1},\dots,\rho_{j-1}$ \oldhl{that} do not reset counter~$1$.  As in
the previous lemma, we consider the set~$I \subseteq
\set{i+1,\ldots,j-1}$ of those indexes~$l$ where~$\rho_l$ does at
least one increment on counter~$1$.  By our assumption,~$I$ contains
at least~$\sqrt K - 2$ indexes (we may have lost~$2$ because of
$\rho_i$ and~$\rho_j$).

A maximal contiguous subset of~$I$ is called an \intro{inc-segment}.
\ttc{There are two possible cases. Either there are few (at most
  $|\states|$) inc-segments,
  in which case one of the}\ttc{ inc-segments must perform many increments
   and the run $\rho$ is $\mc$-captured by item
   2(b),}\oldhl{ or there are many
  inc-segments, in which case the run $\rho$ is $\mc$-captured by item
  2(a).  The details are spelled out below.}

 Consider first the case when there are at least
$|\states|+1$\ttc{ }inc-segments.  In this case, we can find two inc-segments that begin with \oldhl{indexes,
respectively,~$m' > m > 1$,} such that\oldhl{ the states~$q_{m-1}$ and $
q_{m'-1}$ are the same}. Since
inc-segments are maximal, \oldhl{the index}~$m-1$ does not belong to~$I$ and hence the
run~$\rho_{m-1}$ does not increment counter~$1$ and goes from \oldhl{state}
$q_{m-2}$ to \oldhl{state}~$q_{m-1}$. Since the \hhl{transition} graphs~$G_{m-1}$
and~$G_1$ are equivalent, there is such a run in~$G_1$, too.  Since
the run~$\rho_{m-1} \cdots \rho_{m'-1}$ goes from~$q_{m-2}$ to
$q_{m'-1}=q_{m-1}$, we obtain item 2(a) and thus~$\rho$ is
$\mc$-captured.

We are left with the case when there are at most~$|\states|+1$
inc-segments.  Recall that~$I$ contains at least~$\sqrt K -2$
elements. Since there are at most~$|\states| +1$ inc-segments, at
least one of the inc-segments must have size at least:
\begin{equation*}
  \frac{\sqrt{K}-2}{|\states|+1}\ .
\end{equation*}
But by our assumption on the hint~$h$, this inc-segment must contain
two hinted positions, and thus~$\rho$ is~$\mc$-captured by item 2(b).

\medskip
\noindent{\it Correctness. }
We set~$f(K)=K$. Assume that the hint~$h$ is~$\geq K$ and consider a
run~$\rho$ such that every run~$\pi \eqr \mc \rho$ is~$\mc$-captured.
We need to show that
\begin{equation}\label{eq:want-to-show-compl-S}
  \val(\rho,1,\mc)\geq f(K) = K \ .
\end{equation}

We proceed as in the previous lemma: we first transform the run~$\rho$
into an~$\eqr c$-equivalent run~$\pi$ which is more likely to have
fewer increments of counter~$1$. We then show that the run~$\pi$
satisfies property~(\ref{eq:want-to-show-compl-S}), which can then
also be transferred back to~$\rho$.

As in the previous lemma, there are two stages of the transformation:
a local one, and a global one.

The local transformation works just the same as the local
transformation in the previous lemma: if we can replace some subrun
$\rho_i$ by an equivalent one that is not captured, then we do so.

The global transformation makes sure that 2(a) is no longer satisfied.
It works as follows.  Assume that item 2(a) is satisfied, and let~$m$
and~$m'$ be defined accordingly.  Since~$G_1$ is idempotent and all
the~$G$'s are equivalent, the transition graph~$G_{m+1}\dots G_{m'}$
is equivalent to~$G_1$.  It follows that we can find a run with
\oldhl{neither increments nor resets} that can be plugged in place of
$\rho_{m+1}\dots\rho_{m'}$.  The new run obtained is~$\eqr
c$-equivalent to the original one and the value~$\val(\rho,1,\mc)$ is
diminished during this process.  We iterate this transformation until
no more such replacements can be applied (this obviously terminates in
fewer iterations than the number of increments of counter~$1$ in the
original run).

Consider now a run~$\pi$ obtained from~$\rho$ by first applying the
local and then the global transformation. Since~$\pi \eqr c \rho$
holds, our hypothesis says that~$\pi$ is~$\mc$-captured.  Since~$\pi$
does not satisfy item 2(a) by construction, it must satisfy either
item 1 or item 2(b). In case of item 1, we do the same reasoning as in
the previous lemma and use the induction hypothesis to conclude that
(\ref{eq:want-to-show-compl-S}) holds. The remaining case is item
2(b), when there are two distinct positions~$m<m'$ where each of the
runs~$\pi_m,\ldots,\pi_{m'}$ all increment counter~1. Since the global
transformation process only removed subruns that increment counter 1,
this means that all the runs~$\rho_m,\ldots,\rho_{m'}$ also increment
counter 1. But since the hint was~$\geq K$, we have~$m'-m\geq K$ and
we conclude with the
desired~(\ref{eq:want-to-show-compl-S}).\end{proof}

\subsection{Verifying a specification}
\label{subsection:verifying-weak-description}

In this section, we conclude the proof of the ``Moreover..''  part of
Proposition~\ref{prop:main-compl} (and therefore also the proof of
Theorem~\ref{theorem:complementation}). \oldhl{Recall that in the
  previous section, we did the proof only for atomic specifications,
  which talk about a single event. The purpose of this section is to
  generalize those results to all specifications, where positive
  boolean combinations of atomic specifications are involved.}

Given a specification~$\tau$, we want to construct a
sequence~$\bar T$-automaton~$\Aa_\tau$ that verifies if a sequence of
words~$v_1,v_2,\ldots$ satisfies:
\begin{quote}
  (*) For some~$\bar T$-function~$f$, in every word~$v_j$ every run
  satisfies~$\tau$ under~$f(j)$.
\end{quote}

We will actually prove the above result in  a slightly more general form, where
the specification $\tau$ can be a \emph{generalized specification}.
The generalization is twofold.

First, a generalized specification can describe runs in arbitrary
transition graphs, and not just those that describe runs of a counter
automaton. This is a generalization since transitions in a transition
graph contain not only the counter actions, but also some additional
labels (recall that labels in a transition graph \hhl{are} of the form
$\actions \times \lab$).

Second, in a generalized specification, atomic specifications of the
form ``the run has type $t$'' can be replaced by more powerful atomic
specifications of the form ``the run, when treated as a sequence of
labels in $\actions \times \lab$, belongs to a regular language $L
\subseteq (\actions \times \lab)^*$''.

This more general form will be convenient in the induction proof.

In our construction we will  speak quantitatively of
runs of counter automata.  Consider a sequence counter automaton (we
do not specify its acceptance condition) and a \intro{run}~$\rho$ of
this automaton over a finite word~$u$.  The only requirement on this
run is that it starts in an initial state, ends in a final state, and
is consistent with the transition function.  Given such a run~$\rho$
and a natural number~$K$, we say that it is \intro{$(\geq
  K)$-accepting} (\oldhl{respectively, }  \intro{$(\leq K)$-accepting}) if for each
counter, any two distinct resets of this counter are separated by at
least~$K$ increments of it (\oldhl{respectively, at most~$K$}).  The link with the acceptance
condition of~$B$ and~$S$-sequence automata is obvious: a run sequence
$\rho_1,\rho_2,\ldots$ is accepting for a~$B$-automaton if there
exists a natural number~$K$ such that every run~$\rho_i$ is~$(\leq
K)$-accepting. Similarly, this run sequence is accepting for an
$S$-automaton if there exists a strongly unbounded function~$f$ such
that every run~$\rho_i$ is~$(\geq f(i))$-accepting.

The heart of this section is the following lemma, which is established
by repeated use of Lemma~\ref{lemma:hinting-generic}.
\begin{lem}\label{lemma:complementation-gather}
Given a set of
states~$\states$ and a generalized specification~$\tau$ over
$\states$, there exist two strongly unbounded functions $f,g$, and a
counter automaton~$\Aa_\tau$ that reads finite~$\states$-transition
graphs, such that for any~$\states$-transition graph~$G$,
\begin{itemize}
\item (correctness) if there is a~$(\comp K)$-accepting run~$\rho$
  of~$\Aa_\tau$ over~$G$ then all runs \hhl{over}~$G$ satisfy~$\tau$ under
 $f(K)$, and;
\item (completeness) if all runs in~$G$ satisfy~$\tau$ under~$K$ then there is
  a~$(\comp g(K))$-accepting run~$\rho$ of~$\Aa_\tau$ over~$G$.
\end{itemize}
\end{lem}
In other words, acceptance by~$\Aa_\tau$ and being captured by the
specification are asymptotically the same thing.  Before establishing
this lemma, we show how it completes the proof. We only do the case of
$T=S$, the other \hhl{case} is done in a similar manner.

We want to verify if a sequence of words~$v_1,v_2,\ldots$ satisfies
the property (*).  Since a~$B$-function is essentially a constant~$K$
and~$\comp$ is~$\le$, this \hhl{boils} down to verifying that there is some
$K$ such that all runs in~$v_1,v_2,\ldots$ satisfy~$\tau$ under~$\le
K$.
We take the automaton~$\Aa_\tau$ from
Lemma~\ref{lemma:complementation-gather} and set the acceptance
condition so that all of its counters are bounded. The automaton
$\Aa_\tau$ works over transition graphs, while property (*) talks
about input words for the complemented automaton~$\Aa$, but this is
not a problem: the automaton~$\Aa_\tau$ can be modified so that it
treats an input letter as the appropriate transition graph, taken from
the transition function of~$\Aa$.  We claim that this is the desired
automaton for property (*). Indeed, if the automaton~$\Aa_\tau$
accepts a sequence~$v_1,v_2 \dots$ then its counters never exceed some
value~$K$. But then by Lemma~\ref{lemma:complementation-gather}, all
runs of~$\Aa$ in all words~$v_j$ satisfy~$\tau$ under~$\le g(K)$.
Conversely, if all runs of~$\Aa$ in all words~$v_j$ satisfy~$\tau$
under~$\le K$, then by Lemma~\ref{lemma:complementation-gather}, the
automaton~$\Aa_\tau$ has an accepting run where the counters never
exceed the value~$f(K)$.

\smallskip

\newcommand{\lift}[1]	{{#1}{\uparrow}}

\begin{proof}[Proof of Lemma~\ref{lemma:complementation-gather}]
Let us fix a generalized specification~$\tau$, for which we want to  construct the automaton $\Aa_\tau$ of Lemma~\ref{lemma:complementation-gather}.
The proof is by induction on the number of pairs $(\alpha,c)$
such that the atomic specification $\val(\rho,1,c) \comp K$ appears in $\tau$.
If there are no such pairs, satisfying $\tau$ under~$K$
does not depend anymore on~$K$ and can be checked by a standard
finite automaton, without counters.
Up to a renumbering of counters, we assume that the minimum counter appearing
in the generalized specification is~$1$, and we set~$c$ to be such that
$\val(\rho,1,c) \comp K$ appears is~$\tau$.
Our objective is to get rid of all occurrences of $\val(\rho,1,c) \comp K$
in $\tau$.

Given as input the~$\states$-transition graph~$G$, the
automaton~$\Aa_\tau$ we are constructing works as follows. To aid
reading, we present~$\Aa_\tau$ as a cascade of nondeterministic
\hhl{automata}, which is implemented using the standard Cartesian
product construction.
\begin{enumerate}
\item \label{step:decomposition} First,~$\Aa_\tau$ guesses a Simon
      decomposition of the transition graph~$G$, as well as a hint $h$
      over this decomposition (the intention is that the hint $h$ is
      $\comp K$, this will be verified in the next step). Note that
      neither the composition nor the hint need \hhl{to} be unique. The
      automaton $\Aa_\tau$ produces the relabeled transition graph
      $(G,h)$, which is used as input of the next steps.
\item \label{step:hint} In this step, the automaton~$\Aa_\tau$  checks
      that the hint  is~$\comp K$.  This step requires
  (hierarchical) counters---the automaton follows the structure of the
  decomposition and uses a counter for each level to verify that the
  number of \hhl{hinted positions is consistent with~$\comp K$.}
\item \label{step:induction} The automaton  $\Aa_\tau$ accepts the graph~$G$
 if the $\states$-transition graph $(G,h)$
 is accepted by $\Aa_{\lift\tau}$ in which $\lift\tau$ is a new generalized specification
 constructed from~$\tau$ as follows:
 \begin{itemize}[leftmargin=7mm]
\item[(a)] The regular languages in the atomic specifications are
      adapted to the new larger transition alphabet (which is $\lab_S$
      instead of $\lab$). In other words, every atomic specification
      of the kind: ``the labeling of the run belongs to a regular
      language $L$'' is replaced by: ``if the additional coordinates
      from $\lab_S$ are removed, the labeling of the run belongs to
      $L$''.
\item[(b)] Every atomic specification $\val(\rho,1,c)\comp K$ is
      replaced by an atomic specification ``the labeling of the run
      belongs to $L_c$'', with $L_c$ the regular language from
      Lemma~\ref{lemma:hinting-generic}. Note that this way we remove
      all atomic specifications that involve $\val(\rho,1,c)$, and
      therefore the induction assumption can be applied to the
      generalized specification $\lift\tau$.
 \end{itemize}
\end{enumerate}

\smallskip We now proceed to show that this automaton $\Aa_\tau$
satisfies the statement of Lemma~\ref{lemma:complementation-gather}.
This proof is by an  induction parallel to the induction
used in constructing $\Aa_\tau$.


\medskip

\emph{Completeness.}  Let~$g_1$ be the strongly unbounded function
obtained from the completeness clause in
Lemma~\ref{lemma:hinting-generic}, as applied to the event $(1,c)$
that we are eliminating.  By applying the induction assumption to the
smaller specification $\lift \tau$, but a larger set of transition
labels, we know there is some strongly unbounded function $g_2$ such
that if all runs in a graph~$(G,h)$ satisfy~$\lift\tau$ under~$K$ then
there is a~$(\comp g_2(K))$-accepting run of~$\Aa_{\lift \tau}$
over~$(G,h)$.

Let $g$ be the coordinate-wise maximum (with respect to $\comp$) of the
functions $g_1,g_2$; this function is clearly strongly unbounded. We
claim that the completeness clause of
Lemma~\ref{lemma:complementation-gather} holds for the function $g$.
Indeed, let $G$ be a graph where all runs satisfy~$\tau$ under~$K$. We
need to show \hhl{that there is  a} run of~$\Aa_\tau$ that is ($\comp g(K)$)-accepting.
The hint $h$ guessed \oldhl{in} step (2) of the construction is \oldhl{chosen so that
  every two hinted positions are separated by exactly $g_1(K)$
non-hinted positions at the same level}.  Since $g$ is greater than
$g_1$ under $\acomp$, step (2) can be done in a run that is  $(\comp
g(K))$-accepting.  Thanks to the assumption on $g_1$ and the definition of $\lift
\tau$, if a run $\rho$ \hhl{over} a transition graph $G$ satisfies $\tau$
under $K$, then the run $(\rho,h)$ satisfies the
specification~$\rho_{\lift \tau}$ under $K$.  In particular, by \oldhl{the}
assumption that every run $\rho$ \hhl{over} $G$ satisfies $\tau$ under $K$, we
can use the completeness for $\Aa_{\lift \tau}$ to infer that
$\Aa_{\lift \tau}$ has a run over $(G,h)$ that is $(\comp
g_2(K))$-accepting, which gives \oldhl{an accepting} run of
$\Aa_\tau$.

\medskip

\emph{Correctness.}
The correctness proof essentially follows the same scheme,
but requires \hhl{more} care, because of the stronger assumptions in  the correctness clause of
Lemma~\ref{lemma:hinting-generic}.
Let~$G$ be an~$\states$-transition graph and let~$K$ be a natural
number such that there is a~$(\comp K)$-accepting run of~$\Aa_\tau$
over~$G$. This means that the hint~$h$ from step~(2)
is~$\comp K$, and that $(G,h)$ is~$(\comp K)$-accepted by~$\Aa_{\lift\tau}$.

By induction hypothesis, every run in~$(G,h)$
satisfies~$\lift\tau$ under $f(K)$, for some strongly unbounded function~$f$.
In particular, for every run~$\rho$ \hhl{over}~$G$, $(\rho,h)$
satisfies~$\lift\tau$ under $f(K)$.

Let~$\rho$ be a run \hhl{over}~$G$. We will prove that~$\rho$
satisfies~$\tau$ under $f(K)$.
Two cases can happen.
\begin{itemize}
\item For all~$\pi\eqr c\rho$, the run $(\pi,h)$ belongs to $L_c$.
      Then by Lemma~\ref{lemma:hinting-generic}, we get $\val(\rho,1,c)\comp
      f(K)$.  Since the boolean combination in $\tau$ is positive,
      this means that $\rho$ satisfies $\tau$ if it satisfies the
      specification $\tau'$  obtained from $\tau$ by replacing each
      \oldhl{occurrence} of~$\val(\rho,1,c)\comp f(K)$ with {\it true}.
      However, by our assumption\oldhl{,} the run $(\rho,h)$ satisfies the
      specification $\lift \tau$ under $f(K)$, so also $\rho$
      satisfies $\tau'$ under $f(K)$.
\item Otherwise, for some~$\pi\eqr c\rho$, the run $(\pi,h)$ does not
      belong to $L_c$.  Since all runs $(\rho,h)$ in $(G,h)$
      satisfy~$\lift\tau$ under~$f(K)$, this run $\pi$ must satisfy
      the generalized specification $\tau'$ obtained from~$\tau$ by replacing
      every occurrence of~$\val(\rho,1,c)\comp K$ with {\it false}. But
      a property of the $\eqr c$-equivalence is that no atomic
      specification different from~$\val(\rho,1,c)\comp K$ can see the
      difference between two $\eqr c$-equivalent runs.  In particular,
      $\rho$ also satisfies~$\tau'$ under~$f(K)$. By consequence,
      $\rho$ satisfies~$\tau$ under~$f(K)$.
\qedhere
\end{itemize}
\end{proof}


\section{Future work}
\label{section:conclusions}
We conclude the paper with some open questions.

The first \oldhl{set of } questions concerns \oldhl{our proofs}.  In
our proof of Theorem~\ref{theorem:equivalence}, the translation from
non-hierarchical automata to hierarchical ones is very costly, in
particular it uses $\omega BS$-regular expressions as an intermediate
step. Is \oldhl{there} a better and more direct construction?
Second, our complementation proof is very complicated.  In particular,
our construction is non-elementary (it is elementary if the number of
counters is fixed).  \oldhl{It seems that a  more efficient construction
is possible since the, admittedly simpler,  but certainly
related, limitedness problem
for nested distance desert
automata is in  PSPACE}~\cite{kirsten05}.

The second set of questions concerns the model presented in this
paper. We provide here a raw list of such questions.  Are the
$\ombs$-automata (resp. $\omega B$, resp. $\omega S$) equivalent to
their deterministic form?  (We expect a negative answer.)  Is there a
natural form of deterministic automata capturing $\omega
BS$-regularity (in the same way deterministic parity automata describe
all $\omega$-regular languages)?  Are $\omega BS$-automata equivalent
to their alternating form? (We expect a positive answer, at least for
the class of $\omega B$ and $\omega S$-automata.)  Does the number of
counters induce a strict hierarchy of languages? (We expect a positive
answer.)  Similarly, does the nesting of $B$-exponents
(resp. $S$-exponents) induce a strict hierarchy of languages? (We
conjecture a positive answer.)  Is there an algebraic model for
$\omega BS$-regular languages (resp. $\omega B$, $\omega S$-regular
languages), maybe extending $\omega$-semigroups?  Other questions
concern decidability.  Is it decidable if an $\omega BS$-regular
language is $\omega$-regular (resp. $\omega B$-regular, resp. $\omega
S$-regular)? (We think that, at least, it is possible, given
an~$\omega B$ or $\omega S$-regular language, to decide whether it is
$\omega$-regular or not.) Are the hierarchies concerning the number of
counters, and the number of nesting of exponents decidable?

Other paths of research concern the possible
extensions of the model.
As we have defined them, $\omega BS$-regular languages are not closed under complementation.
Can we find a larger class that is?
What are the appropriate automata?
Such an extension would lead to the decidability
of the full logic MSOLB.
Last, but not least, is it possible to extend our results to trees?




\end{document}